\def\submission{0}
\ifnum\submission=0
\documentclass{article}
\usepackage[a4paper,top=3cm,bottom=2.5cm,left=3cm,right=3cm,marginparwidth=1.75cm]{geometry}
\usepackage{authblk}
\else
\documentclass{llncs}
\fi

\usepackage{amsmath,amsfonts,amsthm}
\usepackage{amssymb}
\usepackage{bm}
\usepackage[table,xcdraw]{xcolor}
\usepackage{xy}
\xyoption{matrix}
\xyoption{frame}
\xyoption{arrow}
\xyoption{arc}
\usepackage{ifpdf}
\entrymodifiers={!C\entrybox}
\newcommand{\qw}[1][-1]{\ar @{-} [0,#1]}
\newcommand{\qwx}[1][-1]{\ar @{-} [#1,0]}

\newcommand{\gate}[1]{*+<.6em>{#1} \POS ="i","i"+UR;"i"+UL **\dir{-};"i"+DL **\dir{-};"i"+DR **\dir{-};"i"+UR **\dir{-},"i" \qw}
\newcommand{\meter}{*=<1.8em,1.4em>{\xy ="j","j"-<.778em,.322em>;{"j"+<.778em,-.322em> \ellipse ur,_{}},"j"-<0em,.4em>;p+<.5em,.9em> **\dir{-},"j"+<2.2em,2.2em>*{},"j"-<2.2em,2.2em>*{} \endxy} \POS ="i","i"+UR;"i"+UL **\dir{-};"i"+DL **\dir{-};"i"+DR **\dir{-};"i"+UR **\dir{-},"i" \qw}

\newcommand{\control}{*!<0em,.025em>-=-<.2em>{\bullet}}

\newcommand{\ctrl}[1]{\control \qwx[#1] \qw}

\newcommand{\multigate}[2]{*+<1em,.9em>{\hphantom{#2}} \POS [0,0]="i",[0,0].[#1,0]="e",!C *{#2},"e"+UR;"e"+UL **\dir{-};"e"+DL **\dir{-};"e"+DR **\dir{-};"e"+UR **\dir{-},"i" \qw}
\newcommand{\ghost}[1]{*+<1em,.9em>{\hphantom{#1}} \qw}
\newcommand{\push}[1]{*{#1}}
\newcommand{\gategroup}[6]{\POS"#1,#2"."#3,#2"."#1,#4"."#3,#4"!C*+<#5>\frm{#6}}
\newcommand{\rstick}[1]{*!L!<-.5em,0em>=<0em>{#1}}
\newcommand{\lstick}[1]{*!R!<.5em,0em>=<0em>{#1}}
\newcommand{\ustick}[1]{*!D!<0em,-.5em>=<0em>{#1}}

\newcommand{\Qcircuit}{\xymatrix @*=<0em>}

\usepackage{graphicx}
\usepackage{multirow}
\usepackage{enumitem}
\usepackage{algorithm}

\usepackage{float}

\usepackage[mathscr]{euscript}
\usepackage{mathtools}

\ifnum\submission=0
\usepackage{fullpage}
\fi
\ifnum\submission=0
\usepackage{times} % I commented this out for submission version, is the font problem solved?
\fi
\usepackage{fancyhdr,graphicx,amsmath,amssymb}
\usepackage[algo2e]{algorithm2e}
\usepackage{tikz-cd}
\usepackage{hyperref}

\usetikzlibrary{positioning}
\usepackage{tikzpeople}

\usepackage{rotating}
\usepackage{stmaryrd}

\newcommand{\cpho}{\mathsf{CPhO}}
\newcommand{\qft}{\mathsf{QFT}}

\newcommand{\sto}{\mathsf{StO}}
\newcommand{\pho}{\mathsf{PhO}}
\newcommand{\comp}{\mathsf{Comp}}
\newcommand{\bool}{\{0,1\}}
\newcommand{\id}{\mathsf{I}}
\newcommand{\ket}[1]{|#1\rangle}
\newcommand{\bra}[1]{\langle #1|}
\newcommand{\proj}[1]{\ket{#1}\!\bra{#1}}
\newcommand{\ketbra}[2]{\ket{#1}\!\bra{#2}}

\newcommand{\norm}[1]{\left\Vert #1 \right\Vert}
\newcommand{\any}{\cdot}
\newcommand{\arrowform}[5]{\underset{\substack{{#4\underset{\substack{#3}}{\rightarrow} #5}}}{{#1}_{#2}}}
\newcommand{\hmat}[2]{\cpho_{{#1},{#2}}}
\newcommand{\harg}[4]{\arrowform{\mathsf{H}}{#1}{#2}{#3}{#4}}
\newcommand{\opnorm}[1]{\left\|#1\right\|_{\mathsf{op}}}
\newcommand\restr[2]{{% we make the whole thing an ordinary symbol
  \left.\kern-\nulldelimiterspace % automatically resize the bar with \right
  #1 % the function
  \vphantom{\big|} % pretend it's a little taller at normal size
  \right|_{#2} % this is the delimiter
  }}
\newcommand{\supp}{\mathsf{Supp}}
\newcommand{\gPoSW} {G_{n}^{\mathsf{PoSW}}}
\newcommand{\Prover}{\mathcal{P}}
\newcommand{\V}{\mathcal{V}}
\newcommand{\Hash}{{H}}

\newcommand{\ap}{\mathsf{ap}}
\newcommand{\anc}{\mathsf{anc}}
\newcommand{\rt}{\mathsf{rt}}
\newcommand{\VId}[1]{V_{#1}}
\newcommand{\lab}[1]{\ell_{#1}}
\newcommand{\mpar}{\mathsf{par}}
\newcommand{\parent}{\mathsf{in}}
\newcommand{\lch}[1]{\mathsf{left}({#1})}
\newcommand{\rch}[1]{\mathsf{right}({#1})}
\newcommand{\sib}{\mathsf{sib}}

\newcommand{\suc}{\mathsf{Suc}}
\newcommand{\fail}{\mathsf{Fail}}
\newcommand{\ext}[3]{\mathsf{Extract}^{#1}_{#2}(#3)}
\newcommand{\leaves}{\mathsf{leaves}}
\newcommand{\col}{\mathsf{CL}}
\newcommand{\nocol}{\neg\col}
\newcommand{\chain}{\mathsf{CHN}}
\newcommand{\pchain}{\mathsf{PChain}}
\newcommand{\gro}{\mathsf{ZERO_{\geq 1}}}
\newcommand{\adv}{\mathsf{Adv}}
\newcommand{\Acal}{\mathcal{A}}
\newcommand{\add}{\mathsf{ADD}}
\newcommand{\lsupp}{\mathrm{LSupp}}

\newcommand{\Lcal}{\mathcal{L}}
\newcommand{\Gcal}{\mathcal{G}}
\newcommand{\Dcal}{\mathcal{D}}
\newcommand{\Ucal}{\mathcal{U}}
\newcommand{\Jcal}{\mathcal{J}}
\newcommand{\Ical}{\mathcal{I}}
\newcommand{\Dfk}{\mathfrak{D}}
\newcommand{\Ncal}{\mathcal{N}}
\newcommand{\tcap}[3]{\left[ #1\rightarrow #2\Big| #3 \right]}
\newcommand{\rcap}[2]{\left[ #1 \Big| #2 \right]}

\newcommand{\echain}{\mathsf{ExChain}}

\newcommand{\zo}{{\{0,1\}}}

\newcommand{\Path}{{P}}

\newcommand{\CO}{{\sf cO}}
\newcommand{\bfx}{{\bf x}}
\newcommand{\bfy}{{\bf y}}
\newcommand{\bfr}{{\bf r}}
\newcommand{\bfu}{{\bf u}}
\newcommand{\lext}{\ell^\mathsf{ext}}

\def\spc{\hspace{0.05ex}}
\def\nspc{\hspace{-0.1ex}}

\newcommand{\QTC}[3]{\big\llbracket\spc#1\nspc\stackrel{#3}{\rightarrow}\nspc#2\spc\big\rrbracket}
\newcommand{\longQTC}[3]{\big\llbracket\spc#1\nspc\stackrel{#3}{\longrightarrow}\nspc#2\spc\big\rrbracket}
\newcommand{\qQTC}[4]{\big\llbracket\spc#1\stackrel{#3,#4\spc}{\Longrightarrow}\nspc#2\spc\big\rrbracket}

\newcommand{\condQTC}[4]{\big\llbracket\spc#1\nspc\stackrel{#3}{\rightarrow}\nspc#2\spc\big|\spc#4\spc\big\rrbracket}

\newcommand{\CTC}[3]{\big[\,#1\stackrel{#3}{\rightarrow}#2\,\big]}

\newcommand{\PRMG}{\text{\sf PRMG}}
\newcommand{\SIZE}[1][kq]{\mathsf{SZ}_{\leq #1}}

\renewcommand{\P}{\mathsf{P}}
\newcommand{\Q}{\mathsf{Q}}
\renewcommand{\L}{\mathsf{L}}

\newcommand{\HF}{\mathfrak{H}}  % Mnomic: Hash Function
\newcommand{\DB}{\mathfrak{D}}  % Mnomic: Database

\newcommand{\Ycal}{{\cal Y}}
\newcommand{\Xcal}{{\cal X}}
\newcommand{\CC}{{\mathbb{C}}}
\newcommand{\ZZ}{{\mathbb{Z}}}

\renewcommand{\H}{\mathscr{H}}
\newcommand{\Lin}{\mathcal{L}}

\ifnum\submission=0
\def\serge#1{{\color{red}\sf [SF: #1]}}
\def\yuhsuan#1{{\color{red}\sf [YH: #1]}}
\def\KM#1{{\color{red}\sf [KM: #1]}}
\else
\def\serge#1{}
\def\yuhsuan#1{}
\def\KM#1{}
\fi
\newtheorem{thm}{Theorem}[section] % reset theorem numbering for each chapter
\newtheorem*{lemma*}{Lemma}
\ifnum\submission=0
    \newtheorem{theorem}[thm]{Theorem}
    \newtheorem{lemma}[thm]{Lemma}
    \newtheorem{definition}[thm]{Definition}
    \newtheorem{corollary}[thm]{Corollary}
    \newtheorem{proposition}[thm]{Proposition}
    
    \theoremstyle{definition}
    \newtheorem{remark}[thm]{Remark}
\theoremstyle{definition}
    \newtheorem{example}[thm]{Example}
\fi

\title{On the Compressed-Oracle Technique, and Post-Quantum Security of Proofs of Sequential Work%
\thanks{This is the full version of an article submitted by the authors to the IACR and to Springer Verlag in March 2021. The published version is available from the proceedings of {\it Advances in Cryptology\,--\,EUROCRYPT 2021}. }
}

\begin{document}

\ifnum\submission=0

\author[1]{Kai-Min Chung}
\affil[1]{\small Academia Sinica, Taiwan ({\tt kmchung@iis.sinica.edu.tw})}

\author[2]{Serge Fehr}
\affil[2]{\small CWI Cryptology Group and Leiden University, The Netherlands ({\tt serge.fehr@cwi.nl})}

\author[3]{Yu-Hsuan Huang}
\affil[3]{\small National Chiao-Tung University, Taiwan ({\tt asd00012334.cs04@nctu.edu.tw})}

\author[4]{Tai-Ning Liao}
\affil[4]{\small National Taiwan University, Taiwan ({\tt tonyliao8631@gmail.com})}

\date{}

%    \author{%
%    \begin{tabular}{c} Kai-Min Chung \\ kmchung@iis.sinica.edu.tw \\ \\
%    Yu-Hsuan Huang\\asd00012334@hotmail.com \end{tabular} \and
%    \begin{tabular}{c} Serge Fehr \\[1ex] \small CWI Amsterdam / Leiden University \\\small The Netherlands \\ \small\tt serge.fehr@cwi.nl \\ \\
%    Tai-Ning Liao\\tonyliao8631@gmail.com\\
 %   \end{tabular} }
\else
\author{}
\institute{}
\fi
\maketitle

\ifnum\submission=1
    \renewcommand{\baselinestretch}{0.95}
    \vspace{-4ex}
\fi

\begin{abstract}
\ifnum\submission=0
We revisit the so-called {\em compressed oracle} technique, introduced by Zhandry for analyzing quantum algorithms in the quantum random oracle model (QROM). 
This technique has proven to be very powerful for reproving known lower bound results, but also for proving new results that seemed to be out of reach before. 
Despite being very useful, it is however still quite cumbersome to actually employ the compressed oracle technique. 
%In this work, we contribute to the compressed oracle technique in the following different ways.  
%, ranging from (1) {\em generalizing} the technique and (2)~{\em simplifying} its use, to (3) {\em applying} the (generalized) technique to a new specific problem and, building up on this, to prove quantum security of a cryptographic scheme that was not known to be quantum-secure yet. 

To start off with, we offer a {\em concise} yet {\em mathematically rigorous} exposition of the compressed oracle technique. We adopt a more abstract view than other descriptions found in the literature, which allows us to keep the focus on the relevant aspects. 
Our exposition easily extends to the 
%Regarding (1), we give a concise yet mathematically rigorous introduction into the  we generalize the compressed oracle technique compressed oracle technique and we extend it to the 
{\em parallel-query} QROM, where in each query-round the considered quantum oracle algorithm may make {\em several} queries to the QROM {\em in parallel}. This variant of the QROM allows for a more fine-grained query-complexity analysis of quantum oracle algorithms.
%; in particular, it provides the means to argue about the hardness of {\em parallelizing} certain computational tasks. 

Our main technical contribution is a framework that {\em simplifies} the use of (the parallel-query generalization of) the compressed oracle technique for proving query complexity results. With our framework in place, whenever applicable, it is possible to prove {\em quantum} query complexity lower bounds by means of purely {\em classical} reasoning. 
More than that, we show that, for typical examples, the crucial classical observations that give rise to the classical bounds are {\em sufficient} to conclude the corresponding quantum bounds. 

We demonstrate this on a few examples, recovering known results (like the optimality of parallel Grover), but also obtaining new results (like the optimality of parallel BHT collision search). Our main application is to prove hardness of finding a $q$-chain, i.e., a sequence $x_0,x_1,\ldots,x_q$ with the property that $x_i = H(x_{i-1})$ for all $1 \leq i \leq q$, with fewer than $q$ parallel queries. 
%This is well-know and not hard to see in the classical setting, but has resisted an analysis in the quantum setting.  

%We demonstrate this on a couple of well studied problems; nevertheless, due to our generalized treatment, we still learn some new insight here. For instance, for preimage search we learn that the naive way to parallelize such a search (by doing several independent executions of Grover in parallel) is optimal. 
%The main point here though still is simplicity: the reasoning is by means of exploiting some purely classical properties of the problem at hand and applying our meta-theorem. 

%As for (3), we consider a natural question that has so far not been (successfully) studied in the quantum setting: how many sequential queries to a random function $H:\{0,1\}^m\rightarrow\{0,1\}^m$ are necessary in order to produce a $q$-chain, i.e., a sequence $x_0,x_1,\ldots,x_q$ with the property that $x_i = H(x_{i-1})$ for all $1 \leq i \leq q$? In the classical setting, it is know (and not too hard to show) that, as one would expect, $q$ sequential queries are necessary; furthermore, parallel queries do not (substantially) help. By employing our techniques, we prove this to be true in the quantum setting as well.

The above problem of producing a hash chain is of fundamental importance in the context of {\em proofs of sequential work}. 
Indeed, as a concrete cryptographic application, we prove that the ``Simple Proofs of Sequential Work" proposed by Cohen and Pietrzak remains secure against quantum attacks. Such proof is not simply a matter of plugging in our new bound; the entire protocol needs to be analyzed in the light of a quantum attack, and substantial additional work is necessary. Thanks to our framework, this can now be done with purely classical reasoning. 
\else
We revisit the so-called compressed oracle technique, introduced by Zhandry for analyzing quantum algorithms in the quantum random oracle model (QROM). To start off with, we offer a concise exposition of the technique, which easily extends to the parallel-query QROM, where in each query-round the considered algorithm may make several queries to the QROM in parallel. This variant of the QROM allows for a more fine-grained query-complexity analysis.

Our main technical contribution is a framework that simplifies the use of (the parallel-query generalization of) the compressed oracle technique for proving query complexity results. With our framework in place, whenever applicable, it is possible to prove quantum query complexity lower bounds by means of purely classical reasoning. More than that, for typical examples the crucial classical observations that give rise to the classical bounds are sufficient to conclude the corresponding quantum bounds.

We demonstrate this on a few examples, recovering known results (like the optimality of parallel Grover), but also obtaining new results (like the optimality of parallel BHT collision search). Our main target is the hardness of finding a $q$-chain with fewer than $q$ parallel queries, i.e., a sequence $x_0, x_1,\ldots, x_q$ with $x_i = H(x_{i−1})$ for all $1 \leq i \leq q$.

The above problem of finding a hash chain is of fundamental importance in the context of proofs of sequential work. Indeed, as a concrete cryptographic application of our techniques, we prove that the “Simple Proofs of Sequential Work” proposed by Cohen and Pietrzak remains secure against quantum attacks. Such an analysis is not simply a matter of plugging in our new bound; the entire protocol needs to be analyzed in the light of a quantum attack. Thanks to our framework, this can now be done with purely classical reasoning.
\fi
\end{abstract}

\section{Introduction}

\paragraph{\bf Background. }
The random oracle methodology~\cite{BellareRo93} has proven to be a successful way to design very efficient cryptographic protocols and to argue them secure in a rigorous yet idealized manner. The considered idealization treats a cryptographic hash function $H:\{0,1\}^n\rightarrow\{0,1\}^m$ as an external {\em oracle} that the adversary needs to query on $x \in \{0,1\}^n$ in order to learn $H(x)$. Furthermore, this oracle, called {\em random oracle} (RO) then, answers these queries by means of a {\em uniformly random} function $H:\{0,1\}^n\rightarrow\{0,1\}^m$. 
Even though it is known that in principle the methodology can break down~\cite{canetti2004random} and a ``proven secure" protocol may become insecure in the actual (non-idealized) setting, experience has shown that for natural protocols this does not seem to happen. 

In case of a {\em quantum} adversary that may locally run a quantum computer, the RO needs to be modeled as a quantum operation that is capable of answering queries {\em in superposition}, in order to reasonably reflect the capabilities of an attacker in the non-idealized setting~\cite{BDFLSZ11}. This is then referred to as the {\em quantum random oracle model} (QROM). Unfortunately, this change in the model renders typical RO-security proofs invalid. One reason is that in the ordinary RO model the security reduction can inspect the queries that the adversary makes to the RO, while this is not possible anymore in the quantum setting when the queries are quantum states in superposition\,---\,at least not without disturbing the query state significantly and, typically, uncontrollably.  

\ifnum\submission=1 \par\vspace{-1ex}\noindent \fi
\paragraph{\bf The Compressed Oracle. }
A very powerful tool to deal with the QROM is the so-called {\em compressed oracle} technique, introduced by Zhandry~\cite{zha19}. On a conceptual level, the technique very much resembles the classical "lazy sampling" technique; on a technical level, the idea is to consider a {\em quantum purification} of the random choice of the function $H$, and to analyze the internal state of the RO then in the Fourier domain. 

This idea has proven to be very powerful. On the one hand, it gives rise to new and shorter proofs for known lower bound results on the query complexity of quantum algorithms (like Grover~\cite{grover1996fast,bennett1997strengths}); on the other hand, it allows for proving new cryptographic security results that seemed to be out of reach before, like in the context of {\em indifferentiability} \cite{zha19,CzajkowskiMSZ19}, or, more recently, the {\em Fiat-Shamir transformation}~\cite{LiuZ19}, when considering a quantum adversary. 
Despite being very useful, it is however still quite cumbersome to actually employ the compressed oracle technique. Proofs tend to be hard to read, and they require a good understanding of quantum information science. 

%In this work, we revisit Zhandry's compressed oracle technique, and we contribute to it in different ways as outlined in detail below. 

\ifnum\submission=1 \par\vspace{-1ex}\noindent \fi
\paragraph{\bf Our Results. }
We first present a {\em concise} yet {\em mathematically rigorous} exposition of the compressed oracle technique. 
Our exposition differs from other descriptions found in the literature (e.g. \cite{zha19,HosoyamadaI19,CzajkowskiMSZ19,ChiesaMS19,Hamoudi2020}) in that we adopt a more abstract view in terms of Fourier transform for arbitrary finite Abelian groups, i.e., by considering the range of $H$ to be an arbitrary finite Abelian group. 
Some readers may, to start with, feel uncomfortable with this approach, but it allows us to keep the focus on the relevant aspects, and, on the long run, abstraction simplifies matters and improves the understanding. 

We also consider a generalization of the compressed-oracle technique to the {\em parallel-query} QROM. In this variation of the standard QROM, the considered quantum oracle algorithm may make {\em several} queries to the QROM {\em in parallel} in each query-round. The main difference between parallel and sequential queries is of course that sequential queries may be {\em adaptive}, i.e., the queried value $x$ may depend on the hash learned in a previous query, while parallel queries are limited to be {\em non-adaptive}, i.e., the queries are independent of the hash values that are to be learned.
This variation of the QROM allows for a more fine-grained query-complexity analysis that distinguishes between the number $q$ of query rounds, and the number $k$ of queries made {\em per round}; the {\em total} number of queries made is then obviously given by $Q = kq$. This way of studying the query complexity of quantum oracle algorithms is in particular suited for analyzing how well a computational task can or cannot be parallelized (some more on this below). 

As our first main technical contribution, we propose an abstract framework that simplifies the use of (our generalized version of) the compressed oracle technique in certain cases. In particular, with our new framework in place and whenever it is applicable, it is possible to prove {\em quantum} query complexity lower bounds by means of purely {\em classical} reasoning: all the quantum aspects are abstracted away by our framework. This means that no knowledge about quantum information science is actually necessary in order to apply our framework. If applicable, the reasoning is purely by means of identifying some classical property of the problem at hand and applying our meta-theorems. More than that, the necessary classical property can typically be extracted from the\,---\,typically much simpler\,---\,proof for the classical query complexity bound. 

We demonstrate the workings and the power of our framework on a few examples, recovering known and finding new bounds. For example, with $q,k,m$ as above, we show that the success probability of finding a {\em preimage} is upper bounded by $O(k q^2/2^m)$, compared to the coarse-grained bound $O(Q^2/2^m)$ \cite{bennett1997strengths} that does not distinguish between sequential and parallel queries; this recovers the known fact that the naive way to parallelize a preimage search (by doing several executions of Grover~\cite{grover1996fast} in parallel) is optimal~\cite{Zalka99}.%
\footnote{This parallel lower bound can be improved for ``unbalanced'' algorithms for which $k$ varies from query to query; see e.g.~\cite[Lemma 2]{AHU19}.}
We also show that the success probability of finding a {\em collision} is bounded by $O(k^2 q^3/2^m)$, compared to the coarse-grained bound $O(Q^3/2^m)$ \cite{ambainis2005polynomial} that does not distinguish between sequential and parallel queries. 
Like for Grover, this shows optimality for the obvious parallelization of the BHT collision finding algorithm~\cite{brassard1997quantum}, which makes $k$-parallel queries in the first phase to collect $kq/2$ function values and then runs a parallel Grover in the second phase, which gives a factor $k^2$ improvement. We are not aware of any prior optimality result on parallel collision search; \cite{JefferyMW17} shows a corresponding bound for {\em element distinctness}, but that bound does not apply here when considering a hash function with many collisions. 
Finally, our main example application is to the problem of finding a {\em $q$-chain}, i.e., a sequence $x_0,x_1,\ldots,x_q$ with the property that $x_i = H(x_{i-1})$ for all $1 \leq i \leq q$ (or, more generally, that $H(x_{i-1})$ is a substring of $x_i$, or yet satisfies some other relation). 
While classically it is well known and not too hard to show that $q$ parallel queries are necessary to find a $q$-chain, there has been no proven bound in the quantum setting\,---\,at least not until very recently (see the recent-related-work paragraph below).%
\footnote{The problem of finding a $q$-chain looks very similar to the {\em iterated hashing} studied by Unruh in~\cite{Unruh14}; however, a crucial difference is that the start of the chain, $x_0$, can be freely chosen here. } 
Here, we show that the same does hold in the quantum setting. Formally, we prove that the success probability of finding a $q$-chain using {\em fewer} than $q$ queries is upper bounded by $O(k^3 q^3/2^m)$. 
The proof is by means of recycling an observation that is crucial to the classical proof, and plugging it into the right theorem(s) of our framework.

The problem of producing a hash chain is of fundamental importance in the context of {\em proofs of sequential work} (PoSW); indeed, a crucial ingredient of a PoSW is a computational problem that is hard/impossible to parallelize. 
Following up on this, our second main technical contribution is to show that the ``Simple Proofs of Sequential Work" proposed by Cohen and Pietrzak~\cite{cohen2018simple} remain secure against quantum attacks. 
One might hope that this is simply a matter of plugging in our bound on the chain problem; unfortunately, it is more complicated than that: the entire protocol needs to be analyzed in the light of a quantum attack, and substantial additional work is necessary to reduce the security of the protocol to the hardness of finding a chain. As a matter of fact, we enrich our framework with a ``calculus'' that facilitates the latter. 
In return, relying on our framework, the proof of the quantum security of the PoSW scheme is purely classical, with no need to understand anything about quantum information science.

\ifnum\submission=1 \par\vspace{-1ex}\noindent \fi
\paragraph{\bf Related Work. }

Independently and concurrently to the preparation of our work, the hardness of finding a $q$-chain with fewer than $q$ queries and the security of the Cohen and Pietrzak PoSW scheme~\cite{cohen2018simple} against quantum attacks have also been analyzed and tackled by Blocki, Lee and Zhou in~\cite{BlockiLZ20}.%
\footnote{An early version of [5] with weaker results (a weaker bound on the $q$-chain problem and no PoSW proof) appeared before our work, a newer version with results comparable to ours then appeared almost simultaneously to our work, and an update of the newer version that fixed certain technical issues appeared after our work. } 
Their bounds are comparable to ours, and both works are exploiting the compressed oracle idea; however, the actual derivations and the conceptual contributions are quite different. Indeed, Blocki {\em et al.}'s work is very specific to the $q$-chain problem and the PoSW scheme, and verifying the proofs requires a deep understanding of quantum information science in general and of the compressed oracle technique in particular. In contrast, in our work we provide a {\em general} framework for proving {\em quantum} query complexity bounds by means of {\em classical} reasoning. Verifying our framework also requires a deep understanding of quantum information science and of the compressed oracle,  but once our framework is place, the proofs become purely classical and thus accessible to a much broader audience. Furthermore, even though our original targets were the $q$-chain problem and the PoSW scheme, our framework provides means to tackle other quantum query complexity bounds as well, as is demonstrated with our new collision finding bound. Thus, our framework opens the door for non-quantum-experts to derive quantum query complexity bounds for their problems of merit. 

In the same spirit, Chiesa, Manohar and Spooner~\cite{ChiesaMS19} also offer means to apply the compressed oracle technique using purely classical combinatorial  reasoning. A major difference is that in our work we allow {\em parallel} queries (which is crucial for our PoSW application), which confronted us with the main technical challenges in our work. Our framework easily applies to the main application of the Chiesa {\em et al.}\ paper (post-quantum secure SNARGs), but not vice~versa.

\section{Warm-up: Proving Classical Query Complexity Lower Bounds}\label{sec:CTC}

In this section, we discuss lower bounds on the {\em classical} query complexity in the classical ROM for a few example problems. This serves as a warm-up and as a reminder of how such classical bounds are (or can be) rigorously proven. Additionally, it demonstrates that, when it then comes to analyzing the {\em quantum} query complexity of these problems, it is simply a matter of recycling certain observations from the classical proofs and plugging them into our framework.

\subsection{The Lazy-Sampling Technique}

First, let us briefly recall the {\em lazy sampling} technique, which allows us to efficiently simulate the random oracle. Instead of choosing a uniformly random function $H: {\cal X} \to {\cal Y}$ and answering each query $x$ to random oracle as $y = H(x)$, one can build up the hash function $H$ ``on the fly". Introduce a special symbol $\bot$, which stands for ``not defined (yet)", and initiate $D_0$ to be the constant-$\bot$ function. 
%Then, upon the first query $x_1$, do the following: (1) choose a uniformly random $y_1 \in {\cal Y}$, (2) {\em update} $H_0$ to $H_1$, defined as $H_1(\bar x) := H_0(\bar x)$ for $\bar x \neq x_1$ and $H_1(x_1) := y_1$, and answer the query with $y_1 = H_1(x_1)$. 
Then, inductively for $i = 1,2,\ldots$, on receiving the $i$-th query $x_i$, check if this query has been made before, i.e., if $x_i = x_j$ for some $j < i$. If this is the case then set $D_i := D_{i-1}$; else, do the following: choose a uniformly random $y_i \in {\cal Y}$ and set $D_i$ to $D_i := D_{i-1}[x_i \!\mapsto\! y_i]$, where in general $D[x \!\mapsto\! y]$ is defined by $D[x \!\mapsto\! y](x) = y$ and $D[x \!\mapsto\! y](\bar x) = D(\bar x)$ for $\bar x \neq x$.%
\footnote{We stress that we define $D[x \!\mapsto\! y]$ also for $x$ with $D(x) \neq \bot$, which then means that $D$ is {\em re}defined at point $x$; this will be useful later. }
In either case, answer the query then with $y_i = D_i(x_i)$. We refer to such a function $D_i: {\cal X} \to {\cal Y} \cup \{\bot\}$ as a {\em database}. 

As it is easy to see, the lazy-sampling only affects the ``internal workings" of the random oracle; any algorithm making queries to the standard random oracle (which samples $H$ as a random function at the beginning of time), or to the lazy-sampled variant (which builds up $D_0,D_1,\ldots$ as explained above), cannot see any difference. 

For below, it will be convenient to write $D_i$, the ``update" of $D_{i-1}$ in response to query $x_i$, as $D_i = D_{i-1}^{\circlearrowleft x_i}$. 
Note that since $D_i(x) = y_i$ is chosen in a randomized way, \smash{$D_{i-1}^{\circlearrowleft x_i}$} is a random variable, strictly speaking. 

\ifnum\submission=0
\subsection{Efficient Representation}\label{sec:EffRepClassical}
One important feature of the lazy-sampling technique is that it allows for an {\em efficient} simulation of the random oracle. Indeed, compared to a uniformly random function $H: {\cal X} \to {\cal Y}$, the databases $D_0,D_1,\ldots$ can be efficiently represented by means of an encoding function $enc$ that maps any database $D: {\cal X} \to {\cal Y} \cup \{\bot\}$ to (a suitable representation of) the list of pairs $\big(x,D(x)\big)$ for which $D(x) \neq \bot$.%
\footnote{This representation as a list of pairs somewhat justifies the terminology "database" for $D$. }
Obviously, for a bounded number of queries, the list $enc(D_i)$ remains bounded in size. Furthermore, the update $enc(D_i) \mapsto enc(D_{i+1}) = enc(D_{i}[x_i \!\mapsto\! y_i])$ can be efficiently computed (for any choice of $y_i$). 
\fi

\subsection{Proving Classical Lower Bounds}\label{sec:ClLowerBounds}

\ifnum\submission=1
One important feature of the lazy-sampling technique is that it allows for an {\em efficient} simulation of the random oracle. 
\fi
In the work here, we are more interested in the fact that the lazy sampling idea is useful for showing lower bounds on the query complexity for certain tasks. Our goal here is to show on a few examples that the well-understood classical reasoning is very close to the reasoning that our framework will admit for proving bounds in the quantum setting. In order to align the two, certain argumentation below may appear overkill given the simplicity of the classical case. 

\paragraph{\bf Finding a Preimage. }

We first consider the example of finding a preimage of the random oracle, say, without loss of generality, finding $x \in {\cal X}$ with $H(x) = 0$. Thus, let $\cal A$ be an algorithm making $q$ queries to the random oracle and outputting some $x$ at the end, with the goal of $x$ being a zero-preimage. A first simple observation is the following: if in the lazy-sampling picture after $q$ queries the built-up database $D_q : {\cal X} \to {\cal Y} \cup \{\bot\}$ does not map $\cal A$'s output $x$ to $0$, then $H(x)$ is unlikely to vanish, where $H(x)$ is understood to be obtained by making one more query to the oracle, i.e., $H(x) = D_{q+1}(x)$. 
More formally, if $p$ is the probability that $H(x) = 0$ when $\cal A$ is interacting with the standard oracle, and $p'$ is the probability that $D_q(x) = 0$ when $\cal A$ is interacting with the lazy-sampled oracle, then $p \leq p' + 1/|{\cal Y}|$. Looking ahead, this trivial observation is the classical counterpart of Corollary~\ref{cor:zha} (originally by Zhandry) that we encounter later.  

The above observation implies that it is sufficient to show that $P[\exists \, x: D_q(x) \!=\! 0]$ is small. Furthermore, writing $\PRMG := \{ D: {\cal X} \to {\cal Y} \cup \{\bot\}\, |\, \exists\, x : D(x) = 0 \}$, we can write and decompose 
$$
P[\exists \, x: D_q(x) \!=\! 0] = P[D_q \!\in\! \PRMG\,] \leq \sum_i P[D_i \!\in \! \PRMG \,| \, D_{i-1} \!\not\in \! \PRMG\,] \, .
$$
In order to align the reasoning here with our framework, which relies on the notion of a {\em quantum transition capacity}, we introduce here the {\em classical transition capacity} 
$$
\CTC{\neg \PRMG}{\PRMG}{} := \max_{D \not\in \PRMG \atop x \in {\cal X}} P[D^{\circlearrowleft x} \!\in \PRMG \,]
$$
as the maximal probability that a database $D: {\cal X} \to {\cal Y} \cup \{\bot\}$ with {\em no} zero-preimage will be turned into a database {\em with} a zero-preimage as a result of a query. 
Combining the above observations, we obtain that 
\begin{equation}\label{eq:classicalPREIMG}
p \leq q \cdot \CTC{\neg \PRMG}{\PRMG}{}  + \frac{1}{|{\cal Y}|} \, . 
\end{equation}
Looking ahead, this is the classical counterpart to Theorem~\ref{thm:QTCBound} (with $\P_s$ set to $\PRMG$), which is in terms of the (appropriately defined) {\em quantum} transition capacity $\QTC{\cdot }{\cdot}{}$. 

The reader probably already sees that $\CTC{\neg \PRMG}{\PRMG}{}  = 1/{|{\cal Y}|}$, leading to the (well-known) bound $p \leq (q+1)/{|{\cal Y}|}$. 
However, in order to better understand the general reasoning, we take a more careful look at bounding this transition capacity. For every $D \not\in \PRMG$ and $x \in {\cal X}$, we identify a ``{\em local}" property $\L^{D,x} \subseteq \Ycal$ that satisfies
$$
D[x \!\mapsto\! y] \in \PRMG \;\Longleftrightarrow\; y \in \L^{D,x} \, ;
$$
therefore, $P[D^{\circlearrowleft x} \!\in \PRMG\,] \leq P\bigl[D[x \!\mapsto\! U] \!\in\! \PRMG\bigr] = P[U \!\in\! \L^{D,x}]$ where $U$ is defined to be uniformly random in $\Ycal$.
Here, we can simply set $\L^{D,x} := \{0\}$ and thus obtain $\CTC{\neg \PRMG}{\PRMG}{} = P[U \!=\! 0] = 1/{|{\cal Y}|}$ as claimed. 

The point of explicitly introducing $\L^{D,x}$ is that our framework will offer similar connections between the {\em quantum} transition capacity $\QTC{\cdot }{\cdot}{}$ and the purely classically defined probability $P[U \!\in\! \L^{D,x}]$. Indeed, by means of the very same choice of local property $\L^{D,x}$, but then applying Theorem~\ref{thm:simple}, we obtain 
$$
\QTC{\neg\PRMG }{\PRMG}{} \leq \max_{D,x} \sqrt{10 P\bigl[U \!\in\! \L^{D,x}\bigr]} \leq \sqrt{\frac{10}{|\Ycal|}} \, . 
$$
By Theorem~\ref{thm:QTCBound}, this implies that the success probability $p$ of a {\em quantum} algorithm to find a preimage is bounded by 
$$
p \leq \Bigg(q \QTC{\neg\PRMG }{\PRMG}{} + \frac{1}{\sqrt{|\Ycal|}}\Bigg)^2 \leq \Bigg(q \sqrt{\frac{10}{|\Ycal|}} + \frac{1}{\sqrt{|\Ycal|}}\Bigg)^2 = O\biggl(\frac{q^2}{|\Ycal|}\biggr) \, ,
$$
confirming the optimality of the quadratic speed-up of Grover.

\paragraph{\bf Finding a Preimage with Parallel Queries. }

The above (classical and quantum) reasoning can be extended to the parallel query model, where with each interaction with the random oracle, a query algorithm can make $k$ queries in one go. The lazy-sampling technique then works in the obvious way, with the function update $D_i := D_{i-1}^{\circlearrowleft \bfx_i}$ now involving a query {\em vector} $\bfx_i \in \Xcal^k$. This then gives rise to \smash{$\CTC{\neg\PRMG }{\PRMG}{k}$}, and (\ref{eq:classicalPREIMG}) generalizes accordingly. For $D \not\in \PRMG$ and $\bfx = (x_1,\ldots,x_k) \in {\cal X}^k$, we then identify a {\em family} of $k$ local properties $\L_1^{D,\bfx}, \ldots,\L_k^{D,\bfx} \subseteq \Ycal$ so that  
\begin{equation}\label{eq:StrongRecPREIMG}
D[\bfx \!\mapsto\! {\bf y}] \in \PRMG \;\Longleftrightarrow\; \exists \,i: y_i \in \L_i^{D,\bfx} \, ,
\end{equation}
and therefore, by the union bound, 
$P[D^{\circlearrowleft \bfx} \!\in \PRMG\,] \leq \sum_i P[U \!\in\! \L_i^{D,\bfx}]$. Setting
$\L_1^{D,\bfx} = \ldots = \L_k^{D,\bfx} := \{0\}$, we now obtain \smash{$\CTC{\neg \PRMG}{\PRMG}{k} = k P[U \!=\! 0] = k/{|{\cal Y}|}$}, showing a factor-$k$ increase in the bound as expected. More interesting is that Theorem~\ref{thm:simple} still applies, implying that for the quantum version we have  
$$
\QTC{\neg\PRMG }{\PRMG}{k} \leq \max_{D,\bfx} \sqrt{10 \sum_i P\bigl[U \!\in\! \L_i^{D,\bfx}\bigr]} \leq \sqrt{\frac{10k}{|\Ycal|}} \, .
$$
Plugging this into Theorem~\ref{thm:QTCBound}, we then get the bound
$$
p \leq \Bigg(q \sqrt{\frac{10k}{|\Ycal|}} + \frac{1}{\sqrt{|\Ycal|}}\Bigg)^2 = O\biggl(\frac{q^2 k}{|\Ycal|}\biggr) \, ,
$$
showing optimality of running $k$ parallel executions of Grover.

\paragraph{\bf Finding a Chain (with Parallel Queries). }

Another example we want to discuss here, where we now stick to the parallel query model, is the problem of finding a $(q+1)$-chain, i.e., a sequence $x_0,x_1,\ldots,x_{q+1}$ with $H(x_{i-1}) \triangleleft x_i$, with no more than $q$ (parallel) queries. Here, $\triangleleft\,$ refers to an arbitrary relation among the elements of $\Xcal$ and $\Ycal$; typical examples are: $y \triangleleft x$ if $x = y$, or if $y$ is a prefix of $y$, or if $y$ is an arbitrary continuous substring of $x$. Below, we set $\Ycal^{\triangleleft x} := \{y \in \Ycal \,|\,y \triangleleft x\}$ and $T := \max_x |\Ycal^{\triangleleft x}|$. 

Using the same kind of reasoning as above, we can argue that 
$$
p \leq \sum_{s=1}^q \CTC{\neg \chain^s}{\chain^{s+1}}{k} + \frac{2}{|\Ycal|}  \, , 
$$
where ${\chain}^s = \{D\,|\,\exists\, x_0,x_1,\ldots,x_s \in {\cal X}: D(x_{i-1}) \triangleleft x_i \:\forall i\}$. Here, it will be useful to exploit that after $s$ (parallel) queries, $D_s \in \SIZE[ks] := \{ D\,|\,|\{x|D(x) \!\neq\! \bot\}| \leq ks\}$%
\ifnum\submission=0
   , i.e., that the {\em size} of the database $D_s$, measured as the number of $x$'s for which $D_s(x) \neq \bot$, is at most $ks$
\fi
.
Thus, the above extends to 
\begin{equation}\label{eq:classicalCHAIN}
p \leq \sum_{s=1}^q \CTC{\SIZE[k(s-1)] \backslash \chain^s}{\chain^{s+1}}{k} + \frac{2}{|\Ycal|} \, ,  
\end{equation}
with the (classical) transition capacity here given by $\max P[D^{\circlearrowleft \bfx} \!\in \chain^{s+1}\,]$, maximized over all $\bfx \in \Xcal^k$ and $D \in \SIZE[k(s-1)] \setminus \chain^s$. 
To control the considered (classical and quantum) transition capacity, for any $D$ and any $\bfx = (x_1,\ldots,x_k) \in \Xcal^k$, we introduce the following local properties $\L_i^{D,\bfx} \subseteq \Ycal$ with $i =1,\ldots,k$:  
\begin{equation}\label{eq:LocalPropForChain}
\L_i^{D,\bfx} = \bigcup_{x \in \Xcal \atop D(x)\neq \bot} \!\!\!\!\Ycal^{\triangleleft x} \cup \bigcup_{j=1}^k \Ycal^{\triangleleft x_j} \, , 
\end{equation}
so that $y_i \in \L_i^{D,\bfx}$ if $y_i \triangleleft x$ for some $x \in \Xcal$ with $D(x) \neq \bot$ or $x \in \{x_1,\ldots,x_k\}$. 
They satisfy the following condition, which is slightly weaker than (\ref{eq:StrongRecPREIMG}) used above. 
%The proof is not very hard; we provide it in Appendix~\ref{app:proofs} for completeness. 

\begin{lemma}\label{lem:classicalCHAINproperty}
$
D[\bfx \!\mapsto\! \bfr] \not\in \chain^s \,\wedge\, D[\bfx \!\mapsto\! \bfu] \in \chain^{s+1} \,\Longrightarrow\, \exists \, i : r_i \neq u_i \,\wedge\, u_i \in \L_i^{D,\bfx}. 
$
\end{lemma}

\begin{proof}
Write $D_\circ$ for $D[\bfx \!\mapsto\! \bfr]$ and $D'$ for $D[\bfx \!\mapsto\! \bfu]$. Assume that $D' \in \chain^{s+1}$, and let $\hat x_0,\hat x_1,\ldots,\hat x_{s+1} \in \Xcal$ be such a chain, i.e., so that $D'(\hat x_j) \triangleleft \hat x_{j+1}$ for $j = 0,\ldots,s$. Let $s_\circ$ be the smallest $j$ so that $D_\circ(\hat x_j) \neq D'(\hat x_j)$; if $s_\circ \geq s$ (or no such $j$ exists) then $D_\circ(\hat x_j) = D'(\hat x_j) \triangleleft \hat x_{j+1}$ for $j = 0,\ldots,s-1$, and thus $D_\circ \in \chain^{s}$ and we are done. Therefore, we may assume $s_\circ < s$. Furthermore, since $D_\circ(\bar x) = D'(\bar x)$ for $\bar x \not\in \{x_1,\ldots,x_k\}$, we must have that $\hat x_{s_\circ} = x_i$ for some $i \in \{1,\ldots,k\}$, and therefore $r_i = D_\circ(x_i) = D_\circ(\hat x_{s_\circ}) \neq D'(\hat x_{s_\circ}) = D'(x_i) = u_i$. 
Also, we have that $u_i  = D'(x_i) = D'(\hat x_{s_\circ}) \triangleleft \hat x_{s_\circ+1}$ where $\hat x_{s_\circ+1}$ is such that $D'(\hat x_{s_\circ+1}) \triangleleft \hat x_{s_\circ+2}$ and thus $\neq \bot$. The latter means that either $D(\hat x_{s_\circ+1}) \neq \bot$ or $\hat x_{s_\circ+1} \in \{x_1,\ldots,x_k\}$ (or both). In either case we have that $u_i \in \L_i^{D,\bfx}$. 
\end{proof}

Applied to $\bfr := D(\bfx)$ so that $D[\bfx \!\mapsto\! \bfr] = D$, we obtain $P[D^{\circlearrowleft \bfx} \!\in \chain^{s+1}\,] \leq \sum_i P[\,U \!\in\! \L_i^{D,\bfx}]$. Given that, for $D \in \SIZE[k(s-1)]$, the set $\{x|D(x)\!\neq\!\bot\}$ is bounded in size by $k(s-1)$, and  $|\Ycal^{\triangleleft x}|,|\Ycal^{\triangleleft x_j}| \leq T$, we can bound the relevant probability \smash{$P[U \!\in\! \L_i^{D,x}] \leq ksT/|\Ycal|$}. Hence, the considered classical transition capacity is bounded by $k^2 s T/|\Ycal|$. By (\ref{eq:classicalCHAIN}), we thus have $p = O(k^2 q^2 T/|\Ycal|)$, which is in line with the bound given by Cohen-Pietrzak~\cite{cohen2018simple}. 

Also here, our framework allows us to lift the above reasoning to the quantum setting, simply by plugging the core elements of the above reasoning for the classical case into our framework. Concretely, choosing the local properties $\L_i^{D,\bfx}$ as above whenever $D \in \SIZE[k(s-1)]$, and to be constant-false otherwise, Lemma~\ref{lem:classicalCHAINproperty} ensures that we can apply Theorem~\ref{thm:tricky} to bound the {\em quantum} transition capacity as
\begin{align*}
\QTC{\SIZE[k(s-1)] \!\backslash\! \chain^s}{\chain^{s+1}}{k} &\leq 
e \max_{\bfx,D} \sum_i \!\sqrt{10 P\bigl[U \!\in\! \L_i^{D,\bfx}\bigr]} %\ifnum\submission=1 \\ & \fi
\leq e k \sqrt{\frac{10k(q+1)T}{|\Ycal|}} ,
\end{align*}
where $e$ is Euler's number. Plugging this into Theorem~\ref{thm:QTCBound}, we then get the bound
$$
p \leq \Bigg(q e k \sqrt{\frac{10k(q+1)T}{|\Ycal|}} + \frac{q+2}{|\Ycal|}\Bigg)^2 = O\biggl(\frac{q^3 k^3 T}{|\Ycal|}\biggr) 
$$
on the success probability of a quantum oracle algorithm in finding a $(q\!+\!1)$-chain with no more than $q$ $k$-parallel queries. 
%No such bound was known before in the quantum case. 
Recall, $T$ depends on the considered relation $y \triangleleft x$; $T = 1$ if $y$ is required to be equal to $x$, or a prefix of $x$, and $T = m-n$ if $y$ and $x$ are $n$- and $m$-bit strings, respectively, and $y$ is required to be a continuous substring of $x$.

\paragraph{\bf Finding a Collision (with Parallel Queries). }

In the same spirit, for the query complexity of finding a {\em collision}, it is sufficient to control the transition capacity for $\col := \{ D \, |\, \exists\, x \neq x' : D(x) = D(x') \neq \bot \}$. Indeed, using the same kind of reasoning as above, we can argue that 
%5$$
%p \leq \sum_{s=1}^q \CTC{\neg \col}{\col}{k} + \frac{2}{|\Ycal|}  \, . 
%$$
%Here, it will be useful to exploit that after $s$ (parallel) queries, $D_s \in \SIZE[ks] := \{ D\,|\,|\{x|D(x) \!\neq\! \bot\}| \leq ks\}$, i.e., that the {\em size} of the database $D_s$, measured as the number of $x$'s for which $D_s(x) \neq \bot$, is at most $ks$.
%Thus, the above extends to  
\begin{equation*}
p \leq \sum_{s=1}^q \CTC{\SIZE[k(s-1)] \backslash \col}{\col}{k} + \frac{2}{|\Ycal|} \, ,  
\end{equation*}
with the (classical) transition capacity here given by $\max P[D^{\circlearrowleft \bfx} \!\in \col\,]$, maximized over all $D \in \SIZE[k(s-1)] \setminus \col$ and $\bfx \in \Xcal^k$. In order to analyze this transition capacity, for given $D$ and $\bfx = (x_1,\ldots,x_k) \in \Xcal^k$, we consider the following family of 1-local and 2-local properties: 
%$$
%\col_{i,j} := \{D|D(x_i)=D(x_j)\neq\bot \}
%\qquad\text{and}\qquad
%\col_i = \{D|\,\exists \,\bar x \not\in \{x_1,\ldots,x_k\}:D(x_i) = D(\bar x) \neq \bot \}
%$$
%for $i \neq j \in \{1,\ldots,k\}$, were we leave the dependency on $\bfx$ implicit. Also, we have defined the local properties as properties of $D$, rather than as properties of $\bigl(D(x_i),D(x_j)\bigr)$ and $D(x_i)$, respectively. However, this is just a matter of convention, and later in the paper we will actually switch back and forth between the two. As properties of the function values, we have
$$
\col_{i,j} = \{ (y,y) \,|\, y \in \Ycal \} \subseteq \Ycal \times \Ycal
\qquad\text{and}\qquad
\col_i = \{D(\bar x) \,|\, \bar x \not\in \{x_1,\ldots,x_k\}: D(\bar x) \neq \bot\} \subseteq \Ycal \, ,
$$
indexed by $i \neq j \in \{1,\ldots,k\}$ and $i \in \{1,\ldots,k\}$,  respectively, and where were we leave the dependency on $D$ and $\bfx$ implicit. 
Similar to (\ref{eq:StrongRecPREIMG}), here we have that for any $\bfx \in \Xcal^k$ and $D \in \SIZE[k(s-1)] \setminus \col$ 
$$
D[\bfx \!\mapsto \bfy] \in \col \;\Longleftrightarrow\; \bigl(\exists\, i \!\neq\! j:  (y_i,y_j) \in \col_{i,j}\bigr) \,\vee\, \bigl(\exists\, i : y_i \in \col_{i}\bigr) \, ,
$$
i.e., a collision can only happen for $D[\bfx \!\mapsto \bfy]$ if $y_i = y_j$ for $i \neq j$, or $y_i = D(\bar x)$ for some $i$ and some $\bar x$ outside of $\bfx$. It then follows that 
$$
\CTC{\SIZE[k(s-1)] \backslash \col}{\col}{k} = \sum_{i\neq j} P[(U,U') \!\in\!  \col_{i,j}] + \sum_{i} P[U \!\in\!  \col_{i}] \leq \frac{k(k-1)}{|\Ycal|} + \frac{k^2(s-1)}{|\Ycal|} \, ,
$$
where we exploited that $\col_i $ is bounded in size by assumption on $D$. This then amounts to the classical bound 
$$
p \leq O\biggl(\frac{q^2k^2}{|\Ycal|}\biggr) 
$$ 
on the success probability of finding a collision with no more than $q$ $k$-parallel queries. 

Here, due to the 2-locality of $\col_{i,j}$, there is an additional small complication for deriving the corresponding quantum bound, since in such a case our framework does not relate the corresponding quantum transition capacity to the probability $P[(U,U') \!\in\!  \col_{i,j}]$ of a random pair in $\Ycal \times \Ycal$ satisfying the $2$-local property $\col_{i,j}$. Instead, we have to consider the following derived $1$-local properties. 
For any $i\neq j$ and $D'$, let 
$$
\col_{i,j}|_{D'|^{x_i}} := \col_{i,j} \cap \big((\Ycal \cup \{\bot\}) \times \{D'(x_j)\} \big) =  \{D'(x_j)\} 
\qquad\text{and}\qquad
\col_{i}|_{D'|^{x_i}} := \col_{i} \, .
$$
Then, the considered quantum transition capacity is given in terms of 
$$
P\bigl[U \!\in\! \col_{i,j}|_{D'|^{x_i}}\bigr] = \frac{1}{|\Ycal|}
\qquad\text{and}\qquad
P\bigl[U \!\in\! \col_{i}|_{D'|^{x_i}}\bigr] \leq \frac{kq}{|\Ycal|} \, .
$$
Namely, by Theorem~\ref{thm:simple-general}, 
$$
\QTC{\SIZE[ks]\backslash\col}{\col}{k} 
\leq 2e \sqrt{10 \bigg(\sum_{i \neq j} P\bigl[U \!\in\! \col_{i,j}|_{D'|^{x_i}}\bigr] +  \sum_{i} P\bigl[U \!\in\! \col_{i}|_{D'|^{x_i}}\bigr]\bigg) } 
\leq 2e k \sqrt{10 \, \frac{q+1}{|\Ycal|}} \, .
$$
By Theorem~\ref{thm:QTCBound}, this then amounts to the bound
$$
p \leq O\biggl(\frac{q^3 k^2}{|\Ycal|}\biggr) 
$$ 
on the success probability of a quantum oracle algorithm in finding a collision with no more than $q$ $k$-parallel queries.

\section{Notation}

\subsection{Operators and Their Norms}\label{sec:operators}

Let $\H$ be a finite-dimensional complex Hilbert space; by default, $\H = \CC^d$ for some dimension $d$. 
We use the standard bra-ket notation for covariant and contravariant vectors in $\H$, i.e., for column and row vectors $\CC^d$. 
We write $\Lin(\H,\H')$ for the linear maps, i.e., operators (or matrices), $A: \H \to \H'$, and we use $\Lin(\H)$ as a short hand for $\Lin(\H,\H)$. We write $\id$ for the identity operator in $\Lin(\H)$. 
It is understood that pure states are given by norm-$1$ ket vectors $\ket{\psi} \in \H$ and mixed states by density operators $\rho \in \Lin(\H)$. 

A (possibly) mixed state $\rho \in \Lin(\H)$ is said to be {\em supported} by subspace $\H_\circ \subseteq \H$ if the support of the operator $\rho$ lies in $\H_\circ$, or, equivalently, if any purification $\ket{\Psi} \in \H \otimes \H$ of $\rho$ lies in $\H_\circ \otimes \H$. 
A state is said to be supported by a family of (orthonormal) vectors if it is supported by the span of these vectors. 

We write $\|A\|$ for the {\em operator norm} of $A \in \Lin(\H,\H')$ and recall that it is upper bounded by the {\em Frobenius norm}. 
Special choices of operators in $\Lin(\H)$ are {\em projections} and {\em unitaries}. 
We assume familiarity with these notions, as well as with the notion of an {\em isometry} in $\Lin(\H,\H')$.

If $\H_\circ$ is a subspace of $\H$ and $A \in \Lin(\H_\circ)$ then we can naturally understand $A$ as a map $A \in \Lin(\H)$ by letting $A$ act as zero-map on any $\ket{\psi} \in \H$ that is orthogonal to $\H_\circ$. We point out that this does not cause any ambiguity in $\|A\|$. 
Vice versa, for any $A \in \Lin(\H)$ we can consider its restriction to $\H_\circ$. Here, we have the following. 
If $\H = \H_1 \oplus \ldots \oplus \H_m$ is a decomposition of $\H$ into orthogonal subspaces $\H_i \subseteq \H$, and $A \in \Lin(\H)$ is such that its restriction to $\H_i$ is a map $\H_i \to \H_i$ and coincides with $B_i \in \Lin(\H_i)$ for any $i \in \{1,\ldots,m\}$, then 
$$
\|A\| = \max_{1 \leq i \leq m} \|B_i\| \, .
$$
This is a property we are exploiting multiple times, typically making a reference then to ``basic properties" of the operator norm.

\subsection{The Computational and the Fourier Basis}

Let $\cal Y$ be a finite Abelian group of cardinality $M$, and let $\{\ket{y}\}_{y \in \cal Y}$ be an (orthonormal) basis of $\H = \CC^M$, where the basis vectors are labeled by the elements of $\cal Y$. 
We refer to this basis as the {\em computational basis}, and we also write $\CC[{\cal Y}]$ for $\H = \CC^M$ to emphasize that the considered space is spanned by basis vectors that are labeled by the elements in~$\cal Y$. Let $\hat{\cal Y}$ be the {\em dual group} of $\cal Y$, which consists of all group homomorphisms ${\cal Y} \to \{\omega \in \CC \,|\, |\omega| = 1\}$ and is known to be isomorphic to $\cal Y$, and thus to have cardinality $M$ as well. Up to some exceptions, we consider  $\hat{\cal Y}$ to be an {\em additive} group; the neutral element is denoted~$\hat 0$. 
We stress that we treat $\cal Y$ and $\hat{\cal Y}$ as disjoint sets, even though in certain (common) cases they are {\em naturally} isomorphic and thus considered to be equal. 
The {\em Fourier basis} $\{\ket{\hat y}\}_{\hat y \in \hat{\cal Y}}$ of $\H$ is defined by the basis transformations 
\begin{equation}\label{eq:FourierBasis}
\ket{\hat y} = \frac{1}{\sqrt{M}}\sum_y \hat y(y)^*\ket{y} 
\qquad\text{and}\qquad 
\ket{y} = \frac{1}{\sqrt{M}}\sum_{\hat y} \hat y(y)\ket{\hat y} 
\, ,
\end{equation}
where $(\cdot)^*$ denotes complex conjugation.%
\footnote{By fixing an isomorphism $\Ycal \to \hat\Ycal, y \mapsto \hat y$ we obtain a unitary map $\ket{y} \mapsto \ket{\hat y}$, called {\em quantum Fourier transform (QFT)}. However, we point out that {\em in general} there is no {\em natural} choice for the isomorphism, and thus for the QFT\,---\,but in the common cases there is. We note that in this work we do not fix any such isomorphism and do not make use of a QFT; we merely consider the two bases. }
With the above convention on the notation, we have $\CC[{\cal Y}] = \CC[\hat{\cal Y}] = \H$.%
\footnote{The reader that feels uncomfortable with this abstract approach to the Fourier basis may stick to $\Ycal = \{0,1\}^m$ and replace $\ket{\hat y}$ by ${\sf H}^{\otimes m} \ket{y}$ with $y \in \{0,1\}^m$ and $\sf H$ the Hadamard matrix. }
An elementary property of the Fourier basis is that the operator in $\Lin(\CC[{\cal Y}] \otimes \CC[{\cal Y}])$ defined by $\ket{y}\ket{y'} \mapsto \ket{y\!+\!y'}\ket{y'}$ for $y,y' \in \Ycal$ acts as $\ket{\hat y}\ket{\hat y'} \mapsto \ket{y}\ket{\hat y \!-\! \hat y'}$ for $\hat y, \hat y' \in \hat\Ycal$. 

\smallskip

We will also consider extensions $\Ycal \cup \{\bot\}$ and $\hat\Ycal \cup \{\bot\}$ of the sets $\Ycal$ and $\hat\Ycal$ by including a special symbol $\bot$. We will then fix a norm-$1$ vector $\ket{\bot} \in \CC^{M+1}$ that is orthogonal to $\CC[{\cal Y}] = \CC[\hat{\cal Y}]$, given a fixed embedding of $\CC[{\cal Y}] = \CC^{M}$ into $\CC^{M+1}$. In line with our notation, $\CC^{M+1}$ is then referred to as $\CC[\Ycal \cup \{\bot\}] = \CC[\hat\Ycal \cup \{\bot\}]$. 

%\begin{remark}\label{rem:types}
%In line with the discussion in Section~\ref{sec:operators} above, for any $S \subseteq \Ycal$ the projection $\sum_{y \in S} \proj{y}$ is naturally understood as map 
%$$
%\CC[\Ycal] \to \CC[\Ycal] \subseteq \CC[\Ycal\cup\{\bot\}]
%\qquad\text{or}\qquad
%\CC[\Ycal\cup\{\bot\}] \to \CC[\Ycal] \subseteq \CC[\Ycal\cup\{\bot\}] 
%$$
%where the latter exploits that $\proj{y}\ket{\bot} = \ket{y}\braket{y}{\bot} = 0$ for any $y \in \Ycal$. 
%\end{remark}

\subsection{Functions and Their (Quantum) Representations}

For an arbitrary but fixed non-empty finite set ${\cal X}$, we let $\HF$ be the set of functions $H: {\cal X} \to {\cal Y}$. Similarly,  $\hat{\HF}$ denotes the set of all functions $\hat H: {\cal X} \to \hat{\cal Y}$. Given that we can represent $H$ by its function table $\{H(x)\}_{x \in \Xcal}$, and $\ket{y} \in \CC[\Ycal]$ is understood as a ``quantum representation" of $y \in \Ycal$, we consider $\ket{H} = \bigotimes_x \ket{H(x)}$ to be the ``quantum representation" of $H$, where in such a tensor product we implicitly consider the different registers to be {\em labeled} by $x \in \cal X$ in the obvious way. By our naming convention, the space $\bigotimes_x \CC[\Ycal]$ spanned by all vectors $\ket{H} = \bigotimes_x \ket{H(x)}$ with $H \in \HF$ is denoted $\CC[\HF]$. Similarly, $\ket{\hat H} = \bigotimes_x \ket{\hat H(x)}$ is the ``quantum representation" of $\hat H \in \hat\HF$. By applying (\ref{eq:FourierBasis}) register-wise, any $\ket{H}$ decomposes into a linear combination of vectors $\ket{\hat H}$ with $\hat H \in \hat\HF$, and vice versa. Thus, $\CC[\HF] = \CC[\hat\HF]$. 

\smallskip

Extending $\Ycal$ to $\bar\Ycal := \Ycal \cup \{\bot\}$, we also consider the set $\DB$ of functions (referred to as {\em databases}) $D: {\cal X} \to \bar{\cal Y}$. In line with the above, the ``quantum representation" of a database $D$ is given by $\ket{D} = \bigotimes_x \ket{D(x)} \in \bigotimes_x \CC[\bar\Ycal] = \CC[\DB]$. 
We also consider the set $\hat\DB$ of functions $\hat D: {\cal X} \to \hat\Ycal \cup \{\bot\}$ and have $\CC[\DB] = \CC[\hat\DB]$. 

\smallskip

For $D \in \DB$ and $\bfx = (x_1,\ldots,x_k) \in \Xcal^k$, we write $D(\bfx)$ for $\bigl(D(x_1),\ldots,D(x_k)\bigr) \in \bar\Ycal^k$; similarly for $H \in \HF$. Furthermore, if $\bfx$ has pairwise distinct entries and $\bfr = (r_1,\ldots,r_k) \in \bar\Ycal^k$, we define $D[\bfx \!\mapsto\! \bfr] \in \DB$ to be the database 
$$
D[\bfx \!\mapsto\! \bfr](x_i) = r_i
\qquad\text{and}\qquad
D[\bfx \!\mapsto\! \bfr](\bar x) = D(\bar x) \;\; \forall \: \bar x \not\in \{x_1,\ldots,x_k\} \, .
$$

\section{Zhandry's Compressed Oracle - Refurbished}  \label{sec:zhandry}

We give a concise yet self-contained \ifnum\submission=0 and mathematically rigorous \fi introduction to the compressed-oracle technique. 
\ifnum\submission=0 
For the reader familiar with the compressed oracle, we still recommend to browse over the section to familiarize with the notation we are using, and for some important observations, but some of the proofs can well be skipped then. 
\fi

\subsection{The Compressed Oracle}

The core ideas of Zhandry's compressed oracle are, first, to consider a {\em superposition} $\sum_H \ket{H}$ of all possible functions $H \in \HF$, rather than a uniformly random choice; this {\em purified} oracle is indistinguishable from the original random oracle for any (quantum) query algorithm since the queries commute with measuring the superposition. 
Second, to then analyze the behavior of this purified oracle in the {\em Fourier} basis. Indeed, the initial state of the oracle is given by 
\begin{equation}\label{eq:init}
\ket{\Pi_0} = \sum_H \ket{H} = \bigotimes_x \Bigl(\sum_y\ket{y}\Bigr) = \bigotimes_x \ket{\hat 0} = \ket{\boldsymbol \hat {\bf 0}} \in \CC[\HF] \, ,
\end{equation}
with ${\boldsymbol \hat {\bf 0}} \in \hat{\HF}$ the constant-$\hat 0$ function. Furthermore, an oracle query invokes the unitary map $\sf O$ given by 
$$
{\sf O}: \ket{x}\ket{y} \otimes \ket{H} \mapsto \ket{x}\ket{y+H(x)} \otimes \ket{H} 
$$
in the computational basis; in the Fourier basis, this becomes 
\begin{equation}\label{eq:F-Step}
{\sf O}: \ket{x}\ket{\hat y} \otimes \ket{\hat H} \mapsto \ket{x}\ket{\hat y} \otimes {\sf O}_{x \hat y} \ket{\hat H} = \ \ket{x}\ket{\hat y} \otimes \ket{\hat H - \hat y \cdot \delta_x} \, ,
\end{equation}
where the equality is the definition of ${\sf O}_{x \hat y}$, and $\delta_x: {\cal X} \rightarrow \bool$ satisfies $\delta_x(x) = 1$ and $\delta_x(x') = 0$ for all $x' \neq x$. 
Note that ${\sf O}_{x \hat y}$ acts on register $x$ only, and ${\sf O}_{x \hat y} {\sf O}_{x \hat y'} ={\sf O}_{x,\hat y + \hat y'}$; thus, ${\sf O}_{x \hat y}$ and ${\sf O}_{x' \hat y'}$ all commute. 
As an immediate consequence of (\ref{eq:init}) and (\ref{eq:F-Step}) above, it follows that the internal state of the oracle after $q$ queries is supported by state vectors of the form $\ket{\hat H} = \ket{\hat y_1\delta_{x_1} + \cdots +\hat y_q\delta_{x_q}}$. 

The actual {\em compressed} oracle (respectively some version of it) is now obtained by applying the isometry
$$
\comp_x = \ketbra{\bot}{\hat 0} + \sum_{\hat z \neq \hat 0} \proj{\hat z}: \, \CC[\Ycal] \to \CC[\bar\Ycal], \:
\ket{\hat y} \mapsto 
\left\{\begin{array}{ll}
\ket{\bot} & \text{if $\hat y = \hat 0$} \\[0.5ex]
\ket{\hat y} & \text{if $\hat y \neq \hat 0$}
\end{array}\right.
$$
to register $x$ for all $x \in \cal X$ (and then viewing the result in the computational basis). 
%Here, $\ket{\bot}$ is a state vector in $\mathbb{C}^{M+1}$ that is orthogonal to the space $\mathbb{C}^{M} \subseteq \mathbb{C}^{M+1}$ spanned by $\{\ket{\hat y}\}_z$, respectively $\{\ket{y}\}_y$. 
This ``compression" operator $\comp := \bigotimes_x \comp_x: \CC[\HF] \to \CC[\DB]$ maps $\ket{\Pi_0}$ to
$$
\ket{\Delta_0} := \comp\,\ket{\Pi_0} = \Big(\bigotimes_x \comp_x\Big) \Big(\bigotimes_x\ket{\hat 0}\Big) = \bigotimes_x \comp_x \ket{\hat 0} = \bigotimes_x \ket{\bot} = \ket{\boldsymbol \bot} \, ,
$$
which is the quantum representation of the trivial database ${\boldsymbol \bot}$ that maps any $x \in \cal X$ to $\bot$. More generally, for any $\hat H \in \hat \HF$, $\comp \, \ket{\hat H} = \ket{\hat D}$ where $\hat D \in \hat\DB$ is such that $\hat D(x) = \hat H(x)$ whenever $\hat H(x) \neq 0$, and $\hat D(x) = \bot$ whenever $\hat H(x) = 0$. 
As a consequence, the internal state of the compressed oracle after $q$ queries is supported by state vectors $\ket{D}$ in the computational basis (respectively $\ket{\hat D}$ in the Fourier basis) for which $D(x) = \bot$ (respectively $\hat D(x) = \bot$) for all but (at most) $q$ choices of $x$.

This representation of the internal state of the purified random oracle is referred to as the {\em compressed} oracle because, for a bounded number of queries, these state vectors $\ket{D}$ can be efficiently represented in terms of the number of qubits, i.e., can be {\em compressed}, as $\ket{enc(D)}$, i.e., by employing a classical efficient representation\ifnum\submission=0, similar to the one mentioned in Section~\ref{sec:EffRepClassical}\fi. Furthermore, the unitary that implements an oracle call (see $\sf cO$ below) can then be efficiently computed by a quantum circuit. In this work, we are not concerned with such computational efficiency aspect; nevertheless, for completeness, we formally discuss this in Appendix~\ref{app:Efficiency}. 

%In Section~\ref{sec:TransitionMatrix} below, we will work out the map that describes the evolution that the compressed internal state of the purified oracle undergoes as a result of applying $\sf O$ (to the uncompressed state). 

\subsection{Linking the Compressed and the Original Oracle}

The following result (originally by Zhandry~\cite{zha19}) links the compressed oracle with the original standard oracle. Intuitively, it ensures that one can extract useful information from the compressed oracle. 
%For completeness and self-containedness, we give full proofs in Appendix~\ref{app:proofs}. 
Recall that $M = |{\cal Y}|$. 

\begin{lemma} {\label{lem:zha}}
Consider an arbitrary (normalized) $\ket{\Pi} \in \CC[\HF]$ , and let $\ket{\Delta} = \comp\, \ket{\Pi}$ in $\CC[\DB]$ be the corresponding ``compressed database". Let ${\bf x} = (x_1,\ldots, x_\ell)$ consist of pairwise distinct $x_i \in {\cal X}$, let ${\bf y} = (y_1,\ldots,y_\ell) \in {\cal Y}^\ell$, and set $P_{\bf x} := 
\proj{y_1}\otimes \cdots \otimes \proj{y_\ell}$ with the understanding that $\proj{y_i}$ acts on register $x_i$. Then 
$$
\| P_{\bf x} \ket{\Pi}\| \leq \| P_{\bf x} \ket{\Delta}\| + \sqrt\frac{\ell}{M} \, .
$$
\end{lemma}

%\begin{remark}
%Note that $\| P_{\bf x} \ket{\Delta}\|$ is the square-root of the probability of observing $y_1,\ldots,y_\ell$ when measuring registers $x_1,\ldots, x_\ell$ in $\ket{\Delta}$, and correspondingly for $\| P_{\bf x} \ket{\Gamma}\|$. 
%The statement of the lemma thus ensures that for any quantum algorithm that makes oracles calls to the random oracle and outputs (classical) ${\bf x}$ and ${\bf y}$, the following holds. Conditioned on the considered output ${\bf x}$ and ${\bf y}$, the probability that ${\bf y}$ is the (component-wise) hash of ${\bf x}$ w.r.t.~the {\em standard} oracle can only be marginally larger that the probability that ${\bf y}$ is observed when measuring the {\em compressed} oracle.
%%This statement on the conditional probability then carries over to, say, the average probability over the choice of ${\bf x}$ and ${\bf y}$ {\em conditioned} on ${\bf x}$ and ${\bf y}$ satisfying some relation $R$. 
%This statement on the above (conditional) probabilities then carries over to the respective probabilities conditioned on ${\bf x}$ and ${\bf y}$ satisfying some relation, and thus also to the respective {\em joint} probabilities of ${\bf y}$ being as required {\em and} ${\bf x}$ and ${\bf y}$ satisfying the relation $R$. 
%\end{remark}

This somewhat technical statement directly translates to the following statement in terms of algorithmic language. 

\begin{corollary}[Zhandry]\label{cor:zha}
Let $R \subseteq {\cal X}^\ell \times {\cal Y}^\ell$ be a relation. 
Let $\cal A$ be an oracle quantum algorithm that outputs ${\bf x} \in \Xcal^\ell$ and ${\bf y} \in \Xcal^\ell$. Let $p$ be the probability that ${\bf y} = H({\bf x})$ and $({\bf x},{\bf y}) \in R$ when $\cal A$ has interacted with the standard random oracle, initialized with a uniformly random function $H$. Similarly, let $p'$ be the probability that ${\bf y} = D({\bf x})$ and $({\bf x},{\bf y}) \in R$ when $\cal A$ has interacted with the compressed oracle instead and $D$ is obtained by measuring its internal state (in the computational basis). Then
$$
\sqrt{p} \leq \sqrt{p'} + \sqrt\frac{\ell}{M} \, .
$$
\end{corollary}

\begin{proof}[Proof (of Corollary~\ref{cor:zha})]
Consider an execution of $\cal A$ when interacting with the purified oracle. For technical reasons, we assume that, after having measured and output ${\bf x},{\bf y}$, $\cal A$ measures its internal state in the computational basis to obtain a string~$w$, which he outputs as well. 
We first observe that
$$
p = \sum_{\bfx,\bfy,w \atop (\bfx,\bfy) \in R} q_{\bfx,\bfy,w} \, p_{\bfx,\bfy,w} 
\qquad\text{and}\qquad
p' = \sum_{\bfx,\bfy,w \atop (\bfx,\bfy) \in R} q_{\bfx,\bfy,w} \, p'_{\bfx,\bfy,w} 
$$
where $q_{\bfx,\bfy,w}$ is the probability that $\cal A$ outputs the triple $\bfx,\bfy,w$, and $p_{\bfx,\bfy,w}$ is the probability that ${\bf y} = H({\bf x})$ {\em conditioned} on the considered output of $\cal A$, and correspondingly for $p'_{\bfx,\bfy,w}$. 
More technically, using the notation from Lemma~\ref{lem:zha}, $p_{\bfx,\bfy,w} = \| P_{\bf x} \ket{\Pi}\|^2$ with $\ket{\Pi}$ the internal state of the purified oracle, {\em post-selected} on ${\bf x},{\bf y}$ and $w$. Similarly, $p'_{\bfx,\bfy,w} = \| P_{\bf x} \comp\,\ket{\Pi}\|^2$. Thus, applying Lemma~\ref{lem:zha} and squaring, we obtain
$$
p_{\bfx,\bfy,w} \leq \Big(\sqrt{p'_{\bfx,\bfy,w}} + \varepsilon\Big)^2 = p'_{\bfx,\bfy,w} + 2 \sqrt{p'_{\bfx,\bfy,w}}\,\varepsilon + \varepsilon^2 \, .
$$
Averaging with the $q_{\bfx,\bfy,w}$'s, applying Jensen's inequality, and taking square-roots, then implies the claim. 
%Applying Lemma~\ref{lem:zha} to the internal state $\ket{\Delta}$ of the purified oracle, {\em post-selected} on ${\bf x},{\bf y}$ and $w$, we obtain that $\sqrt{p_{\bfx,\bfy,w}} \leq \sqrt{p'_{\bfx,\bfy,w}} + \sqrt{\ell/M}$. Lemma~\ref{lem:convex} shows that such a relation is preserved under taking convex linear combinations, and thus claim follows. 
\end{proof}

\begin{proof}[Proof (of Lemma~\ref{lem:zha})]
We set $\comp_{\bf x} :=\bigotimes_i\comp_{x_i}$; the subscript $\bf x$ again emphasizing that $\comp_{\bf x}$ acts on the registers $x_1,\ldots,x_\ell$ only. In line with this, we write $\id_{\bar{\bf x}}$ for the identity acting on the registers $x \not\in \{x_1,\ldots,x_\ell\}$.
Then%
\footnote{In line with the discussion in Section~\ref{sec:operators}, since it maps any $\ket{\bot}$-component to~$0$, $P_\bfx$ can be understood to have domain $\CC[\Ycal]^{\otimes \ell}$ or $\CC[\bar\Ycal]^{\otimes \ell}$; the same for its range. Thus, below, in $P_{\bf x} \ket{\Pi}$ it is understood as $\CC[\Ycal]^{\otimes \ell} \to \CC[\Ycal]^{\otimes \ell} \subseteq \CC[\bar\Ycal]^{\otimes \ell}$, while in $P_{\bf x} \comp_{\bf x}\ket{\Pi}$ as $\CC[\bar\Ycal]^{\otimes \ell} \to \CC[\bar\Ycal]^{\otimes \ell}$. }
%\footnote{In line with Remark~\ref{rem:types} and the discussion in Section~\ref{sec:operators}, different occurrences of $P_\bfx$ below have different domains and ranges. }
%\footnote{We point out that by default, $\proj{y_i}$ is a map $\mathbb{C}^M \to \mathbb{C}^M$, defined on the space spanned by the vectors $\ket{y}$ for $y \in \cal Y$. However, we may also think of $\proj{y_i}$ as a map $\mathbb{C}^{M+1} \to \mathbb{C}^M$ by acting trivially on $\ket{\bot}$.}
\begin{align*}
\| P_{\bf x} \ket{\Pi}\| - \| P_{\bf x} \comp\ket{\Pi}\| 
&=  \| P_{\bf x} \ket{\Pi}\| - \| P_{\bf x} \comp_{\bf x}\ket{\Pi}\| && \text{(since the $\comp_x$'s are isometries)} \\
&\leq \|(P_{\bf x} - P_{\bf x}\comp_{\bf x})\ket{\Pi}\| && \text{(by triangle inequality)} \\
&\leq \|(P_{\bf x} - P_{\bf x}\comp_{\bf x}) \otimes \id_{\bar{\bf x}}\| && \text{(by definition of the operator norm)} \\
&= \|P_{\bf x} - P_{\bf x}\comp_{\bf x}\| && \text{(by basic property of the operator norm)}\
\end{align*}
We will work out the above operator norm. For this, recall that in the Fourier basis 
$$
P_{\bf x} =  \bigotimes_i \Bigg(\frac{1}{M}\!\sum_{\hat y \in \hat{\cal Y} \atop \hat z \in \hat{\cal Y}} \omega_{\hat z/ \hat y}(y_i) \ketbra{\hat z}{\hat y}\Bigg)
\quad\text{and}\quad
\comp_{\bf x} = \bigotimes_i \Bigg(\ketbra{\bot}{0} + \!\!\sum_{0 \neq \hat y \in \hat{\cal Y}}\!\! \proj{\hat y}\Bigg) \, ,
$$
with the understanding that in the above respective tensor products the $i$-th component acts on register $x_i$, and where the $\omega_{\hat z/ \hat y}(y_i)$ are suitable phases, i.e., norm-$1$ scalars, which will be irrelevant though.%
\footnote{For the record, switching back to multiplicative notation for the elements in the dual group $\hat{\cal Y}$, we have $\omega_{\hat z/ \hat y}(y_i) = (\hat z/\hat y)(y_i)$. } 
By multiplying the two, we get 
$$
P_{\bf x} \comp_{\bf x} = \bigotimes_i \Bigg(\frac{1}{M}\!\!\sum_{0 \neq \hat y \in \hat{\cal Y} \atop \hat z \in \hat{\cal Y}}\!\! \omega_{\hat z/ \hat y}(y_i) \ketbra{\hat z}{\hat y}\Bigg) \, .
$$
Multiplying out the respective tensor products in $P_{\bf x}$ and $P_{\bf x} \comp_{\bf x}$, and subtracting the two expressions, we obtain 
$$
P_{\bf x} - P_{\bf x} \comp_{\bf x} = \frac{1}{M^\ell} \!\sum_{\hat y_1,\ldots,\hat z_\ell \in \hat{\cal Y} \atop \exists i: \hat y_i = 0}\! \bigotimes_i \omega_{\hat z_i/ \hat y_i}(y_i) \ketbra{\hat z_i}{\hat y_i} 
= \frac{1}{M^\ell} \!\!\sum_{{\bf \hat y},{\bf \hat z} \atop \exists i: \hat y_i = 0}\!\! \omega_{\bf{\hat z}/{\bf \hat y}} \ketbra{{\bf \hat z}}{{\bf \hat y}} \, ,
$$
where the sum is over all ${\bf \hat y} = (\hat y_1,\ldots,\hat y_\ell)$ and ${\bf \hat z} = (\hat z_1,\ldots,\hat z_\ell)$ in $\hat{\cal Y}^\ell$ subject to that at least one $\hat y_i$ is $0$, and where $\omega_{\bf{\hat z}/{\bf \hat y}}$ is the phase $\omega_{\bf{\hat z}/{\bf \hat y}} := \prod_i \omega_{\hat z_i/\hat y_i}(y_i)$. 
Bounding the operator norm by the Frobenius norm, we thus obtain that
\begin{align*}
\|P_{\bf x} - P_{\bf x}& \comp_{\bf x}\|^2 \leq \sum_{{\bf \hat y},{\bf \hat z}} |\bra{{\bf \hat z}}(P_{\bf x} - P_{\bf x} \comp_{\bf x}) \ket{{\bf \hat y}}|^2 \\
&= \frac{1}{M^{2\ell}} \!\!\sum_{{\bf \hat y},{\bf \hat z} \atop \exists i: \hat y_i = 0}\!\! |\omega_{\bf \hat z/\hat y}|^2 
\leq \frac{1}{M^{2\ell}} \, \ell M^{2\ell-1} = \frac{\ell}{M}  \, ,
\end{align*}
where the inequality is a standard counting argument: there are $\ell$ choices for $i$, and for each $i$ there are $M^{\ell-1}$ choices for ${\bf \hat y} \in \hat{\cal Y}^\ell$ with $\hat y_i = 0$ (however, ${\bf \hat y}$'s with multiple zeros are counted multiple times this way). 
\end{proof}

%\begin{lemma}\label{lem:convex}\serge{To be put into the appendix}
%Let $0 \leq q_1,\ldots,q_n \leq 1$ with $\sum_i q_i = 1$. Furthermore, let $\varepsilon \geq 0$, and let $p_1,\ldots,p_n \geq 0$ and $p'_1,\ldots,p'_n \geq 0$ so that $\sqrt{p_i} \leq \sqrt{p'_i} + \varepsilon$ for all $i$. Then 
%$$
%\sqrt{\sum_i q_i p_i} \leq \sqrt{\sum_i q_i p'_i} + \varepsilon \, .
%$$
%\end{lemma}
%
%\begin{proof}
%By squaring both sides, we obtain
%$$
%p_i \leq \big(\sqrt{p'_i} + \varepsilon\big)^2 = p'_i + 2 \sqrt{p'_i}\,\varepsilon + \varepsilon^2 \, .
%$$
%Multiplying with $q_i$ and summing over $i$ results in
%$$
%\sum_i q_i p_i \leq \sum_i q_i p'_i + 2 \sum_i q_i  \sqrt{p'_i}\,\varepsilon + \varepsilon^2 \leq \sum_i q_i p'_i + 2 \sqrt{\sum_i q_i p'_i}\,\varepsilon + \varepsilon^2 =  \Bigg(\sqrt{\sum_i q_i p'_i} + \varepsilon\Bigg)^2 \, ,
%$$
%where the second inequality is Jensen's inequality. The claim is now obtained by taking square-roots again. 
%\end{proof}

\subsection{Working Out the Transition Matrix}\label{sec:TransitionMatrix}

Here, we explicitly work out the matrix (in the computational basis) that describes the evolution that the compressed oracle undergoes as a result of an oracle query. 
For this, it is necessary to extend the domain $\CC[\Ycal]$ %of $\CO_{x \hat y}$ and 
of $\comp_x$ to $\CC[\bar\Ycal]$ by declaring that %${\sf O}_{x \hat y} \ket{\bot} = \ket{\bot}$ and 
$\comp_x \ket{\bot} = \ket{\hat 0}$. 
This turns %$\CO_{x \hat y}$ and 
$\comp_x$ into a {\em unitary} on $\CC[\bar\Ycal]$, and correspondingly then for %${\sf O}$ and 
$\comp$. Formally, we are then interested in the unitary
$$
\CO := \comp \circ {\sf O} \circ \comp^\dagger \in \Lin\bigl(\CC[\Xcal] \otimes \CC[\Ycal] \otimes \CC[\DB]\bigr) \, ,
$$
which maps $\ket{x}\ket{\hat y} \otimes \ket{D}$ to $\ket{x}\ket{\hat y} \otimes \CO_{x \hat y}\ket{D}$ for any $D \in \DB$, where $\CO_{x \hat y} := \comp_x \circ {\sf O}_{x \hat y}\circ \comp_x^\dagger \in \Lin(\CC[\bar\Ycal])$ acts on the $x$-register only. 
In the form of a commuting diagram, we thus have
$$
\begin{array}{ccccl}
        & \CC[\HF]   & \xrightarrow{\;\comp\;} & \CC[\DB] \\[1ex]
{\sf O}_{x \hat y}\!\!\!\!\!\! & \Big\downarrow & & \Big\downarrow & \!\!\!\!\!\!\CO_{x \hat y}  \\[0.8ex]
        & \CC[\HF]   & \xrightarrow{\;\comp\;} & \CC[\DB]
\end{array}
$$

\begin{lemma}\label{lemma:cO-Matrix}
For any $\hat y \neq 0$, in the computational basis the unitary $\CO_{x \hat y}$ on $\mathbb{C}^{M+1}$ is represented by the matrix given in Figure~\ref{fig:Evolution}; i.e, for all $r,u \in \bar{\cal Y} := {\cal Y} \cup \{\bot\}$ it holds that $\bra{u}\CO_{x \hat y} \ket{r} = \gamma_{u,r}^{\hat y}$. 
Furthermore, $\CO_{x,\hat 0} = \id$. 
\end{lemma}

%The proof is a straightforward computation and thus moved to Appendix~\ref{app:proofs}. 

\begin{figure}[h]
$$
\begin{array}{c||c|c}
 & & \\[-1.5ex]
 \phantom{\gamma_{u,r}^{\hat y}} & \bot & r \in {\cal Y}  \\[1.5ex] \hline\hline 
 & & \\[-0.5ex]
 \bot              & \gamma_{\bot,\bot}^{\hat y} \!=\! 0\!    & \displaystyle \gamma_{\bot,r}^{\hat y} = \frac{\hat y^*(r)}{\sqrt M}  \\[3ex] \hline 
  & & \\[-1ex]
 \begin{turn}{90} 
\makebox[0.1cm]{\mbox{$u \in {\cal Y}$}}
\end{turn} & \displaystyle\frac{{\hat{y}}(u)}{\sqrt M} & 
\makebox[1cm][l]{$
%%% Actual text in the "super cell"
\gamma_{u,r}^{\hat y} = \left\{
\begin{array}{ll}
   \displaystyle \Big(1-\frac{2}{M}\Big) {\hat{y}}(u) + \frac{1}{M} & \mbox{if $u = r \in \cal Y$} \\[2ex]
   \displaystyle \frac{1-{\hat{y}}(r) - {\hat{y}}(u)}{M}  & \mbox{if $u \neq r$, both in $\cal Y$}  
\end{array}\right.
%%% End actual text in the "super cell" 
$} \\
\end{array}
%%% The following is invisible and for proper spacing only
\phantom{\displaystyle \frac{1-\hat y(r) - \hat y(u)}{M} \mbox{if $u \neq r$, both in $\cal Y$} }
%%%
$$
\vspace{-1ex}
    \caption{The matrix describing the evolution of the compressed oracle in the computational basis. }
    \label{fig:Evolution}\vspace{-2ex}
\end{figure}

\begin{proof}
From simple but somewhat tedious manipulations, using basic properties of the Fourier transform, we obtain the following. 
For any $r \neq \bot$ (and $\hat y \neq \hat 0$), we have
$$
\sqrt{M} \, \ket{r} = \sum_{\hat r} \hat{r}(r) \ket{\hat r} = \ket{0} + \sum_{\hat r \neq \hat 0} \hat{r}(r) \ket{\hat r} \, ,
$$
which gets mapped to
$$
\xmapsto{\comp^\dagger\!\!} \ket{\bot} + \sum_{\hat r \neq \hat 0} \hat{r}(r) \ket{\hat r} \, , % = \ket{\bot} - \ket{\hat 0} + \sum_{\hat r} \hat{r}(r) \ket{\hat r} = \ket{\bot} - \frac{1}{\sqrt{M}}\sum_y \ket{y} + \sqrt{m} \, \ket{r} 
$$
which gets mapped to
\begin{align*}
\xmapsto{{{\sf O}}_{x \hat y}}\; & \ket{\bot} + \sum_{\hat r \neq \hat 0} \hat{r}(r) \ket{\hat r {-} \hat y} = \ket{\bot} - \ket{{-}\hat y} + \sum_{\hat r} \hat{r}(r) \ket{\hat r {-} \hat y} \\
& = \ket{\bot} - \ket{{-}\hat y} + {\hat{y}}(r)\sum_{\hat r} \hat{r}(r) \, \ket{\hat r} = \ket{\bot} - \ket{{-}\hat y} + {\hat{y}}(r)\ket{\hat 0} + {\hat{y}}(r)\sum_{\hat r \neq \hat 0} \hat{r}(r) \, \ket{\hat r} \, ,
\end{align*}
which gets mapped to
\begin{align*}
\xmapsto{\comp}\; & \ket{\hat{0}} - \ket{{-}\hat y}  + {\hat{y}}(r)\ket{\bot}  + {\hat{y}}(r)\sum_{\hat r \neq \hat 0} \hat{r}(r) \ket{\hat r} \\
&= \ket{\hat{0}} - \ket{{-}\hat y}  + {\hat{y}}(r)\ket{\bot} - {\hat{y}}(r) \ket{\hat 0} + {\hat{y}}(r) \sum_{\hat r} \hat{r}(r) \, \ket{\hat r} \\
&= \frac{1}{\sqrt M}\sum_{u}\ket{u} - \frac{1}{\sqrt M} \sum_u {\hat{y}}(u) \ket{u} + {\hat{y}}(r)\ket{\bot} - \frac{{\hat{y}}(r)}{\sqrt M}\sum_{u}\ket{u} + \sqrt{M}{\hat{y}}(r) \ket{r} \, .
\end{align*}
From this expression, one can now easily read out the coefficients $\gamma_{u,r}^{\hat y}$ for $r \neq \bot$. Finally, from
$$
\ket{\bot} \xmapsto{\comp^\dagger\!\!} \ket{\hat 0} \xmapsto{{{\sf O}}_{x \hat y}} \ket{ {-} \hat y} \xmapsto{\comp} \ket{ {-} \hat y} = \frac{1}{\sqrt{M}}\sum_u {\hat{y}}(u) \ket{u} 
$$
we obtain the coefficients for $r = \bot$ (with $\hat y \neq 0$). The case $\hat y = 0$ follows from the fact that ${\sf O}_{x,\hat 0} = \id$. 
\end{proof}

Since, for any fixed $\hat y$, this matrix is unitary, the squares of the absolute values of each column add up to $1$. Thus, for any $\hat y$ and $r$ we can consider the (conditional) probability distribution defined by $\tilde P[U\!=\!u|r, \hat y] := |\gamma_{u,r}^{\hat y}|^2$. This offers us a convenient notation, like $\tilde P[U\!\in \!{\cal S}|r, \hat y]$ for $\sum_{u \in \cal S}|\gamma_{u,r}^{\hat y}|^2$ or $\tilde P[U\!\neq \!r|r, \hat y]$ for $\sum_{u \neq r}|\gamma_{u,r}^{\hat y}|^2$. For later purposes, it is useful to observe that, for any $\L \subseteq \Ycal$ (i.e., $\bot \not\in \L$), 
\begin{align}\label{eq:connection}
\begin{split}
%\sum_{r \not\in \L} \tilde P[U \!\in\! \L |r,\hat y] \leq
\sum_{r} \tilde P[r \!\neq\! U \!\in\! \L &|r,\hat y] \leq \tilde P[U \!\in\! \L |\bot,\hat y] + \sum_{r \neq \bot} \tilde P[r \!\neq\! U  \!\in\! \L|r,\hat y] % \\  & 
\leq |\L|\frac{1}{M} + M|\L|\frac{9}{M^2} = 10 P[U \!\in\! \L] 
\end{split}
\end{align}
where $P[U \!\in\! \L] = \frac{|\L|}{M}$ is the probability for a uniformly random $U$ in $\Ycal$ to be in $\L$. 

%\begin{equation}
%P[U\!=\!u|\bot] = \left\{
%\begin{array}{cl}
%0 & \text{if $u = \bot$} \\[1ex]
%\displaystyle\frac{1}{M} & \text{if $u \neq \bot$}
%\end{array}\right.
%\qquad\text{and}\qquad
%P[r \!\neq\! U\!\neq \! \bot|r] \leq \frac{9}{M} \, ,
%\end{equation}
%where we omit $\hat y$ given that the respective right hand sides do not depend. 

\subsection{The Parallel-Query (Compressed) Oracle}

Here, we extend the above compressed-oracle technique to the setting where a quantum algorithm may make {\em several} queries to the random oracle {\em in parallel}. We recall that distinguishing between {\em parallel} and {\em sequential} queries allows for a more fine-grained query-complexity analysis of quantum algorithms. In particular, by showing a lower bound on the number of necessary {\em sequential} queries (with each sequential query possibly consisting of a large number of {\em parallel} queries), one can show  the impossibility (or bound the possibility) of {\em parallelizing} computational tasks. 

Formally, for any positive integer $k$, a {\em $k$-parallel query} is given by $k$ parallel applications of $\sf O$, with the understanding that each application acts on a different input/output register pair. More explicitly, but slightly abusing notation of writing a $k$-th power, a $k$-parallel query is given by
$$
{\sf O}^k: \ket{{\bf x}}\ket{{\bf y}} \otimes \ket{H} \mapsto \ket{{\bf x}}\ket{{\bf y}\!+\!H({\bf x})} \otimes \ket{H} 
$$
for any ${\bf x} = (x_1,\ldots,x_k) \in {\cal X}^k$ and ${\bf y} = (y_1,\ldots,y_k) \in {\cal Y}^k$. 
The operator $\CO^k := \comp \circ {\sf O}^k \circ \comp^\dagger$, which described the evolution of the compressed oracle under such a $k$-parallel query, 
then acts as
$$
\CO^k : \ket{{\bf x}}\ket{{\bf \hat y}} \otimes \ket{\Delta} \mapsto \ket{{\bf x}}\ket{{\bf \hat y}} \otimes \CO_{{\bf x}{\bf \hat y}} \ket{\Delta}  
$$
for any $\ket{\Delta} \in \CC[\DB]$, where $\CO_{\bfx \hat\bfy}$ is the product  $\CO_{x_1 \hat y_1} \cdots \CO_{x_k \hat y_k}$.
We recall that $\CO_{x_i \hat y_i}$ acts on register $x_i$ (only), and $\CO_{x_i \hat y_i}$ and $\CO_{x_j \hat y_j}$ commute (irrespectively of $x_i$ and $x_j$ being different or not).

\section{A Framework for Proving Quantum Query Lower Bounds}\label{sec:QTC_framework}

In this section, we set up a framework for proving lower-bounds on the query complexity (actually, equivalently, upper bounds on the success probability) of {\em quantum} algorithms in the quantum random oracle model. Our framework closely mimics the reasoning for classical algorithms and allows to easily ``lift" the typical kind of reasoning to the quantum setting. 
%As a matter of fact, our framework provides the means to prove lower-bound result on the query complexity of {\em quantum} algorithms by purely {\em classical} means. 
%We showcase this on different examples, recovering known but also proving new results. 

\subsection{Setting Up the Framework}

\begin{definition}\label{def:DBProp}
A {\em database property} on $\DB$ is a subset $\P \subseteq \DB$ of the set of databases $D$. 
%We write $D \in \P$ and $\P(D)$ interchangably to express that $D$ satisfies property $\P$. 
\end{definition}

\begin{remark}%
As the naming suggests, we think of $\P$ as a property that is either {\em true} or {\em false} for any $D \in \DB$; we thus also write $\P(D)$ to denote that $D \in \P$, i.e., to express that ``$D$ satisfies $\P$". Furthermore, by convention, for any database property $\P \in \DB$, we overload notation and use $\P$ also to refer to the projection $\sum_{D \in \P} \proj{D} \in \Lin(\CC[\DB])$. 
\end{remark}
%By default, we assume a database property $\P$ to be {\em non-trivial}: the trivial function that maps every $x \in \cal X$ to $\bot$ is not in $\P$ \serge{Collision with later use of "trivial".}, and to be {\em monotone}: \serge{Not sure we need the latter.}
%$$
%D \in \P \:\wedge\: D(x) = \bot \;\Longrightarrow\; D[x \!\mapsto\! r] \in \P 
%$$
%for any $r \in \cal Y$, where $D[x \!\mapsto\! r]$ is defined as $D[x \!\mapsto\! r](x) = r$ and $D[x \!\mapsto\! r](x') = D(x')$ for all $x' \neq x$.%
%\footnote{We may also write $D = D[x \!\mapsto\! r]$ to emphasize that the considered function $D$ is so that $D(x) = r$. }

Examples that we will later consider are
$$
\PRMG := \{ D\, | \exists\, x : D(x) = 0 \} 
\quad\text{and}\quad
\col := \{D \,|\, \exists\, x,x':D(x) = D(x') \neq \bot\} \, ,
$$
as well as
$$
{\chain}^q := \{D\,|\,\exists\, x_0,x_1,\ldots,x_q \in {\cal X}: D(x_{i-1}) \triangleleft x_i  \:\forall i\} \, ,
$$
where $\triangleleft$ denotes an arbitrary relation, e.g., $y \triangleleft x$ if $y$ is a prefix of $x$. 

We introduce the following notation. 
For any tuple $\bfx = (x_1,\ldots,x_k)$ of pairwise distinct $x_i \in \cal X$ and for any $D: {\cal X} \to \bar{\cal Y}$ we let
$$
D|^\bfx := \big\{D[\bfx \!\mapsto \bfr]\,|\,\bfr \in \bar\Ycal^k\big\} \subseteq \DB 
$$
be the set of databases that coincide with $D$ outside of $\bfx$. Furthermore, for any database property $\P \subseteq \DB$, we then let
$$
\P|_{D|^\bfx} := \P \cap D|^\bfx
$$
be the restriction of $\P$ to the databases in $D|^\bfx$. 
We then typically think of $\P|_{D|^\bfx}$ as a property of functions $D' \in D|^\bfx$. 

%Consider now an arbitrary but fixed tuple $\bfx = (x_1,\ldots,x_k)$ of pairwise distinct $x_i \in \cal X$. For any $D: {\cal X} \to \bar{\cal Y}$, we set

%We also introduce the following variation (respectively generalization) of the above definition, which will become relevant soon. Consider an arbitrary but fixed tuple $\bfx = (x_1,\ldots,x_k)$ of pairwise distinct $x_i \in \cal X$. Furthermore, consider a function $D$, which we understand to be {\em fixed} (to $D(x)$) for $x \neq x_i$, while the values $D(x_i)$ are still subject to different choices. In other words, we consider the {\em family} of functions given by $D[\bfx \!\mapsto\! \bfr]$ with $\bfr$ arbitrary in $\Gcal\subseteq\bar{\cal Y}^k$. 

%\begin{definition}\label{def:rDBProp}
%A {\em database property} (w.r.t.~$\bfx$ and $D$) is a subset $\Gcal\subseteq\bar{\cal Y}^k$ with the understanding that $D[\bfx \!\mapsto\! \bfr]$ satisfies property $\Gcal$ if and only if $\bfr \in \Gcal$. 
%Also here, we write $\bfr \in \Gcal$ and $\Gcal(\bfr)$ interchangably. 
%\end{definition} 

%Obviously, if $\Gcal$ is a database property as in Definition~\ref{def:DBProp}, and if $\bfx$ and $D$ are as above, then 
%$$
%\Gcal|_{D,\bfx} := \{ \bfr | D[\bfx \!\mapsto\! \bfr] \in \Gcal\} \subseteq \bar{\cal Y}^k \, .
%$$

\begin{remark}\label{rem:convention}
For fixed choices of $\bfx$ and $D$, we can, and often will, identify $D|^\bfx$ with $\bar\Ycal^k$ by means of the obvious identification map $\bfr \mapsto D[\bfx \!\mapsto \bfr]$. The property $\P|_{D|^\bfx}$ can then be considered to be a property/subset of $\bar\Ycal^k$, namely $\{\bfr \in \bar\Ycal^k \,|\, D[\bfx \!\mapsto \bfr] \in \P\}$. Accordingly, we do not distinguish between the projections 
$$
\sum_{D' \in  \P|_{D|^\bfx}} \proj{D'} \in \Lin(\CC[D|^\bfx]) \subseteq \Lin(\CC[\DB])
\quad\text{and}\quad
\sum_{\bfr \in \bar\Ycal^k \atop D[\bfx \mapsto \bfr] \in \P} \proj{\bfr} \in \Lin(\CC[\bar\Ycal^k])
$$
but refer to both as $\P|_{D|^\bfx}$, using our convention to use the same variable for a property and the corresponding projection. This is justified by the fact that on the space spanned by $\ket{D[\bfx \!\mapsto \bfr]}$ with $\bfr \in \bar\Ycal^k$, both act identically (with the understanding that the latter acts on the registers labeled by $\bfx$). In particular, they have the same operator norm. 
\end{remark}

%\begin{remark}
%Representing $D: {\cal X} \to \bar{\cal Y}$ by its function table, which lists $D(x)$ for all $x \in \cal X$ (given a fixed ordering on~$\cal X$), we can also understand a database property $\Gcal$ on $\DB$ (as defined in Definition~\ref{def:DBProp}) as a subset of $\bar{\cal Y}^{|{\cal X}|}$. 
%As such, $\Gcal|_{D,\bfx}$ is then a database property on $\DB|_\bfx$, where $\DB|_\bfx$ is the set of functions  $\{x_1,\ldots,x_k\} \to \bar{\cal Y}$. In the following, we will though stick to the convention that we think of a database property on $\DB$ as a property of an abstract {\em function} $D: {\cal X} \to {\cal Y}$, while we think of a database property on $\DB|_\bfx$ as a property of a {\em vector} $\bfr$, i.e., of a function that is explicitly represented by its function table. However, we emphasize that these are simply two different views of the same concept. 
%\end{remark}

%is a database property (w.r.t.~$\bfx$ and $D$) as in Definition~\ref{def:rDBProp}. Furthermore, choosing $\bfx$ to be the list of {\em all} elements in $\cal X$, the two definitions coincide conceptually, justifying the clash in terminology. 

\begin{example}
For a given $\bfx$ and $D$, as a subset of $\bar{\cal Y}^k$, we have 
$$
\PRMG|_{D|^\bfx} = \left\{\begin{array}{ll}
\bar{\cal Y}^k & \text{if $D(\bar x) = 0$ for some $\bar x \not\in\{x_1,\ldots,x_k\}$} \\
\{\bfr\,|\,\exists\, i: r_i=0\} & \text{else}
\end{array}\right.
$$ 
In words: if $D$ has a zero outside of $\bfx$ then $D[\bfx \!\mapsto \bfr]$ has a zero for any $\bfr \in \bar{\cal Y}^k$; otherwise, $D[\bfx \!\mapsto \bfr]$ has a zero if and only if one of the coordinates of $\bfr$ is zero. 
\end{example}

The following definition is the first main ingredient of our framework. Theorem~\ref{thm:QTCBound} below, which relates the success probability of a quantum algorithm to the quantum transition capacity, then forms the second main ingredient. 

\begin{definition}[Quantum transition capacity]\label{def:QTC}
Let $\P, \P'$ be two database properties. Then, the {\em quantum transition capacity} (of order $k$) is defined as
$$
\QTC{\P}{\P'}{k} := \max_{\bfx,\hat\bfy,D} \|\P'|_{D|^\bfx} \,\CO_{\bfx\hat\bfy} \, \P|_{D|^\bfx}\| \, .
$$
Furthermore, we define
$$
\qQTC{\P}{\P'}{k}{q} := \sup_{U_1,\ldots,U_{q-1}} \| \P' \CO^k \,U_{q-1} \, \CO \cdots \CO^k \, U_1 \, \CO^k \, \P \|  \, .
$$
where the supremum is over all positive $d \in \ZZ$ and all unitaries $U_1,\ldots,U_{q-1}$ acting on $\CC[\Xcal] \otimes \CC[\Ycal] \otimes \CC^d$.
%where the supremum is over all unitaries $U_2,\ldots,U_q$ acting on $\CC[\Xcal] \otimes \CC[\Ycal] \otimes \CC^d$ with $d$ arbitrary. 
%More generally, 
%$$
%\qQTC{\neg\P}{\P'}{k}{q} := \min_{\P_0,\ldots, \P_q \atop \P_0=\P, \P_q=\P' } \sum_{s=1}^q \QTC{\neg\P_{s-1}}{\P_{s}}{k}  \, .
%$$
\end{definition}

By definition, the notion $\qQTC{\P}{\P'}{k}{q}$ equals the square-root of the maximal probability that the internal state of the compressed oracle, when supported by databases $D \in \P$, turns into a database $D' \in \P'$ by means of a quantum query algorithm that performs $q$ $k$-parallel queries, and when we then measure the internal state to obtain $D'$. In particular, for $p'$ as in Corollary~\ref{cor:zha} and $\P^R$ as below in Theorem~\ref{thm:QTCBound}, it holds that \smash{$\qQTC{{\boldsymbol \bot}}{\P^R}{k}{q} = \sqrt{p'}$}. 

In a similar manner, $\QTC{\P}{\P'}{k}$ represents a measure of how likely it is that, as a result of {\em one} $k$-parallel query, a database $D \in \DB$ that satisfies $\P$ turns into a database $D'$ that satisfies $\P'$. 
In the context of these two notations, ${\boldsymbol \bot}$ is understood to be the database property that is satisfied by ${\boldsymbol \bot} \in \DB$ only, and $\neg\P$ is the complement of $\P$, i.e., $\neg\P = \id - \P$ (as projections). 
%The intuition behind the notation is that $\QTC{\neg\P}{\P'}{}$ represents a measure of how likely it is that, as a result of a query (or several queries), a compressed database $D \in \DB$ that does {\em not} satisfy $\P$ turns into a database $D'$ that satisfies $\P'$. We also use natural variations of this notation, like $\QTC{{\boldsymbol \bot}}{\P'}{}$, which captures how likely it is that the initial all-$\bot$ database turns into a database that satisfies $\P'$, or $\QTC{{\sf Q} \backslash \P}{\P'}{}$, which captures how likely it is that the database that satisfies $\sf Q$ but not $\P$ turns into a database that satisfies $\P'$. 
We also write $\P \to \P'$ and refer to this as a {\em database transition} when considering two database properties $\P$ and $\P'$ in the context of the above two notions. Formally, they are related as follows. 

\begin{lemma}\label{lem:tech}
For any sequence of database properties $\P_0,\P_1,\ldots,\P_q$, 
$$
\qQTC{\neg\P_0}{\P_q}{k}{q} \leq \sum_{s=1}^q \QTC{\neg\P_{s-1}}{\P_{s}}{k} \, .
$$ 
\end{lemma}

\begin{proof}
By means of inserting $\id = \P_q + (\id - \P_q)$ before $U_{q-1}$ and using properties of the norm and the fact that $U_{q-1}$ and $\P_{q-1}$ commute (as they act on different subsystems) , we obtain
\begin{align*}
    \| \P_q \, \CO^k \,U_{q-1} \, \CO^k \cdots \CO^k \, (\id - \P_0) \| \leq \| \P_{q-1} \, \CO^k \,U_{q-2} \, \CO^k \cdots \CO^k \, (\id - \P_0) \| + \| \P_q \, \CO^k \, (\id - \P_{q-1}) \| \, .
\end{align*}
To the first term, we apply induction; so it remains to bound the second term by $\QTC{\neg\P_{q-1}}{\P_q}{k}$. Setting $\P = \neg \P_{q-1}$ and $\P' = \P_q$, this follows from%
\footnote{In line with Remark~\ref{rem:convention}, we consider $\P|_{D|^\bfx}$ to be a projection acting on $\CC[\bar{\Ycal}^k]$, and thus $\id$ in the last term is the identity in $\Lin(\CC[\bar\Ycal^k])$. }
$$
\|\P' \CO^k\, \P \| \leq \max_{\bfx,\hat\bfy} \|\P' \CO_{\bfx\hat\bfy}\, \P \| \leq \max_{\bfx,\hat\bfy,D} \|\P'|_{D|^\bfx} \,\CO_{\bfx\bfy} \, (\id - \P|_{D|^\bfx})\|    \, ,
$$
where for the first inequality we observe that $\P' \CO^k \P$ maps $\ket{\bfx}\ket{\hat\bfy} \otimes \ket{\Gamma}$ to $\ket{\bfx}\ket{\hat\bfy}\otimes \P' \CO_{\bfx\hat\bfy} \P\ket{\Gamma}$, and so the first inequality holds by basic properties of the operator norm. 
Similarly for the second inequality: For any fixed $D$, consider the subspace of $\CC[\DB]$ spanned by $\ket{D[\bfx\!\mapsto\!\bfr]}$ with $\bfr \in \bar{\cal Y}^k$. On this subspace, $\P$ and $\P|_{D|^\bfx}$ are identical projections (and similarly for $\P'$). Also, $\CO_{\bfx\bfy}$ is a unitary on this subspace. The claim then again follows again by basic properties of the operator norm. 
\end{proof}

The following is now an direct consequence of Corollary~\ref{cor:zha}, the definition of $\qQTC{{\boldsymbol \bot}}{\P^R}{k}{q}$, and the above lemma. 

\begin{theorem}\label{thm:QTCBound}
Let $R$ be a relation, and let $\cal A$ be a $k$-parallel $q$-query quantum oracle algorithm outputting ${\bf x} \in \Xcal^\ell$ and ${\bf y} \in \Ycal^\ell$ and with success probability $p$, as considered in Corollary~\ref{cor:zha}. Consider the database property
$$
\P^R = \big\{D \in \DB \,|\, \exists \, \bfx \in {\cal X}^\ell : \big(\bfx,D(\bfx)\big) \in R\big\} 
$$
induced by $R$. 
%Assume that there exists a tower $\Gcal^R \cap \SIZE = \Gcal_q \subseteq \Gcal_{q-1} \subseteq \ldots \subseteq \Gcal_0$ of non-trivial database properties. 
%Then, the probability $p$ that $\cal A$'s output $\bfx,\bfy$ satisfies $\bfy = H(\bfx)$ and $(\bfx,\bfy) \in R$, is bounded by
%$$
%\sqrt{p} \leq \sum_{i=1}^q \QTC{\neg\Gcal_{i-1}}{\Gcal_{i}}{k} + \sqrt{\frac{\ell}{M}} \, .
%$$
%In particular, if $\QTC{\neg\Gcal_{i-1}}{\Gcal_{i}}{k}$ is %uniformly bounded by $\delta$ then $p \leq q \delta + %\sqrt{\ell/M}$. 
%\serge{Or, in short, using the above notation:}  
Then, for any database properties $\P_0,\ldots, \P_q$ with $\P_0=\neg {\boldsymbol \bot}$ and $\P_q=\P^R$:
$$
\sqrt{p} \,\leq\, \qQTC{{\boldsymbol \bot}}{\P^R}{k}{q} + \sqrt{\frac{\ell}{M}}
\,\leq\, %\min_{\P_0,\ldots, \P_q \atop \P_0=\neg {\boldsymbol \bot}, \P_q=\P^R } 
\sum_{s=1}^q \QTC{\neg\P_{s-1}}{\P_{s}}{k} + \sqrt{\frac{\ell}{M}} \: .
$$
%I.e., in short, $\sqrt{p} \leq \qQTC{{\boldsymbol \bot}}{\P^R}{k}{q} + \sqrt{\frac{\ell}{M}}$. 
\end{theorem}

\begin{remark}\label{rem:size_condition}
This result implies that in order to bound $p$, it is sufficient to find a sequence ${\boldsymbol \bot} \not\in \P_0,\ldots,\P_q = \P^R$ of properties for which all quantum transition capacities $\QTC{\neg\P_{s-1}}{\P_{s}}{}$ are small. Often, it is good to keep track of the (growing but bounded) size of the database and instead bound the capacities 
$$
\QTC{\SIZE[k(s-1)] \backslash\P_{s-1}}{\P_{s}}{} = \QTC{\SIZE[k(s-1)] \backslash\P_{s-1}}{\P_{s} \cup \neg\SIZE[ks]}{} \, ,
$$ 
where the equality is due to the fact that the size of a database cannot grow by more than $k$ with one $k$-parallel query. Formally, we would then consider the database properties $\P'_s = \neg(\SIZE[ks] \setminus\P_s)= \P_{s} \cup \neg\SIZE[ks]$. % to get an upper bound. 
\end{remark}

%\begin{proof}[Proof (of Theorem~\ref{thm:QTCBound})]
%Let $\P_0,\ldots,\P_q$ be as considered. 
%Applying Lemma~\ref{lem:tech1} to $\ket{\phi} = \ket{\phi_q}$, the state produced by $\cal A$ after $q$ queries, and to $\P' = \P_q$ and $\P = \P_{q-1}$, and applying induction (and using, for the base case, that $\|\P_0 \ket{\phi_0}\| = 0$), we obtain that
%$$
%\|\P_q \ket{\phi_q}\| \leq \sum_{i=1}^q \QTC{\neg\P_{i-1}}{\P_{i}}{k} \, . %+ \sqrt{\frac{\ell}{M}} \, .
%$$
%Further note that $\|\P_q \ket{\phi_q}\| = \|\P^R \ket{\phi_q}\|$, which  equals the square-root of the probability that there exists $\bfx$ so that $(\bfx,D(\bfx)) \in R$, with $D$ obtained as in Corollary~\ref{cor:zha}. Thus, the same upper bound applies to the probability that $\bfx$ output by $\cal A$ satisfies $(\bfx,D(\bfx)) \in R$. The claim then follows from Corollary~\ref{cor:zha}. 
%\ifnum\submission=1 \qed \fi
%\end{proof}

We note that Theorem~\ref{thm:QTCBound} assumes that algorithm $\cal A$ not only outputs $\bfx$, but also $\bfy$, which is supposed to be $H(\bfx)$. Most of the time, this can be assumed without loss of generality, by letting $\cal A$ make a few more queries to obtain and then output $\bfy = H(\bfx)$. This increases the parameter $q$ by at most $\lceil\ell/k\rceil$. However, this is problematic when the relation $R$ of interest actually depends on $q$, as is the case for the hash chain problem. There, we will then use the following variant, which does not require $\cal A$ to output $\bfy$, but requires some additional quantum transition capacity to be analyzed. 

\begin{theorem}\label{thm:QTCBound2}
Let $R$ be a relation. Let $\cal A$ be a $k$-parallel $q$-query oracle quantum algorithm that outputs ${\bf x} \in \Xcal^\ell$, and let $p_\circ$ be the probability that $({\bf x},H({\bf x})) \in R$ when $\cal A$ has interacted with the standard random oracle, initialized with a uniformly random function $H$.  
Then, for any database properties $\P_0,\ldots, \P_q$ with $\P_0=\neg {\boldsymbol \bot}$ and $\P_q=\P^R$:
$$
\sqrt{p_\circ} \,\leq\, \qQTC{{\boldsymbol \bot}}{\P^R}{k}{q} + \max_\bfx\QTC{\neg \P^R_\bfx}{\P^R_\bfx}{\ell} + \sqrt{\frac{\ell}{M}}
\,\leq\, %\min_{\P_0,\ldots, \P_q \atop \P_0=\neg {\boldsymbol \bot}, \P_q=\P^R } 
\sum_{s=1}^q \QTC{\neg\P_{s-1}}{\P_{s}}{k} + \max_\bfx\QTC{\neg \P^R_\bfx}{\P^R_\bfx}{\ell} + \sqrt{\frac{\ell}{M}} \: ,
$$
where $\P^R_\bfx := \{ D \in \DB \,|\, (\bfx, D(\bfx)) \in R \}$. 
\end{theorem}

\begin{proof}
Let $\tilde{\cal A}$ be the algorithm that runs $\cal A$ and then makes $\ell$ more {\em classical} queries to obtain $\bfy = H(\bfx)$, which it then outputs along with $\bfx$. Furthermore, let
\begin{itemize}
    \item $\tilde p$ be the probability $(\bfx,\bfy) \in R$ and $\bfy = H(\bfx)$ when $\tilde{\cal A}$ interacts with the standard random oracle, 
    \item $\tilde p'$ be the probability that $(\bfx,\bfy) \in R$ and $\bfy = D(\bfx)$ when $\tilde{\cal A}$ interacts with the compressed oracle instead, and 
    \item  $\tilde p'_\circ$ be the probability that $(\bfx,D(\bfx)) \in R$ when $\tilde{\cal A}$ interacts with the compressed oracle. 
\end{itemize}
We trivially have that $p_\circ = \tilde p$ and $\tilde p' \leq \tilde p'_\circ$. Also, applying Corollary~\ref{cor:zha} to $\tilde{\cal A}$, we see that 
$$
\sqrt{\tilde p} \leq \sqrt{\tilde p'} + \sqrt\frac{\ell}{M} \, .
$$
It remains to bound $\tilde p'_\circ$. For that purpose, we write $\CO_{\bfx}$ for the the evolution of the compressed oracle upon {\em classical} queries to the components of $\bfx$. Formally, $\CO_{\bfx}$ is a unitary acting on $\CC[\DB] \otimes \CC[\Ycal]^{\otimes \ell}$.%
\footnote{We emphasize that the components $\CO_{x_i}$ of $\CO_{\bfx}$ act on different response registers. Also, for these queries to be purely classical, the response registers have to be initialized in state $\ket{0}$; however, we do not exploit this. }
Also, we let $\{\proj{\bfx}\}$ be the measurement acting on $\cal A$'s output register to produce the output $\bfx$. Then, with the supremum over all $0 < d \in \ZZ$ and all unitaries $U_1,\ldots,U_{q}$ acting on $\CC[\Xcal] \otimes \CC[\Ycal] \otimes \CC^d$, 
\begin{align*}
\sqrt{\tilde p'_\circ} &\leq \sup_{U_1,\ldots,U_{q}} \bigg\| \sum_\bfx (\P^R_\bfx \, \CO_{\bfx} \otimes \proj{\bfx}) \, U_{q} \, \CO^k \,U_{q-1} \, \CO^k \cdots U_1 \, \CO^k {\boldsymbol \bot} \bigg\| \\
&\leq  \bigg\| \sum_\bfx (\P^R_\bfx \, \CO_{\bfx} \otimes \proj{\bfx}) (\id - \P^R) \bigg\|+ \sup_{U_1,\ldots,U_{q}} \big\| \P^R \, U_{q} \, \CO^k \,U_{q-1} \, \CO^k \cdots U_1 \, \CO^k {\boldsymbol \bot} \big\| \\
    &\leq \bigg\| \sum_\bfx \P^R_\bfx \, \CO_{\bfx}  (\id - \P^R) \otimes \proj{\bfx} \bigg\|  + \sup_{U_1,\ldots,U_{q-1}} \big\| \P^R \, \CO^k \,U_{q-1} \, \CO^k \cdots U_1 \, \CO^k {\boldsymbol \bot} \big\| \\
    &\leq \max_\bfx \big\| \P^R_\bfx \, \CO_{\bfx}  (\id - \P^R) \big\|  + \qQTC{{\boldsymbol \bot}}{\P^R}{k}{q} \\
    &\leq \max_\bfx \big\| \P^R_\bfx \, \CO_{\bfx}  (\id - \P^R_\bfx) \big\|  + \qQTC{{\boldsymbol \bot}}{\P^R}{k}{q}  \, .
\end{align*}
Using the same reasoning as in the proof of Lemma~\ref{lem:tech}, 
$$
\big\| \P^R_\bfx \, \CO_{\bfx}  (\id - \P^R_\bfx) \big\| \leq \max_{D,\hat\bfy} \big\| \P^R_\bfx|_{D|^\bfx} \, \CO_{\bfx \hat\bfy}  (\id - \P^R_\bfx|_{D|^\bfx}) \big\| \leq \QTC{\neg \P^R_\bfx}{\P^R_\bfx}{\ell} \, ,
$$
which concludes the proof. 
\end{proof}

In the following section, we offer techniques to bound the quantum transition capacities (in typical cases) using {\em classical} reasoning. Together with Theorem~\ref{thm:QTCBound} or \ref{thm:QTCBound2}, this then provides means to prove lower bounds on the quantum query complexity (for certain computational problems in the random oracle model) using purely classical reasoning.

\subsection{Bounding Quantum Transition Capacities Using Classical Reasoning Only}

The general idea is to ``recognize" a database transition $\neg\P \to \P$ in terms of {\em local} properties $\L$, for which the truth value $\L(D)$ only depends on the function value $D(x)$ at {\em one single point} $x$ (or at few points), and then to exploit that the behavior of the compressed oracle at a single point $x$ is  explicitly given by Lemma~\ref{lemma:cO-Matrix}. In the following two sections, we consider two possible ways to do this, but first we provide the formal definition for local properties. 

%\begin{definition}
%A database property ${\cal L}$ on $\DB|_\bfx$ is called {\em $t$-local} if there exists $I \subseteq \{1,\ldots,k\}$, called the {\em support of ${\cal L}$}, so that that ${\cal L}(\bfr)$ is uniquely determined by $r_I$. We then also write ${\cal L}(r_I)$ for ${\cal L}(\bfr)$, and similarly $r_I \in {\cal L}$. 
%\end{definition}

%\begin{definition}
%Given database properties $\Gcal' \subseteq \Gcal$ on $\DB|_\bfx$, a family of $t$-local properties ${\cal L}_1,\ldots,{\cal L}_N$ with disjoint supports \serge{Is the latter necessary?} is called {\em necessary for $\Gcal'$ and sufficient for $\Gcal$} if
%$$
%\Gcal' \subseteq \bigcup_i {\cal L}_i \subseteq \Gcal
%$$
%\end{definition}

\begin{definition}\label{def:Locality}
A database property $\L \subseteq \DB$ is {\em $\ell$-local} if $\,\exists \, \bfx = (x_1,\ldots,x_{\ell}) \in \Xcal^\ell$ so that
\begin{enumerate}
    \item\label{it:locality} the truth value of $\L(D)$ is uniquely determined by $D(\bfx)$, and 
%    \item\label{it:monotonicity} if $D \in \L$ with $D(x_i) = \bot$ for some $i \in \{1,\ldots,\ell\}$, then $D[x_i \!\mapsto\!r_i] \in \L$ for any $r_i \in \Ycal$. 
    \item\label{it:monotonicity} if $D \in \L \,\wedge\, (\exists \,i \in \{1,\ldots,\ell\} : D(x_i) = \bot)$ then $D[x_i \!\mapsto\!r_i] \in \L \; \forall \, r_i \in \Ycal$. 
\end{enumerate}
The set $\{x_1,\ldots,x_{\ell}\}$ is then called the {\em support} of $\L$, and denoted by $\supp(\L)$. 
\end{definition}

\begin{remark}\label{rem:Supp}
We observe that, as defined above, the support of an $\ell$-local property is not necessarily uniquely defined: if $\ell$ is not minimal with the required property then there are different choices. A natural way to have a unique definition for $\supp(\L)$ is to require it to have minimal size. For us, it will be more convenient to instead consider the choice of the support to be part of the specification of $\L$.% 
\footnote{E.g., we may consider the constant-true property $\L$ with support $\supp(\L) = \emptyset$, in which case it is $\ell$-local for any $\ell \geq 0$, or we may consider the same constant-true property $\L$ but now with the support set to $\supp(\L) = \{x_\circ\}$ for some $x_\circ \in \Xcal$, which then is $\ell$-local for $\ell \geq 1$.}
Furthermore, we then declare that $\supp(\L \cup {\sf M}) = \supp(L) \cup \supp({\sf M})$, and $\supp(\L|_{D|^\bfx}) = \supp(\L) \cap \{x_1,\ldots,x_k\}$ for any $D \in \DB$ and $\bfx = (x_1,\ldots,x_k)$.%
\footnote{The above mentioned alternative approach would give $\subseteq$. }
\end{remark}

\begin{remark}
Condition \ref{it:monotonicity} captures that $\bot$ is a dummy symbol with no more ``value" than any other $r \in \Ycal$. 
\end{remark}

For example, for any database property $\P$, and for any $\bfx = (x_1,\ldots,x_\ell)$ and $D$, the property $\P|_{D|^\bfx}$ satisfies requirement \ref{it:locality}.\ of Definition~\ref{def:Locality}. In line with this, Remark~\ref{rem:convention} applies here as well: we may identify an $\ell$-local property $\L$ with a subset of $\bar\Ycal^\ell$.

\subsubsection{Reasoning via Strong Recognizability}

\begin{definition}\label{def:UnifStrongRec}
A database transition $\neg\P \to \P'$ is said be {\em (uniformly) strongly recognizable} by $\ell$-local properties if there exists a family of $\ell$-local properties $\{\L_i\}_i$ so that
\begin{equation}\label{eq:UnifStrongRec}
\P' \subseteq \bigcup_i \L_i \subseteq \P \, .
\end{equation}
\end{definition}
We also consider the following weaker but somewhat more intricate version. 

\begin{definition}\label{def:NonunifStrongRec}
A database transition $\neg\P \to \P'$ is said be {\em $k$-non-uniformly strongly recognizable} by $\ell$-local properties if for every $\bfx = (x_1,\ldots,x_k) \in {\cal X}^k$ with disjoint entries, and for every $D \in \DB$, there exist a family $\{\L^{\bfx,D}_i\}_i$ of $\ell$-local properties $\L^{\bfx,D}_i$ with supports in $\{x_1,\ldots,x_k\}$ so that
\begin{equation}\label{eq:NonunifStrongRec}
\P'|_{D|^\bfx} \subseteq \bigcup_i \L^{\bfx,D}_i \subseteq \P|_{D|^\bfx} \, .
\end{equation}
\end{definition}

It is easiest to think about these definitions for the case $\P = \P'$, where (\ref{eq:UnifStrongRec}) and  (\ref{eq:NonunifStrongRec}) become equalities. Requirement (\ref{eq:UnifStrongRec}) then means that for $D$ to satisfy $\P$ it is {\em necessary} and {\em sufficient} that $D$ satisfies one of the local properties. 

% Note that, informally, in (\ref{eq:NonunifStrongRec}) we fix the database outside of $\bfx$, and consider databases that vary within $\bfx$ only. 

\begin{remark}\label{rem:DistinctSupports}
In the above definitions, as long as the support-size remains bounded by $\ell$, one can always replace two properties by their union without affecting (\ref{eq:UnifStrongRec}), respectively (\ref{eq:NonunifStrongRec}). Thus, we may\,---\,and by default do\,---\,assume the $\L_i$'s to have {\em distinct} (though not necessarily disjoint) supports in Definition~\ref{def:UnifStrongRec}, and the same we may assume for the $\L^{\bfx,D}_i$'s for every $\bfx$ and $D$ in Definition~\ref{def:NonunifStrongRec}. 
%This will allow us to bound the size of the family $\{\L^{\bfx,D}_i\}_i$. 
\end{remark}

\begin{remark}\label{rem:FromUnifToNon}
It is easy to see that Definition~\ref{def:UnifStrongRec} implies Definition~\ref{def:NonunifStrongRec} with $\L^{\bfx,D}_i := \L_i|_{D|^{\bfx}}$. 
%If the support of $\L_i$ lies entirely outside of $\{x_1,\ldots,x_k\}$ then $\L_i|_{D|^\bfx}$ is {\em trivial}, meaning constant true or false.%
%\footnote{By convention, the support of such a constant property is empty.}
%In case $\L^{\bfx,D}_i$ is constant false, it can be omitted; in case it is constant true, all $\L^{\bfx,D}_j$ with $j \neq i$ can be omitted, while (\ref{eq:NonunifStrongRec}) remains satisfied. 
\end{remark}

%\paragraph{Example.} $\P' = \P = \PRMG$ is uniformly strongly recognized by the $1$-local properties $\L_x = \{D|D(x)\!=\!0\}$, i.e., $\PRMG = \bigcup_x \L_x$. Note that, as a subset of $\bar{\cal Y}$, the property $\L_x|_{D|^\bfx}$ is either $\{0\}$ or trivial.%
%\footnote{In more detail, $\L_x|_{D|^\bfx}=\{0\}$ whenever $x \in \{x_1,\ldots,x_k\}$, and otherwise it is constant true if $D(x) = 0$ and constant false if $D(x) \neq 0$. }
%\par\bigskip

\begin{theorem}\label{thm:simple}
Let $\neg\P \to \P'$ be $k$-non-uniformly strongly recognizable by $1$-local properties $\{\L_1^{\bfx,D},\ldots,\L_k^{\bfx,D}\}$, where, without loss of generality, the support of $\L_i^{\bfx,D}$ is $\{x_i\}$. Then
$$
\QTC{\neg\P}{\P'}{k} \leq \max_{\bfx,D}\sqrt{10 \sum_i P\bigl[U \!\in\! \L_i^{\bfx,D}\bigr]}
$$
\ifnum\submission=1 \par\vspace{-1ex}\noindent \fi
with the convention that $P\bigl[U \!\in\! \L_i^{\bfx,D}\bigr] = 0$ if $\L_i^{\bfx,D}$ is trivial (i.e. constant true or false). 
%Furthermore, assume that there exists $\delta$ so that, uniformly for all $\bfx,D,\hat y$ and $i$, 
%\begin{equation}\label{eq:simple}
%\sum_{r \not\in \L_i^{\bfx,D}} P[U \!\in\! \L_i^{\bfx,D}| r, \hat y] \leq \delta^2 \, .
%P\bigl[U \!\in\! \L_i^{\bfx,D}\bigr] \leq \delta \, .
%\end{equation}
%Then, the quantum transition capacity $\QTC{\neg\P}{\P'}{k}$ is bounded by $\QTC{\neg\P}{\P'}{k} \leq \sqrt{k (\delta+\frac{9}{M})}$. 
\end{theorem}

Before doing the proof, let us show how the above can be used to bound the probability of finding a $0$-preimage. 

\begin{example}
$\P' = \P = \PRMG$ is uniformly strongly recognized by the $1$-local properties $\L_x = \{D|D(x)\!=\!0\}$. Furthermore, as a subset of $\bar{\cal Y}$, the property $\L_x^{\bfx,D} := \L_x|_{D|^\bfx}$ is either $\{0\}$ or trivial.%
\footnote{In more detail, $\L_x|_{D|^\bfx}=\{0\}$ whenever $x \in \{x_1,\ldots,x_k\}$, and otherwise it is constant true if $D(x) = 0$ and constant false if $D(x) \neq 0$. }
In the non-trivial case, we obviously have  $P\bigl[U \!\in\! \L_i^{\bfx,D}\bigr] = P[U \!=\! 0] = 1/M$. 
It then follows from Theorem~\ref{thm:simple} that
\ifnum\submission=1 
  we can bound the transition capacity as $\QTC{\neg{\PRMG}}{{\PRMG}}{k} \leq \sqrt{10k/M}$
\else
  $$
  \QTC{\neg{\PRMG}}{{\PRMG}}{k} \leq \sqrt{\frac{10k}{M}} \, ,
  $$
\fi
and thus from Theorem~\ref{thm:QTCBound}, setting $\P_i = {\PRMG}$ for all $i$, that the probability $p$ of any $k$-parallel  $q$-query algorithm outputting a $0$-preimage $x$ is bounded by 
$$
p \leq \biggl(q \sqrt{\frac{10k}{M}} + \frac{1}{\sqrt{M}}\biggr)^2 = O\biggl(\frac{k q^2 }{M}\biggr) \, .
$$
\end{example} 

\begin{proof}[\ifnum\submission=0 Proof (of Theorem~\ref{thm:simple}). \else of Theorem~\ref{thm:simple}\fi]
Consider arbitrary $\bfx$ and $D$.
To simplify notation, we then write $\L_i$ for $\L^{\bfx,D}_i$. 
We introduce the properties ${\sf M}_i := \L_i \setminus (\bigcup_{j<i} \L_j)$ for $i \in \{1,\ldots,k\}$. By assumption (\ref{eq:NonunifStrongRec}), as projectors they satisfy
\ifnum\submission=1 
  the operator inequalities $\P'|_{D|^\bfx} \leq \sum_i {\sf M}_i \leq \sum_i \L_i$ and ${\sf M}_i \leq \L_i \leq \P|_{D|^\bfx} \: \forall i$,
\else
  $$
  \P'|_{D|^\bfx} \leq \sum_i {\sf M}_i \leq \sum_i \L_i
  \qquad\text{and}\qquad
  \forall i: {\sf M}_i \leq \L_i \leq \P|_{D|^\bfx} \, ,
  $$
\fi
where, additionally, the ${\sf M}_i$'s are mutually orthogonal. 
Then, exploiting the various properties, for any $\hat\bfy$ we have
\begin{align*}
\|\P'|_{D|^\bfx} \,\CO_{\bfx\hat\bfy} \, (\id - \P|_{D|^\bfx})&\|^2 
\leq 
\bigg\|\sum_i {\sf M}_i \,\CO_{\bfx\hat\bfy} \, (\id - \P|_{D|^\bfx})\bigg\|^2 
= \sum_i \|{\sf M}_i \,\CO_{\bfx\hat\bfy} \, (\id - P|_{D|^\bfx})\|^2 \\
&\leq \sum_i \|\L_i \,\CO_{\bfx\hat\bfy} \, (\id - \L_i)\|^2
= \sum_i \|\L_i \,\CO_{x_i \hat y_i} \, (\id - \L_i)\|^2 \, ,
\end{align*}
where, by considering the map as a map on $\CC[\bar\Ycal]$ and bounding operator norm by the Frobenius norm, 
\begin{align*}
\|\L_i \,\CO_{x_i \hat y_i} \, (\id - \L_i) & \|^2 \leq \!\sum_{r_i,u_i \in \bar\Ycal}\! |\bra{u_i} \L_i \,\CO_{x_i \hat y_i} \, (\id - \L_i) \ket{r_i}|^2   %\ifnum\submission=1 \\ & \fi
= \sum_{r_i \not\in {\cal \L}_i \atop u_i \in \L_i} |\bra{u_i} \CO_{x_i \hat y_i} \ket{r_i}|^2 
= \sum_{r_i \not\in \L_i} \tilde P[U \!\in\! \L_i| r_i, \hat y_i] \, .
\end{align*}
The claim now follows from (\ref{eq:connection}), with the additional observations that if $\bot \in \L_i$ (in which case (\ref{eq:connection}) does not apply) then $\L_i$ is constant-true (by property \ref{it:monotonicity} of Definition~\ref{def:Locality}), and that the sum vanishes if $\L_i$ is constant-true. 
\ifnum\submission=1 \qed \fi
%respectively from Remark~\ref{rem:FromUnifToNon} for the case of a trivial $\L_i$. 
\end{proof}

\subsubsection{Reasoning via Weak Recognizability}

Here, we consider a weaker notion of recognizability, which is wider applicable but results in a slightly worse bound. Note that it will be more natural here to speak of a transition $\P \to \P'$ instead of $\neg\P \to \P'$, i.e., we now write $\P$ for what previously was its complement. 

%\begin{definition}
%Given database properties $\Gcal' \subseteq \Gcal$ on $\DB|_\bfx$, a family of $t$-local properties ${\cal L}_1,\ldots,{\cal L}_N$ with disjoint supports \serge{Is the latter necessary?} is called {\em necessary} if
%$$
%\bfr \not\in \Gcal \:\wedge\: \bfu \in \Gcal'  \;\Longrightarrow\; \exists \, n : \bfr \not\in {\cal L}_n \:\wedge\: \bfu \in {\cal L}_n \, .
%$$
%\end{definition}

%\begin{definition}
%Given database properties $\Gcal' \subseteq \Gcal \subseteq \DB$, a family of $t$-local properties ${\cal L}_1,\ldots,{\cal L}_N$ with disjoint supports \serge{Is the latter necessary?} is called {\em necessary} if
%$$
%D \not\in \Gcal \:\wedge\: D' \in \Gcal'  \;\Longrightarrow\; \exists \, n : D \not\in {\cal L}_n \:\wedge\: D' \in {\cal L}_n \, .
%$$
%\end{definition}

\begin{definition}\label{def:UnifWeakRec}
A database transition $\P \to \P'$  is said be {\em (uniformly) weakly recognizable} by $\ell$-local properties if there exists a family of $\ell$-local properties $\{\L_i\}_i$ so that
%$$
%D \not\in \P \:\wedge\: D' \in \P' \;\Longrightarrow\; \exists \, i : D \not\in \L_i \:\wedge\: D' \in \L_i 
%$$
$$
D \in \P \:\wedge\: D' \in \P' \;\Longrightarrow\; \exists \, i : D' \in \L_i \:\wedge\: \big( \exists\, x \!\in\! \supp(\L_i) : D(x) \neq D'(x) \big) \, .
$$
%when we speak of being {\em (uniformly) weakly recognizable}. 
\end{definition}
%
%Note that the latter in particular implies that $D$ and $D'$ must differ on the support of $\L_i$. 
Also here, we have a non-uniform version (see below). Furthermore, Remarks~\ref{rem:DistinctSupports} and \ref{rem:FromUnifToNon} apply correspondingly;%
\footnote{We point out that this is thanks to our convention on the definition of the support, as discussed in Remark~\ref{rem:Supp}.}
in particular, we may assume the supports in the considered families of local properties to be distinct. 

\begin{definition}\label{def:NonunifWeakRec}
A database transition $\P \to \P'$ is said be {\em $k$-non-uniformly  weakly recognizable} by $\ell$-local properties if for every $\bfx = (x_1,\ldots,x_k) \in {\cal X}^k$ with disjoint entries, and for every $D \in \DB$, there exist a family of $\ell$-local properties $\{\L^{\bfx,D}_i\}_i$ with supports in $\{x_1,\ldots,x_k\}$ so that
%\begin{equation}\label{eq:NonunifRec}
%\bfr \not\in \P|_{D|^\bfx} \:\wedge\: \bfu \in \P'|_{D|^\bfx}  \;\Longrightarrow\; \exists \, i : \bfr \not\in \L^{\bfx,D}_i \:\wedge\: \bfu \in \L^{\bfx,D}_i 
%\end{equation}
\begin{equation}\label{eq:NonunifRec}
D_\circ \in \P|_{D|^\bfx} \:\wedge\: D' \in \P'|_{D|^\bfx}  \,\Longrightarrow\, \exists \, i : D' \in \L^{\bfx,D}_i  \:\wedge\: \big( \exists\, x \!\in\! \supp(\L^{\bfx,D}_i) : D_\circ(x) \neq D'(x) \big) \, .
\end{equation}
%If we replace the above condition by 
%\begin{equation}\label{eq:NonunifWeakRec}
%\bfr \not\in \Gcal|_{\bfx,D} \:\wedge\: \bfu \in \Gcal'|_{\bfx,D}  \;\Longrightarrow\; \exists \, i : \bfr|_{\supp({\cal L}^{\bfx,D}_i)} \neq \bfu|_{\supp( {\cal L}^{\bfx,D}_i)} \:\wedge\: \bfu \in {\cal L}^{\bfx,D}_i \, .
%\end{equation}
%when we speak of being {\em $k$-non-uniformly weakly recognizable}. 
\end{definition}

\begin{remark}
Viewing $\L^{\bfx,D}_i$ as subset of $\bar\Ycal^k$, and its support $\L^{\bfx,D}_i = \{x_{i_1},\ldots,x_{i_\ell}\}$ then as subset $\{i_1,\ldots,i_\ell\}$ of $\{1,\ldots,k\}$,  (\ref{eq:NonunifRec}) can equivalently be written as follows, which is in line with Lemma~\ref{lem:classicalCHAINproperty} (where $\supp(\L^{\bfx,D}_i) = \{i\}$): 
$$
D[\bfx \mapsto \bfr] \in \P \:\wedge\: D[\bfx \mapsto \bfu] \in \P'  \;\Longrightarrow\; \exists \, i : \bfu \in \L^{\bfx,D}_i \:\wedge\:  \big(\exists \, j \in \supp(\L^{\bfx,D}_i): \bfr_j \neq \bfu_j \big) \, .
$$
\end{remark}

\begin{example}
Consider ${\chain}^q = \{D\,|\,\exists\, x_0,x_1,\ldots,x_q \in {\cal X}: D(x_{i-1}) \triangleleft x_i \:\forall i\}$ for an arbitrary positive integer $q$. For any $\bfx$ and $D$, we let $\L_i = \L^{\bfx,D}_i$ be the $1$-local property that has support $\{x_i\}$ and, as a subset of $\bar{\cal Y}$, is defined as (\ref{eq:LocalPropForChain}), 
%$$
%\L_i = \{x|D(x)\!\neq\!\bot\} \cup \{x_1,\ldots,x_k\} \, ,
%$$
i.e., so that $u \in \L_i$ if and only if $u \triangleleft x$ for some $x$ with $D(x) \neq \bot$ or $x \in \{x_1,\ldots,x_k\}$. 
Lemma~\ref{lem:classicalCHAINproperty} from the classical analysis shows that condition (\ref{eq:NonunifRec}) is satisfied for the database transition $\neg{\chain}^q \to {\chain}^{q+1}$. 
This in particular implies that (\ref{eq:NonunifRec}) is satisfied for the database transition $\SIZE[k(q-1)] \setminus {\chain}^q \to {\chain}^{q+1}$; in this latter case however, whenever $D$ is not in $\SIZE$, which then means that the left hand side of (\ref{eq:NonunifRec}) is never satisfied, we may simply pick the constant-false property as family of local properties satisfying (\ref{eq:NonunifRec}). 
\end{example}

\begin{theorem}\label{thm:tricky}
Let $\P \to \P'$ be $k$-non-uniformly weakly recognizable by $1$-local properties $\L_i^{\bfx,D}$, where the support of $\L_i^{\bfx,D}$ is $\{x_i\}$ or empty. 
Then
$$
\QTC{\P}{\P'}{k} \leq \max_{\bfx,D} e \sum_i \sqrt{10 P\bigl[U \!\in\! \L_i^{\bfx,D}\bigr]} \, ,
$$
\ifnum\submission=1 \par\vspace{-1ex}\noindent \fi
%with the convention that $P\bigl[U \!\in\! \L_i^{\bfx,D}\bigr] = 0$ if $\L_i^{\bfx,D}$ is trivial, and 
where $e$ is Euler's number.%
\footnote{Unlike Theorem~\ref{thm:simple}, here is no convention that \smash{$P\bigl[U \!\in \L_i^{\bfx,D}\bigr] = 0$} if $\L_i^{\bfx,D}$ is constant-true. This has little relevance since $\L_i^{\bfx,D}$ being constant-true can typically be avoided via Remark~\ref{rem:DistinctSupports}. }
%Assume that there exists $\delta \leq 1/k$ \serge{Possibly needs to be adjusted} so that, uniformly for all $\bfx,D,\hat y$ and~$i$, 
%\begin{equation}\label{eq:tricky}
%\sum_{r \in \bar{\cal Y}} P[r \!\neq\! U \!\in\! {\cal L}_i^{\bfx,D}| r, \hat y] \leq \delta^2 \, .
%P[U \!\in\! {\cal L}_i^{\bfx,D}] \leq \delta \, .
%\end{equation}
%Then, $\QTC{\neg\P}{\P'}{k} \leq e k \sqrt{\delta+\frac{9}{M}}$, where $e$ is Euler's number.
\end{theorem}

\begin{example} \label{ex:chain_bound}
In the above example regarding ${\chain}^q$ with the considered $\L_i$'s for $D \in \SIZE[kq]$, as in the derivation of the classical bound in Section~\ref{sec:ClLowerBounds}, it holds that $P[U \!\in\! \L_i] \leq kqT/M$, where $T$ denotes the maximal number of $y \in \Ycal$ with $y \triangleleft x$ (for any $x$).%
\footnote{For $D \not\in \SIZE[kq]$ we get the trivial bound $0$ since we may then choose $\L_i$ to be constant false. }
%using again Fig.~\ref{fig:Evolution} it is straightforward to see that%
%\footnote{For $D \not\in \SIZE$ we get the trivial bound $0$ since we may then choose ${\cal L}_i$ to be constant false, i.e., the empty set. }
%$$
%\sum_{r \in \bar{\cal Y}} P[r \!\neq\! U \!\in\! {\cal L}_i| r, \hat y] \leq \frac{9}{M^2} \cdot qk \cdot M + \frac{qk}{M} = \frac{10 qk}{M} \, .
%$$
Thus, 
$$
\QTC{\SIZE[k(q-1)]\backslash{\chain}^q }{{\chain}^{q+1}}{k} \leq e k \sqrt{\frac{10 k qT}{M}} \, . 
$$
%and applying Theorem~\ref{thm:QTCBound} (and the subsequent remark) to the database transitions $\SIZE[k(s-1)]\setminus{\chain}^s \to {\chain}^{s+1}$ for $s=1,\ldots,q$, we obtain the following bound, which we state as a theorem here given that this is a new bound. 
In order to conclude hardness of finding a $(q+1)$-chain for a $q$-query algorithm in the general case where $x_i$ is not necessarily uniquely determined by the next element in the chain, i.e., when $T > 1$, we have to use Theorem~\ref{thm:QTCBound2} instead of Theorem~\ref{thm:QTCBound}. 
In order to apply Theorem~\ref{thm:QTCBound2}, we also need to bound $\QTC{\neg \chain^{q+1}_{\bfx'}}{\chain^{q+1}_{\bfx'}}{}$ for an arbitrary $\bfx'$. This works along the same lines: Consider arbitrary $D$ and $\bfx$. For the purpose of showing (non-uniform) weak recognizability, we note that if $D[\bfx\!\mapsto \bfr] \not\in \chain^{q+1}_{\bfx'}$ yet $D[\bfx\!\mapsto \bfu] \in \chain^{q+1}_{\bfx'}$ then there must exist a coordinate $x_i$ of $\bfx$, which is also a coordinate $x'_j$ of $\bfx'$, so that $u_i \neq r_i$ and \smash{$u_i \in \L_i^{\bfx,D} := \{ y \,|\, \exists \, i: y \triangleleft x'_i\}$}, and so 
$$
\QTC{\neg \chain^{q+1}_{\bfx'}}{\chain^{q+1}_{\bfx'}}{q+2} \leq e (q+2) \sqrt{ \frac{10 T (q+2)}{M}} \, . 
$$
Applying Theorem~\ref{thm:QTCBound2}, we obtain the following bound. 
\end{example}

\begin{theorem}\label{thm:qChain}
Let $\triangleleft$ be a relation over $\Ycal$ and $\Xcal$. 
The probability $p$ of any $k$-parallel $q$-query oracle algorithm $\cal A$ outputting $x_0,x_1,\ldots,x_{q+1} \in \Xcal$ with the property that $H(x_i) \triangleleft x_{i+1}$ for all $i \in \{0,\ldots,q\}$ is bounded by 
$$
p \leq \biggl(q k e \sqrt{\frac{10qkT}{M}} + e (q+2) \sqrt{ \frac{10 T (q+2)}{M}} + \sqrt{\frac{q+2}{M}}\biggr)^2 = O\biggl(\frac{q^3 k^3 T}{M}\biggr) \, , 
$$
where $T: = \max_x |\{y \in \Ycal \,|\, y \triangleleft x\}|$, and $M$ is the size of the range $\Ycal$ of $H: \Xcal \to \Ycal$. 
\end{theorem}

\begin{proof}[\ifnum\submission=0 Proof (of Theorem~\ref{thm:tricky}). \else of Theorem~\ref{thm:tricky}\fi]
We consider fixed choices of $\bfx$ and $D$, and we then write $\L_i$ for $\L^{\bfx,D}_i$. For arbitrary but fixed $\hat\bfy$, we introduce
$$
A_i := \!\!\sum_{u_i,r_i \text{ s.t.} \atop u_i \in \L_i \wedge r_i \neq u_i}\!\!  \proj{u_i} \,\CO_{x_i \hat y_i} \proj{r_i}
\ifnum\submission=1
$$
and
$$
\else
\qquad\text{and}\qquad
\fi
B_i := \CO_{x_i \hat y_i} - A_i = \!\!\sum_{u_i,r_i \text{ s.t.} \atop u_i \not\in \L_i \vee r_i = u_i}\!\! \proj{u_i} \,\CO_{x_i y_i} \proj{r_i}
$$
and observe that, taking it as understood that the operators $\CO_{x_1 \hat y_1} ,\ldots,\CO_{x_k \hat y_k}$ act on different subsystems,%
\footnote{I.e., strictly speaking, we have $\CO_{\bfx \hat\bfy} = \bigotimes_{j=1}^k \CO_{x_j \hat y_j}$. }
\begin{align*}
\CO_{\bfx \hat\bfy} = \prod_{j=1}^k \CO_{x_j \hat y_j} 
&= \prod_{j=1}^{k-1} \CO_{x_j \hat y_j} A_k + \prod_{j=1}^{k-1} \CO_{x_j \hat y_j} B_k \\
&= \prod_{j=1}^{k-1} \CO_{x_j \hat y_j} A_k + \prod_{j=1}^{k-2} \CO_{x_j \hat y_j} A_{k-1} B_k + \prod_{j=1}^{k-2} \CO_{x_j \hat y_j} B_{k-1} B_k \\
&= \cdots
= \sum_{i=0}^{k} \bigg(\prod_{j < k-i}\!\! \CO_{x_j \hat y_j}\bigg) A_{k-i} \bigg(\prod_{j > k-i}\!\! B_j \bigg) 
\end{align*}
with the convention that $A_0 = \id$. Furthermore, by assumption on the $\L_i$'s, it follows that
$$
Q:= \P'|_{D|^\bfx} \bigg(\prod_{j > 0} B_j\bigg) \P|_{D|^\bfx} = 0 \, .
$$
Indeed, by definition of $\P'|_{D|^\bfx}$ and $\P|_{D|^\bfx}$ (considering them as subsets of $\bar\Ycal^k$ now), for $\bra{\bfu} Q \ket{\bfr}$ {\em not} to vanish, it is necessary that $\bfr \in \P|_{D|^\bfx}$ and $\bfu \in \P'|_{D|^\bfx}$. But then, by assumption, for such $\bfr$ and $\bfu$ there exists $i$ so that $u_i \in \L_i$ and $r_i \neq u_i$, and thus for which $\bra{u_i}B_i\ket{r_i} = 0$. Therefore, $\bra{\bfu} Q \ket{\bfr} = \bra{\bfu} \prod_j B_j \ket{\bfr} = \prod_j \bra{u_j}B_j\ket{r_j}$ still vanishes. 
As a consequence, we obtain 
\begin{align*}
\|\P'|_{D|^\bfx} \,\CO_{\bfx \hat\bfy} \, \P|_{D|^\bfx}\| &\leq \Bigg\| \,\sum_{i=0}^{k-1} \bigg(\prod_{j < k-i}\!\! \CO_{x_j \hat y_j}\bigg) A_{k-i} \bigg(\prod_{j > k-i}\!\! B_j \bigg) \Bigg\| 
\leq \sum_{i=0}^{k-1} \Big(\| A_{k-i}\| \!\prod_{j > k-i}\!\! \|B_j\|\Big) \, .
\end{align*}
Using that $\| B_i\| = \|\CO_{x_i \hat y_i} - A_i\| \leq 1+\| A_i\|$, this is bounded by  
$$
\leq \sum_{i=1}^k \| A_i\| \prod_{j=1}^k(1+\| A_j\|)
\leq \sum_i \| A_i\| \, e^{\sum_j \ln(1+\| A_j\|)}
\leq \sum_i \| A_i\| \, e^{\sum_j\|A_j\|} 
\leq \sum_i \| A_i\| \, e
$$
where the last inequality holds if $\sum_j\|A_j\| \leq 1$, while the final term is trivially an upper bound on the figure of merit otherwise. 
Using the fact that the operator norm is upper bounded by the Frobenius norm, we observe that 
\begin{align*}
\| A_i\|^2 \leq \sum_{r_i,u_i} &|\bra{u_i} A_i \ket{r_i}|^2 = \!\!\sum_{u_i,r_i \text{ s.t.} \atop u_i \in L_i \wedge r_i \neq u_i}\!\!\!|\bra{u_i} \CO_{x_i y_i} \ket{r_i}|^2 
= \sum_{r_i} \tilde P[r_i \!\neq\! U \!\in\! L_i| r_i, y_i] \leq 10 P[U \!\in\! \L_i] \, ,
\end{align*}
where the last inequality is due to (\ref{eq:connection}), here with the additional observation that if $\bot \in \L_i$ (and so (\ref{eq:connection}) does not apply) then, by condition~\ref{it:monotonicity} of Definition~\ref{def:Locality}, $\L_i = \bar\Ycal$, and hence the bound holds trivially. 
\ifnum\submission=1 \qed \fi
\end{proof}

%In Appendix~\ref{app:GeneralLocality}, we consider the case of strong recognizability again but drop the restriction on the locality being $\ell = 1$.  Applied to the example $\col$, we then obtain a bound on the success probability of any algorithm finding a collision (Theorem~\ref{thm:ColBound}). 

\subsubsection{General $\ell$-Locality and Collision Finding}\label{sec:GeneralLocality}

We now remove the limitation on the locality being $\ell = 1$. The bound then becomes a bit more intricate, and we only have a version for {\em strong} recognizability. 

\begin{theorem}\label{thm:simple-general}
Let $\P \to \P'$ be a database transition that is $k$-non-uniformly strongly recognizable by $\ell$-local properties $\L_t$, where we leave the dependency of $\L_t= \L_t^{\bfx,D}$ on $\bfx$ and $D$ implicit. 
Then
$$
\QTC{\P}{\P'}{k} 
\leq \max_{\bfx,D} \, e \ell \sqrt{10 \sum_t \: \max_{x \in \supp(\L_t)} \: \max_{D' \in D|^{\supp(\L_t)}} P\bigl[U \!\in\! \L_t|_{D'|^x}\bigr]} \, .
$$
with the convention that $P\bigl[U \!\in\! \L_t|_{D'|^x}\bigr]$ vanishes if $\L_t|_{D'|^x}$ is trivial. 
%For any $\bfx,D,\hat y,t$ and for any $x_i$ in the support of $\L_t$
%\footnote{We assume the elements of the support to be indexed by $i \in \{1,\ldots,\ell\}$}
%assume that
%$$
%\max_{D' \in D|^{x_i}} \sum_{r \not\in \L_t|_{D'|^{x_i}}}\!\! \tilde P[U \!\in\! \L_t|_{D'|^{x_i}}| r, \hat y_i]  \leq \delta^2 
%$$
%uniformly for all $\bfx,D,\hat y,t$ and~$i$ for a fixed $\delta \leq 1/\ell$. 
%Then $\QTC{\neg\Gcal}{\Gcal'}{k} \leq k^{\ell/2} \ell \delta e$. 
\end{theorem}

In case of {\em uniform} recognizability, where there is no dependency of $\L_t$ on $\bfx$ and $D$, the quantification over first $D$ and then over $D' \in D|^{\supp(\L_t)}$ collapses to a single quantification over $D$, simplifying the statement again a bit. 

\begin{corollary}\label{cor:simple-general}
Let $\P \to \P'$ be a database transition that is uniformly strongly recognizable by $\ell$-local properties $\L_t$. 
Then
$$
\QTC{\P}{\P'}{k} \leq \max_{\bfx,D} e \ell \sqrt{10 \sum_t \max_{x \in \supp(\L_t)} P\bigl[U \!\in\! \L_t|_{D|^{x}}\bigr]} \, .
$$
with the convention that $P\bigl[U \!\in\! \L_t|_{D|^{x}}\bigr]$ vanishes if $\L_t|_{D|^{x}}$ is trivial. 
\end{corollary}

\begin{example}\label{ex:collision_bound}
Consider $\col = \{D \,|\, \exists\, x,x':D(x) = D(x') \neq \bot\}$. For any $D \in \DB$ and $\bfx = (x_1,\ldots,x_k)$, consider the family of $2$-local properties consisting of 
\begin{align*}
&\col_{i,j} := \{D_\circ \in D|^\bfx \,|\, D_\circ(x_i) = D_\circ(x_j)\neq\bot \} \qquad\text{and} \\
&\col_i := \{D_\circ \in D|^\bfx \,|\,\exists \,\bar x \not\in \{x_1,\ldots,x_k\}:D_\circ(x_i) = D(\bar x) \neq \bot \}
\end{align*}
for $i \neq j \in \{1,\ldots,k\}$, with respective supports $\{x_i,x_j\}$ and $\{x_i\}$. 
%Note that these local properties depend on $\bfx$ but work uniformly for any choice of $D$. 

It is easy to see that this family of $2$-local properties satisfies (\ref{eq:NonunifStrongRec}) for the database transition $\neg\col \to \col$. Indeed, if $D$ and $D'$ are identical outside of $\bfx$, and $D$ has no collision while $D'$ has one, then $D'$'s collision must be for $x_i,x_j$ inside $\bfx$, or for one $x_i$ inside and one $\bar x$ outside. 
As an immediate consequence, the family also satisfies (\ref{eq:NonunifStrongRec}) for the database transition $(\SIZE[ks]\setminus\col) \to \col$. In this case though, whenever $D \not\in \SIZE[k(s+1)]$ the left hand side of (\ref{eq:NonunifStrongRec}) is never satisfied and so we may replace the family of local properties to consist of (only) the constant-false property. 

Consider $\bfx = (x_1,\ldots,x_k)$ and $D\in \SIZE[k(s+1)]$ with $s \leq q-1$. Then, for $i\neq j$, as subsets of $\bar\Ycal$ we have that
$$
\col_{i,j}|_{D'|^{x_i}} = \{D'(x_j)\}
\quad\text{and}\quad
\col_{i}|_{D'|^{x_i}} = \{D'(\bar x) \,|\, \bar x \not\in \{x_1,\ldots,x_k\}: D'(\bar x) \neq \bot\}
$$
for any $D' \in D|^{(x_i,x_j)}$ and $D' \in D|^{x_i}$, respectively, 
and therefore
$$
P\bigl[U \!\in\! \col_{i,j}|_{D'|^{x_i}}\bigr] = \frac{1}{M}
\qquad\text{and}\qquad
P\bigl[U \!\in\! \col_{i}|_{D'|^{x_i}}\bigr] \leq \frac{k(q-1)}{M} \, .
$$
So, by Theorem~\ref{thm:simple-general}, 
$$
\QTC{\SIZE[ks]\backslash\col}{\col}{k} 
\leq 2e \sqrt{10 \bigg(\frac{k^2}{M}+\frac{k^2 (q-1)}{M}\bigg)} \, .
= 2e k \sqrt{10 \, \frac{q}{M}}
$$
Hence, by Theorem~\ref{thm:QTCBound}, increasing $q$ by $1$ to allow for the additionally query for $\cal A$ to learn (what he claims to be) $H(x_1) = H(x_2)$, we obtain the following bound.
\end{example}

\begin{theorem}\label{thm:ColBound}
The probability $p$ of any $k$-parallel $q$-query algorithm outputting a collision is bounded by 
$$
p \leq \biggl(2 (q+1) e k \sqrt{10 \, \frac{q+1}{M}} + \sqrt{\frac{2}{M}}\,\biggr)^2 = O\biggl(\frac{k^2 q^3 }{M}\biggr) \, . 
$$
\end{theorem}

The above easily generalizes to a more general notion of collision, where the goal is to find $x$ and $x'$ for which $f\bigl(x,H(x)\bigr) = f\bigl(x',H(x')\bigr)$ for a given function $f: \Xcal \times \Ycal \to {\cal Z}$. Here, writing $f\!D(x)$ as a shorthand of $f\bigl(x,D(x)\bigr)$ with $f\!D(x) = \bot$ if $D(x) = \bot$, one would then consider 
\begin{align*}
&\col_{i,j} := \{D_\circ \in D|^\bfx \,\, f\!D_\circ(x_i) = f\!D_\circ(x_j)\neq\bot \} \qquad\text{and} \\
&\col_i := \{D_\circ \in D|^\bfx \,|\,\exists \,\bar x \not\in \{x_1,\ldots,x_k\}: f\!D_\circ(x_i) = f\!D(\bar x) \neq \bot \}
\end{align*}
where then, as subsets of $\bar \Ycal$, 
\begin{align*}
\col_{i,j}|_{D'|^{x_i}} &= \{ y_i \in \Ycal \,|\, f(x_i,y_i) = f\!D'(x_j)  \} \qquad\text{and}\qquad 
\col_{i}|_{D'|^{x_i}} = \!\!\!\bigcup_{\bar x \not\in \{x_1,\ldots,x_k\}}\!\!\!\! \{ y_i \in \Ycal \,|\, f(x_i,y_i) = f\!D'(\bar x)  \}
\end{align*}
for any $D' \in D|^{(x_i,x_j)}$ and $D' \in D|^{x_i}$, respectively, 
and therefore
$$
P\bigl[U \!\in\! \col_{i,j}|_{D'|^{x_i}}\bigr] = \frac{\Gamma}{M}
\qquad\text{and}\qquad
P\bigl[U \!\in\! \col_{i}|_{D'|^{x_i}}\bigr] \leq \frac{k(q-1)\Gamma}{M} 
$$
with $\Gamma := \max_{x\neq x',y'} |\{ y\in \Ycal \,|\, f(x,y)=f(x',y')\}|$. So, by Theorem~\ref{thm:simple-general}, with the obvious generalization of $\col$, 
$$
\QTC{\SIZE[ks]\backslash\col}{\col}{k} 
\leq 2e \sqrt{10 \bigg(\frac{k^2 \Gamma}{M}+\frac{k^2 (q-1) \Gamma} {M}\bigg)}
= 2e k \sqrt{10 \Gamma \, \frac{q}{M}}  
$$
and so we obtain the following generalization of the collision finding bound. 

\begin{theorem}\label{thm:GenColBound}
For any function $f: \Xcal \times \Ycal \to {\cal Z}$, the probability $p$ of any $k$-parallel $q$-query algorithm outputting $x,x' \in \Xcal$ with  $f\bigl(x,H(x)\bigr) = f\bigl(x',H(x')\bigr)$ is bounded by 
$$
p \leq \biggl(2 (q+1) e k \sqrt{10 \Gamma \, \frac{q+1}{M}} + \frac{2}{\sqrt{M}}\biggr)^2 = O\biggl(\frac{k^2 q^3 \Gamma}{M}\biggr) 
$$
for $\Gamma := \max_{x\neq x',y'} |\{ y\in \Ycal \,|\, f(x,y)=f(x',y')\}|$. 
\end{theorem}

\begin{proof}[\ifnum\submission=0 Proof (of Theorem~\ref{thm:simple-general}) \else of Theorem~\ref{thm:simple-general}\fi]
We first observe that one can recycle the proof of Theorem~\ref{thm:simple} to bound 
$$
\|\P'|_{D,\bfx} \,\CO_{\bfx\hat\bfy} \, \P|_{D,\bfx} \|^2 
\leq \sum_t \|\L_t \,\CO_{\bfx_t \hat \bfy_t} \, (\id - \L_t)\|^2 \, ,
$$
where $\bfx_t$ is the restriction of $\bfx$ to those coordinates that are in $\supp(\L_t)$, and the same for $\hat \bfy_t$. 

We now consider an arbitrary but fixed choice of $t$ and write $\L$ for $\L_t$. We write $\{x_1,\ldots,x_\ell\}$ for its support and set $\bfx' := (x_1,\ldots,x_\ell)$. In order to control $\|\L \,\CO_{\bfx' \hat\bfy'} \, (\id - \L)\|$, we use a similar technique as in the proof of Theorem~\ref{thm:tricky}.%
\footnote{We point out that $\L(D')$ is determined by $D'(\bfx')$; thus, we may consider $\L$ as a property of functions $D' \in D|^{\bfx'} \subseteq D|^\bfx$.} 

For any $x_i \in \supp(\L)$, we set 
$$
A_i := \L \,\CO_{{x_i} \hat y_i} (\id-\L)
\qquad\text{and}\qquad
B_i := \CO_{x_i \hat y_i} - A_i \, .
$$
By means of the same generic manipulations as in the proof of Theorem~\ref{thm:tricky}, we have 
$$
\CO_{\bfx' \hat \bfy'} = \prod_{i=1}^\ell \CO_{x_i \hat y_i}
= \sum_{i=0}^{\ell} \bigg(\prod_{j < \ell-i}\!\! \CO_{x_j \hat y_j}\bigg) A_{\ell-i} \bigg(\prod_{j > \ell-i}\!\! B_j \bigg) 
$$
with the convention that $A_0 = \id$. Furthermore, using
$B_i (\id - \L_t)  = (\id - \L_t) \CO_{x_{\lambda_i} \hat y_{\lambda_i}} (\id - \L_t)$, we see that
$$
\L_t \bigg(\prod_{j > 0} B_j\bigg) (\id - \L_t) = 0 \,.
$$
As a consequence, verbatim as in the proof of Theorem~\ref{thm:tricky}, we obtain 
$$
\|\L \,\CO_{\bfx' \hat\bfy'} \, (\id - \L)\| 
\leq \sum_{i=0}^{\ell-1} \Big(\| A_{\ell-i}\| \!\prod_{j > \ell-i}\!\! \|B_j\|\Big) 
\leq \sum_{i=1}^\ell \|A_i\| \, e \, .
$$
Furthermore, for any $D' \!\in\! D|^\bfx$ and $i \in \{1,\ldots,\ell\}$, on the subspace spanned by $D'|^{x_i}$, the map $A_i$ acts identically to $\L|_{D'|^{x_i}} \,\CO_{x_i \hat y_i} \, (\id - \L|_{D'|^{x_i}})$, and thus, by basic properties of the operator norm, the norm of $A_i$ equals the largest norm of these restrictions: 
$$
\| A_i\| \leq \big\|\L|_{D'|^{x_i}} \,\CO_{x_{\lambda_i} \hat y_i} \, (\id - \L|_{D'|^{x_i}})\big\| \, .
$$
Bounding the operator norm by the Frobenius norm, we then obtain
$$
\| A_i\|^2 \leq \sum_{r \not\in L|_{D'|^{x_i}} \atop u \in L|_{D'|^{x_i}}} |\bra{u} \CO_{x_i \hat y_i} \ket{r}|^2 \leq  \sum_{r} \tilde P[r \!\neq\!U \!\in\! \L|_{D'|^{x_i}}| r, y_i] \leq 
10 P\bigl[U \!\in\! \L|_{D'|^{x_i}}\bigr]
 \, .
$$
where the last inequality is due to (\ref{eq:connection}), with the additional observation that if $\bot \in \L_i$ then, by condition \ref{it:monotonicity} of Definition~\ref{def:Locality}, $\L|_{D'|^{x_i}} = \bar\Ycal$, and thus the sum vanishes. 

Putting things together, we obtain
$$
\|\P'|_{D,\bfx} \,\CO_{\bfx\hat\bfy} \, \P|_{D,\bfx}\|
\leq \sqrt{\sum_t \bigg(\sum_{i=1}^\ell \|A_i\| \, e\bigg)^2 }
\leq e \sqrt{\sum_t 10 \ell^2 P\bigl[U \!\in\! \L|_{D'|^{x_i}}\bigr] } 
$$
which proves the claimed bound. 
\end{proof}

%For completeness, we also state here the weakly recognizable version. 

%\begin{theorem}\label{thm:tricky-general}
%Let $\P \to \P'$ be a database transition that is $k$-non-uniformly weakly recognizable by $\ell$-local properties $\L_t$, where we leave the dependency of $\L_t = \L_t^{\bfx,D}$ on $\bfx$ and $D$ implicit. Then
%$$
%\QTC{\P}{\P'}{k} 
%\leq \max_{\bfx,D} \, e^2 \ell \: \sum_t \: \max_{x \in \supp(\L_t)} \: \max_{D' \in D|^{\supp(\L_t)}} \sqrt{10 P\bigl[U \!\in\! \L_t|_{D'|^x}\bigr]} \:\serge{?????}  \, . 
%$$
%%with the convention that $P\bigl[U \!\in\! \L_t|_{D'|^x}\bigr]$ vanishes if $\L_t|_{D'|^x}$ is trivial.
%\end{theorem}

\subsection{Some Rules for the Quantum Transition Capacity}

As we have seen, certain ``simple" lower bounds on the query complexity (respectively upper bound on the success probability) can be obtained rather directly by bounding the quantum transition capacity by the means discussed above. In more complex scenarios, as we will encounter in the next section, it will be convenient to first {\em manipulate} the quantum transition capacity, e.g., to decompose it into different cases that can then be analyzed individually. We thus show some useful manipulation rules here.  

To start with, since $\CO_{\bfx\hat\bfy}^\dagger = \CO_{\bfx\hat\bfy^*}$, we note that the quantum transition capacity is symmetric:
$$
\QTC{\P}{\P'}{k} = \QTC{\P'}{\P}{k} \, .
$$
Therefore, the following bounds hold correspondingly also for $\QTC{\P}{\P'\cap \Q}{k}$ etc. 

\begin{lemma}\label{lem:shrink}
For any database properties $\P,\P'$ and $\Q$, 
$$
\QTC{\P\cap {\sf Q}}{\P'}{k} \, \leq \min\bigl\{\QTC{\P}{\P'}{k},\QTC{{\sf Q}}{\P'}{k}\bigr\}  \qquad\text{and}
$$
%and
$$
\max\bigl\{\QTC{\P}{\P'}{k},\QTC{{\sf Q}}{\P'}{k}\bigr\} \leq \QTC{\P\cup {\sf Q}}{\P'}{k}\, \leq  \QTC{\P}{\P'}{k} + \QTC{{\sf Q}}{\P'}{k} \, . 
$$
\end{lemma}
In particular, we have the following intuitive rule.  
\begin{corollary}\label{cor:subset} If $\P \subseteq \Q$ then
$\QTC{\P}{\P'}{k} \leq \QTC{\Q}{\P'}{k}$ and $\QTC{\P'}{\P}{k} \leq \QTC{\P'}{\Q}{k}$. 
\end{corollary}

\begin{proof}[\ifnum\submission=0 Proof (of Lemma~\ref{lem:shrink}).\else of Lemma~\ref{lem:shrink}\fi]
As subsets, $(\P\cap {\sf Q})|_{D|^\bfx} = (\P\cap {\sf Q}) \cap D|^\bfx = (\P\cap D|^\bfx) \cap ({\sf Q} \cap D|^\bfx) = \P|_{D|^\bfx} \cap {\sf Q}|_{D|^\bfx}$, and, as projections, $\P|_{D|^\bfx}$ and ${\sf Q}|_{D|^\bfx}$ commute, and $\P|_{D|^\bfx} \cap {\sf Q}|_{D|^\bfx} = \P|_{D|^\bfx} {\sf Q}|_{D|^\bfx} = {\sf Q}|_{D|^\bfx} \P|_{D|^\bfx} {\sf Q}|_{D|^\bfx} \leq \P|_{D|^\bfx}$ and similarly $\leq {\sf Q}|_{D|^\bfx}$. This implies that
$$
\|\P'|_{D|^\bfx} \,\CO_{\bfx\hat\bfy} \, (\P\cap {\sf Q})|_{D|^\bfx}\| \leq \min\bigl\{
\|\P'|_{D|^\bfx} \,\CO_{\bfx\hat\bfy} \, \P|_{D|^\bfx}\|, \|\P'|_{D|^\bfx} \,\CO_{\bfx\hat\bfy} \,{\sf Q}|_{D|^\bfx}\| \bigr\} \, ,
$$
and thus proves the first claim. 
Similarly, but now using that, as projections, 
$$
\P|_{D|^\bfx}, {\sf Q}|_{D|^\bfx} \leq \P|_{D|^\bfx} \cup {\sf Q}|_{D|^\bfx} \leq \P|_{D|^\bfx} + {\sf Q}|_{D|^\bfx} \, , 
$$
we obtain the second claim. 
\ifnum\submission=1 \qed \fi
\end{proof}

In the following, we extend the definition of the quantum transition capacity as follows, which captures a restriction of the query vector $\bfx = (x_1,\ldots,x_k)$ to entries $x_i$ in $X \subseteq \cal \Xcal$. 
\begin{equation}\label{eq:condQTC}
   \condQTC{\P}{\P'}{k}{X} := \max_{\bfx \in X^k \atop \hat\bfy,D} \|\P'|_{D|^\bfx} \,\CO_{\bfx\hat\bfy} \, \P|_{D|^\bfx}\| \, . 
\end{equation}
where the max is restricted to $\bfx \in X^k$. Obviously, $\QTC{\P}{\P'}{k} = \condQTC{\P}{\P'}{k}{\Xcal}$. 

\begin{lemma}\label{lem:ParCond}
Let $X = X' \cup X'' \subseteq \Xcal$ and $k = k' + k''$. Furthermore, let $\P,\P',\P''$ and $\Q$ be database properties. Then 
\begin{equation}\label{eq:Cond1}
\condQTC{\P}{\P''}{k}{X} \leq  
\condQTC{\P}{\P''\backslash\Q}{k}{X} + 
\condQTC{\P}{\Q \cap\P''}{k}{X}  \, ,
\end{equation}
where furthermore
\begin{equation}\label{eq:Cond2}
\condQTC{\P}{\Q \cap\P''}{k}{X} \leq  \condQTC{\P}{\neg\Q}{k'}{X}  +  
\condQTC{\P}{\Q \cap\P'}{k'}{X}  +  \condQTC{\Q\backslash\P'}{\Q \cap\P''}{k''}{X}
\end{equation}
as well as 
\begin{equation}\label{eq:Cond3}
\condQTC{\P}{\Q \cap\P''}{k}{X} \leq  \condQTC{\P}{\neg\Q}{k}{X'} +  
\condQTC{\P}{\Q \cap\P'}{k}{X'} +  \condQTC{\Q\backslash\P'}{\Q \cap\P''}{k}{X''} . 
\end{equation}
\end{lemma}

\begin{proof}
The first inequality follows immediately from Lemma~\ref{lem:shrink}, using that $(\P''\setminus\Q) \cup (\Q \cap\P'') = \P''$. 
%$$
%\|\P''|_{D|^\bfx} \,\CO_{\bfx\bfy} \, \P|_{D|^\bfx}\| \leq \|\Q'|_{D|^\bfx}\P''|_{D|^\bfx} \,\CO_{\bfx\bfy} \, \P|_{D|^\bfx}\| + \|(\id - \Q'|_{D|^\bfx}) \,\CO_{\bfx\bfy} \, \P|_{D|^\bfx}\| \, .
%$$
For the other two, let $\bfx \in X^k$, $\hat\bfy \in \hat\Ycal^k$, $D \in \DB$ be the choices that achieve the maximal value in the definition of \smash{$\condQTC{\P}{\Q \cap \P''}{k}{X}$}. We may assume without loss of generality that $\bfx$ consists of pairwise distinct entries. 
For proving the first inequality, we split up $\bfx$ into $(\bfx',\bfx'') \in X^{k'} \times X^{k''}$, and correspondingly then for $\hat\bfy \in \hat\Ycal^k$. For proving the second inequality, we let $\bfx'$ consist of all coordinates of $\bfx$ that lie in $X'$, and we let $\bfx''$ consist of all coordinates of $\bfx$ that lie in $X''$ but not in $X'$, and $\hat\bfy'$ and $\hat\bfy''$ consists of the corresponding coordinates of $\hat\bfy$; in this case, \smash{$(\bfx',\bfx'') \in X'^{\ell'} \times {X''}^{\ell''}$} with $\ell' + \ell'' = k$. 
In both cases, we have $\CO_{\bfx\hat\bfy} = \CO_{\bfx'\hat\bfy'} \CO_{\bfx''\hat\bfy''}$, and, writing $\P_\bfx$ for $\P|_{D|^\bfx}$ etc., we obtain
\begin{align*}
\condQTC{&\P}{\Q\cap\P''}{k}{X} = \|\P''_\bfx \Q_\bfx  \,\CO_{\bfx''\hat\bfy''} \CO_{\bfx'\hat\bfy'} \, \P_\bfx\| \\
&\leq \|\P''_\bfx \Q_\bfx \,\CO_{\bfx''\hat\bfy''} \Q_\bfx \CO_{\bfx'\hat\bfy'} \, \P_\bfx\| + \| (\id-\Q_\bfx) \CO_{\bfx'\hat\bfy'} \, \P_\bfx\| \\
&\leq \| \P'_\bfx\Q_\bfx \CO_{\bfx'\hat\bfy'} \, \P_\bfx\| + \|\P''_\bfx \Q_\bfx \,\CO_{\bfx''\hat\bfy''} (\id - \P'_\bfx)\Q_\bfx \| + \| (\id-\Q^\bfx) \CO_{\bfx'\hat\bfy'} \, \P_\bfx\| \\
&\leq \| \P'_{\bfx'}\Q_{\bfx'} \CO_{\bfx'\hat\bfy'} \, \P_{\bfx'}\| + \|\P''_{\bfx''} \Q_{\bfx''} \,\CO_{\bfx''\hat\bfy''} (\id - \P'_{\bfx''})\Q_{\bfx''} \| + \| (\id-\Q_{\bfx'}) \CO_{\bfx'\hat\bfy'} \, \P_{\bfx'}\| \, ,
\end{align*}
where the last equality follows from basic properties of the operator norm. The first of the two remaining bounds is now obtained by maximizing the individual terms on the right hand side over $\bfx' \in X^{k'}$ and $\bfx'' \in X^{k''}$ (as well as over $\hat\bfy',\hat\bfy''$ and $D$). For the other case, we maximize over $\bfx' \in X'^{\ell'}$ and $\bfx'' \in {X''}^{\ell''}$ and exploit that, for instance, $\condQTC{\P}{\neg\Q}{\ell'}{X'} \leq \condQTC{\P}{\neg\Q}{k}{X'}$, given that $\ell' \leq k$. 
\ifnum\submission=1 \qed \fi
\end{proof}

By recursive application of Lemma~\ref{lem:ParCond}, we obtain the following. 

\begin{corollary}[Parallel Conditioning]\label{cor:parallel_condition}
Let $X = X_1 \cup \ldots \cup X_h \subseteq \Xcal$ and $k = k_1 + \cdots + k_h$, and let $\P_0,\P_1,\ldots,\P_h$ and $\neg\P_0 \subseteq \Q$ be database properties. Then 
\begin{align*}
\condQTC{\neg\P_0}{\P_h}{k}{X} &\leq \sum_{i=1}^h \condQTC{\neg\P_0}{\neg\Q}{\bar k_i}{X} + \sum_{i=1}^h \condQTC{\Q \backslash \P_{i-1}}{\Q \cap \P_i}{k_i}{X} && \text{and} \\
\condQTC{\neg\P_0}{\P_h}{k}{X} &\leq \sum_{i=1}^h \condQTC{\neg\P_0}{\neg\Q}{k}{\bar X_i} + \sum_{i=1}^h \condQTC{\Q \backslash \P_{i-1}}{\Q \cap \P_i}{k}{X_i} \, ,
\end{align*}
where $\bar k_i = k_1 + \cdots + k_i$ and $\bar X_i = X_1 \cup \ldots \cup X_i$. 
\end{corollary}

\begin{proof}
Applying (\ref{eq:Cond1}) and (\ref{eq:Cond2}) with $\P := \neg\P_0$, $\P' := \P_{h-1}$ and $\P'' := \P_h$, and omitting the ``conditioning" on $X$ for simplicity, we get
\begin{align*}
\QTC{\neg\P_0}{\P_h}{k} \leq & \QTC{\neg\P_0}{\neg\Q}{k} + \longQTC{\neg\P_0}{\neg\Q}{\bar k_{h-1}} 
+ \longQTC{\neg\P_0}{\Q \cap\P_{h-1}}{\bar k_{h-1}} + \longQTC{\Q \backslash \P_{h-1}}{\Q \cap \P_h}{k_h} \, .
\end{align*}
Recursively applying (\ref{eq:Cond2}) to $\longQTC{\neg\P_0}{\Q \cap\P_{h-1}}{\bar k_{h-1}}$ gives the first claim. The second is argued correspondingly. 
\ifnum\submission=1 \qed \fi
\end{proof}

The quantum transition capacity {\em with restricted input}, defined in (\ref{eq:condQTC}), is just the original definition of the quantum transition capacity (Definition~\ref{def:QTC}) but with the considered set $\Xcal$ replaced by $X$. As a consequence, properties for $\QTC{\P}{\P'}{}$ carry over to $\condQTC{\P}{\P'}{}{X}$. For instance, it is still symmetric, and Lemma \ref{lem:shrink} carries over to 
$$
\condQTC{\P\cap {\sf Q}}{\P'}{k}{X} \, \leq \min\bigl\{\condQTC{\P}{\P'}{k}{X},\condQTC{{\sf Q}}{\P'}{k}{X}\bigr\} 
$$
etc. For completeness we spell out here the definition of non-uniform recognizability as well as Theorem~\ref{thm:tricky} for such input-restricted database transitions $\P \to \P' \,|\, X$ (the other types of recognizability can be generalized similarly). 

\begin{definition}\label{def:NonunifWeakRecL}
A database transition $\P \to \P'$ with input restricted in $X\subseteq{\Xcal}$ is said to be {\em $k$-non-uniformly  weakly recognizable} by $\ell$-local properties if for every $\bfx = (x_1,\ldots,x_k) \in X^k$ with disjoint entries, and for every $D \in \DB$, there exist a family of $\ell$-local properties $\{\L^{\bfx,D}_i\}_i$ with supports in $\{x_1,\ldots,x_k\}$ so that
\begin{equation*}%\label{eq:NonunifRec_restricted}
D_\circ \!\in\! \P|_{D|^\bfx} \:\wedge\: D' \!\in\! \P'|_{D|^\bfx}  \:\Longrightarrow\: \exists \, i\! : D' \in \L^{\bfx,D}_i  \:\wedge\: \big( \exists\, x \!\in\! \supp(\L^{\bfx,D}_i) : D_\circ(x) \!\neq\! D'(x) \big) 
\end{equation*} 
\end{definition}
\begin{theorem}\label{thm:tricky_restricted}
Let $\P \to \P'$ with input restricted in $X$ be $k$-non-uniformly weakly recognizable by $1$-local properties $\L_i^{\bfx,D}$, where the support of $\L_i^{\bfx,D}$ is $\{x_i\}$ or empty. 
Then
$$
\condQTC{\P}{\P'}{k}{X} \leq \max_{\bfx,D} e \sum_{i} \sqrt{10 P\bigl[U \!\in\! \L_i^{\bfx,D}\bigr]} \, ,
$$
\ifnum\submission=1 \par\vspace{-1ex}\noindent \fi
where the $\max$ now is over all $\bfx = (x_1,\ldots,x_k) \in X^k$. 
\end{theorem}

\section{Post-Quantum Proof of Sequential Works}\label{sec:PoSW}

In this section, we prove post-quantum security of the  proof of sequential work (PoSW) construction by Cohen and Pietrzak~\cite{cohen2018simple} (referred to as Simple PoSW) using our framework developed in the last section. As a matter of fact, we directly analyze the non-interactive variant of their construction after applying the Fiat-Shamir transformation \cite{fiat1986prove}. As we shall see, the proof is by means of purely classical reasoning, recycling observations that are relevant for arguing classical security and combining them with results provided by our framework. 

%uses only classical arguments and mimics the reasoning in classical security proof closely.   

%We review the notion of PoSW as well as the construction of~\cite{cohen2018simple} and its classical security analysis in Section~\ref{subsec:classical-PoSW}. Then we present our proof of post-quantum security in Section~\ref{subsec:PQ-PoSW}

\subsection{Simple Proof of Sequential Works}\label{sec:classical_PoSW}

For readers not familiar with PoSW, we review the definition in Appendix~\ref{sec:posw-def}. Typically, underlying the construction of a PoSW is a directed acyclic graph (DAG) $G$ with certain ``depth-robust'' properties, and a graph labeling that the prover $\Prover$ is required to compute using a hash function $\Hash$. We proceed to describe the DAG used in Simple PoSW and the graph labeling.

\paragraph{\bf Simple PoSW DAG and Graph Labeling.} 
Let $n \in \mathbb{N}$ and $N = 2^{n+1}-1$. Consider the (directed) complete binary tree $B_n = (\VId{n}, E_n')$ of depth $n$, where $\VId{n} := \bool^{\leq n}$ and $E_n'$ consists of the edges directed towards the root (black edges in Fig.~\ref{fig:PoSW_DAG}). 
The Simple PoSW DAG, denoted by $\gPoSW$, is obtained by adding some additional edges to $B_n$ (red edges in Fig.~\ref{fig:PoSW_DAG}). 
Before giving the formal definition of $\gPoSW$ (Definition~\ref{def:DAG}), we recall some basic terminology and notation in the context of the complete binary tree $B_n$, which we will then also use in the context of $\gPoSW$. 

\begin{definition} 
We write $\rt := \epsilon$ for the {\em root}, and we write $\leaves(\VId{n}) := \bool^n$ for the {\em leaves} in $\VId{n}$. For $T \subseteq \VId{n}$, we set $\leaves(T) := T \cap \bool^n$. 
For $v\notin \leaves(\VId{n})$, let $\lch{v} := v\Vert 0$ and $\rch{v} := v\Vert 1$. 
For $b\in\bool$ and $v\in\bool^{<n}$, let $\mathsf{par}(v\Vert b) := v$ and  $\sib(v\Vert b) := v\Vert \neg b$ (see~Fig.~\ref{fig:PoSW_DAG}, right).

Finally, for a leaf $v \in \leaves(\VId{n})$, we define the \emph{ancestors} of $v$ as $\anc(v) = \smash{\{\mathsf{par}^i(v) \,|\, 0\le i \le n\}}$ and  the \emph{authentication path} of $v$ (as in the Merkle tree) as
$\ap(v) = (\anc(v) \backslash \{\rt\}) \cup \smash{\{\sib(u)  \,|\,  \rt \neq u \in \anc(v) \}}$.
%Finally, for a leaf $v \in \leaves(\VId{n})$, we define the \emph{authentication path} of $v$ (as in the Merkle tree) as $\ap(v) = \smash{\{\mathsf{par}^i(v), \sib(\mathsf{par}^i(v))  \,|\,  0\le i< n\}}.$

%
%
\iffalse
For any $v\in \VId{n}$, define 
\begin{align*}
&\mathsf{par}(b\Vert v) = v \quad\textnormal{ whenever } b\in\bool, v\in\bool^{<n} && \text{(the {\em parent node}),}\\
&\lch{v} = 0\Vert v \quad\textnormal{ whenever } v\not\in\mathsf{leaves}(\VId{n}) := \bool^n && \text{(the {\em left child}),} \\
&\rch{v} = 1\Vert v \quad\textnormal{ whenever } v\not\in\mathsf{leaves}(\VId{n}) &&\text{(the {\em right child}),} \\
&\sib(b\Vert v) = \neg b\Vert v \quad\textnormal{ whenever } b\in\bool, v\in\bool^{<n}  && \text{(the {\em sibling}),} 
\end{align*}

\fi
\iffalse
For any $v\in \VId{n}$, define 
\begin{align*}
&\mathsf{par}(b\Vert v) = v \quad\textnormal{ whenever } b\in\bool, v\in\bool^{<n} && \text{(the {\em parent node}),}\\
&\lch{v} = 0\Vert v \quad\textnormal{ whenever } v\not\in\mathsf{leaves}(\VId{n}) := \bool^n && \text{(the {\em left child}),} \\
&\rch{v} = 1\Vert v \quad\textnormal{ whenever } v\not\in\mathsf{leaves}(\VId{n}) &&\text{(the {\em right child}),} \\
&\sib(b\Vert v) = \neg b\Vert v \quad\textnormal{ whenever } b\in\bool, v\in\bool^{<n}  && \text{(the {\em sibling}),} 
\end{align*}
\fi

\end{definition}

\begin{figure}[h]\vspace{-1ex}
    \centering
    \includegraphics[width=0.63\textwidth]{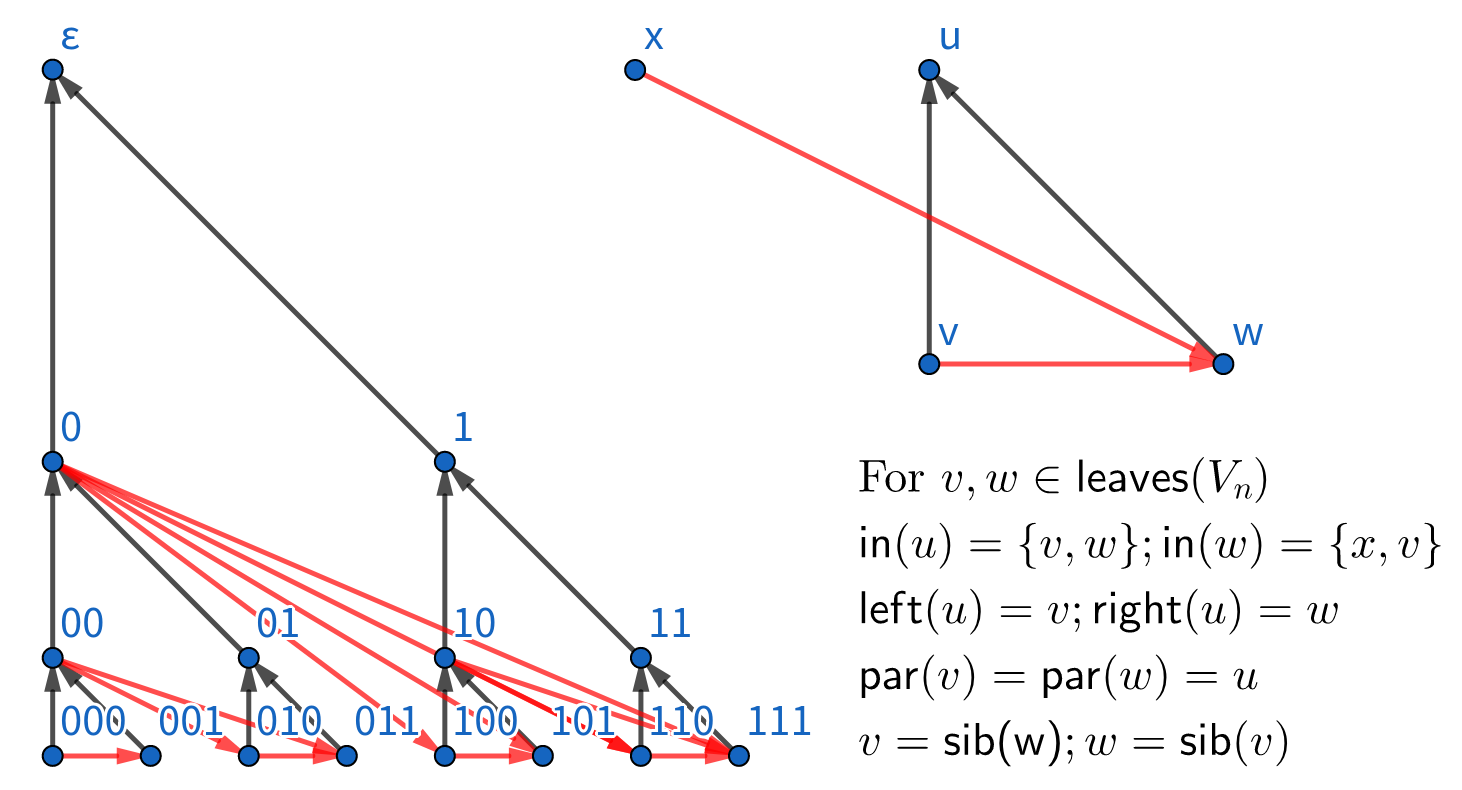}
    \vspace{-1ex}
    \caption{Illustration of the Simple PoSW DAG $\gPoSW$ for $n=3$. %with $E_n'=\{\text{black edges}\}$} and $E_n''=\{\text{red edges}\}$
    }
    \label{fig:PoSW_DAG}
\end{figure}\vspace{-1ex}

\begin{definition}\label{def:DAG}
For $n\in \mathbb{N}$,  define the Simple PoSW DAG $\gPoSW := (V_n,E_n'\cup E_n'')$ with vertex set $V_n$ and edges
\begin{align*}
E_n' &:=\{(\lch{v},v),(\rch{v},v)\,|\, v\in V_n\setminus\mathsf{leaves}(V_n)\} \quad\text{and} \\
E_n''&:=\{(\mathsf{sib}(u),v)\,|\, v\in V_n, u \in \anc(v) \text{ s.t. }u=\mathsf{right}(\mathsf{par}(u))\} \, .
\end{align*}
\end{definition}

%\KM{remove discussion about ``depth-robust'' property of $\gPoSW$ here for now}

%The ``depth-robust'' property that $\gPoSW$ satisfies is the following. Let $S \subseteq \VId{n}$ be a vertex set obtained by removing (an arbitrary number of) subtrees in $\VId{n}$. There exists a path in $\gPoSW$ that go through exactly every vertices in $S$.  \KM{the definition of ``subtree'' is not clear here.} \yuhsuan{This is incorrect.}

For $v\in V_n$, we write $\mathsf{in}(v) := \{u\in V_n\,|\,(u,v)\in E_n'\cup E_n''\}$ to denote the inward neighborhood of $v$. 
We consider a fixed ordering of the vertices (e.g. lexicographic), so that for any set $\{v_1,\dots,v_d\} \in V_n$ of vertices, the corresponding ordered list $(v_1,\dots,v_d)$ is well defined.

We proceed to define the graph labeling for $\gPoSW$ with respect to a hash function $\Hash:\bool^{\leq B}\to\bool^w$, were $w$ is a security parameter, and $B$ is arbitrary large (and sufficiently large for everything below being well defined).   

\begin{definition}[Graph Labeling]\label{def:GraphLabelling}
A function $\ell: V_n \to \bool^w$, $v \mapsto \ell_v$ is a {\em labeling} of $\gPoSW$ with respect to $\Hash$ if 
\begin{equation}\label{eq:label}
\ell_v = \Hash(v,\ell_{\parent(v)})
\end{equation}
for all $v \in V_n$, were $\ell_{\parent(v)}$ is shorthand for $(\ell_{v_1},\dots,\ell_{v_d})$ with $\{v_1,\dots,v_d\} = \parent(v)$. Similarly, for a subtree%
\footnote{By a {\em subtree} of $\gPoSW$ we mean a sub{\em graph} of $\gPoSW$ that is a sub{\em tree} of the complete binary tree $B_n$ when restricted to edges in $E'_n$. We are also a bit sloppy with not distinguishing between the graph $T$ and the vertices of $T$. }  
$T$ of $\gPoSW$, a function $\ell: T \to \bool^w$, $v \mapsto \ell_v$ is a called a {\em labeling} of $T$ with respect to $\Hash$ if $\ell_v = \Hash(v,\ell_{\parent(v)})$ for all $v \in V_n$ for which $\parent(v) \subseteq T$. 
\end{definition}

By the structure of the graph, $\gPoSW$ admits a unique labeling, which can be computed by making $N = 2^{n+1}-1$ sequential queries to $H$, starting with the leftmost leaf. 
We sometimes speak of a {\em consistent} labeling (of $\gPoSW$ or $T$) when we want to emphasize the distinction from an arbitrary function $\ell$. The definition also applies when replacing the function $H$ by a database $D:\bool^{\leq B}\to\bool^w \cup \{\bot\}$, where the requirement (\ref{eq:label}) then in particular means that $\Hash(v,\ell_{\parent(v)}) \neq \bot$.

%\begin{definition}[Graph Labelling]\label{def:GraphLabelling}
%Let $\Hash:\bool^{\leq B}\to\bool^w$ be a hash function. We define the following (vertex) labelling for $\gPoSW$ with respect to $\Hash$. For vertex $v\in \VId{n}$, the label $\ell_v \in \bool^w$ is recursively defined as
%$\begin{align}
%\ell_v:=\Hash(v,\ell_{v_1},\dots,\ell_{v_d}),  \quad \mbox{ where } (v_1,\dots,v_d) = \parent(v)  
%\ell_v:=\Hash(v,\ell_{\parent(v)}),  \quad \mbox{ where } \ell_{\parent(v)} := (\ell_{v_1},\dots,\ell_{v_k}) \mbox{ for } (v_i,\dots,v_d) = \parent(v)  
%\label{eq:compute_label}    
%\end{align}
%\end{definition}

%Note that the labelling of $\gPoSW$ can be computed by making $N = 2^{n+1}-1$ sequential queries to $H$ by computing them in an arbitrary topological order. 
We also make the following important remark.

\begin{remark} \label{rmk:path-chain}
Let $T$ be a subtree of $\gPoSW$ with a consistent labeling $\ell$. 
Then, any path $\Path = (v_0,\dots,v_r)$ of length $|P|=r$ in $T$ induces an $r$-chain $(x_0,\dots,x_r)$, where $x_i = (v_i, \ell_{v'_1},\dots,\ell_{v'_d})$ with $\{v'_1,\dots,v'_d\} = \parent(v_i)$, and where the relation $\triangleleft$ is defined as follows. $y \triangleleft x$ if and only if $x$ is of form $(v,\lab{1},\lab{2},\dots,\lab{d})$ with $v \in \VId{n}, \lab{j}\in\bool^w$, $|d|=|\parent(v)|\leq n$, and $y = \lab{j}$ for some $j$. 
%Combining this with the ``depth-robust'' property of $\gPoSW$, we know that for any vertex set $S \subseteq \VId{n}$  obtained by removing (an arbitrary number of) subtrees in $\VId{n}$, the labelling induces a $(|S|-1)$-chain in $\Hash$.
\end{remark}
%\yuhsuan{This is incorrect.}

\paragraph{\bf Simple PoSW Construction.} We are ready to describe the (non-interactive) Simple PoSW construction, which amounts to asking the prover $\Prover$ to compute the root label of $\gPoSW$ with respect to the hash function $\Hash_{\chi}$ defined by $\Hash_{\chi}(\cdot) := \Hash(\chi,\cdot)$ for a random $\chi \in \bool^w$ sampled by the verifier $\mathcal{V}$, and open the labels of the authentication paths of the challenge leaves. 
%For notational convenience, for a vertex set $S\subseteq \VId{n}$, we use $\ell_S$ to denote the set of labels $\ell_v$ for $v\in S$.

Specifically, given parameters $w, t$ and $N = 2^{n+1}-1$, and a random oracel $\Hash:\bool^{\leq B} \rightarrow \bool^w$, the Simple PoSW protocol is defined as follows.
\begin{itemize}
    \item $(\phi,\phi_{\Prover}) := \mathsf{PoSW}^\Hash(\chi,N)$: $\Prover$ computes the unique consistent labeling $\ell$ of $\gPoSW$ with respect to hash function $\Hash_{\chi}$ defined by $\Hash_{\chi}(\cdot) := \Hash(\chi,\cdot)$, and stores it in $\phi_{\Prover}$. $\Prover$ sets $\phi = \lab{\rt}$ as the root label.
    \item The opening challenge:  $\gamma := \Hash_{\chi}^{\mathsf{ChQ}}(\phi) := \bigl(\Hash_{\chi}(\phi,1),\dots,\Hash_{\chi}(\phi,d)\bigr) \in \bool^{dw}$ for sufficiently large $d$, parsed as $t$ leaves $\{v_1,\ldots,v_t\} \subseteq \leaves(\VId{n})$.%
    %\footnote{For simplicity, we will treat $\Hash_{\chi}^{\mathsf{ChQ}}(\phi)$ as one oracle query, i.e., ``charge" only one query for such an {\em challenge query}; however, we keep the superscript $\mathsf{ChQ}$ to remind that the query response is (understood as) a set of leaves. We also silently take it as given that the input to such a challenge query is specially formatted so as to distinguish it from a {\em label query}. E.g., a challenge query comes with a prefix $0$ and a label query with prefix~$1$.} 
%    \item The opening challenge:  $\gamma := \Hash_{\chi}(\phi) \in \bool^w$, parsed as $t$ leaves $\{v_1,\ldots,v_t\} \subseteq \leaves(\VId{n})$.\footnote{Note that this requires $w \geq t\cdot n$. Such requirement can be easily removed by instead setting $\gamma = \Hash_{\chi}(\phi,1)||\dots||\Hash_{\chi}(\phi,d)$ for sufficiently large $d$. We choose to just set $\gamma := \Hash_{\chi}(\phi)$ and require $w \geq t\cdot n$ here for notional simplicity.} 
    \item $\tau:= \mathsf{open}^\Hash(\chi,N,\phi_{\Prover},\gamma):$ For challenge $\gamma = \{v_1,\ldots,v_t\}$, the opening $\tau$ consists of the labels of vertices in the authentication path $\ap(v_i)$ of  $v_i$ for $i \in [t]$, i.e., $\tau = \{\ell_{\ap(v_i)}\}_{i \in [t]}$.
    \item $\mathsf{verify}^\Hash(\chi, N, \phi, \gamma, \tau)$:  $\mathcal{V}$ verifies if the ancestors of every $v_i$ are consistently labeled by $\tau$. Specifically, for each $i \in [t]$, $\mathcal{V}$ checks if $\ell_{u} = H_{\chi}(u, \ell_{\parent(u)})$ for all $u \in \anc(v_i)$.
    %\footnote{Note that by the construction of $\gPoSW$, $\parent(u) \subseteq \ap(v_i)$ for every $i \in [t]$ and $0\leq j \leq n$ so $\mathcal{V}$ can check the consistency given the labels in $\tau$.} 
    $\mathcal{V}$ outputs $\mathsf{accept}$ iff all the consistency checks pass. 
\end{itemize}

%In the language of Definition~\ref{def:GraphLabelling}, the prover needs to provide consistent labellings for the subtrees $\ap(v_i)$. 

Note that since we consider the non-interactive version of Simple PoSW after applying the Fair-Shamir transformation, the random oracle $\Hash$ is used to compute both the labels (as $\Hash_\chi(v,\ell_{\parent(v)})$) and the challenge (as $\Hash_\chi^{\mathsf{ChQ}}(\phi)$). We silently assume that the respective inputs are specially formatted so as to distinguish a {\em label query} from a {\em challenge query}. E.g., a label query comes with a prefix $0$ and a challenge query with prefix~$1$. We then denote the set of inputs for label and challenge queries by $\mathsf{LbQ}$ and $\mathsf{ChQ} \subseteq \bool^{\leq B}$, respectively.
Also, for simplicity, we will treat $\Hash_{\chi}^{\mathsf{ChQ}}(\phi)$ as {\em one} oracle query, i.e., ``charge" only one query for a challenge query; however, we keep the superscript $\mathsf{ChQ}$ to remind that the query response is (understood as) a set of leaves.

\paragraph{\bf Classical Security Analysis of Simple PoSW.} Before presenting our proof of post-quantum security for Simple PoSW, we first review the classical security analysis in~\cite{cohen2018simple}. For simplicity, here we consider the original (interactive) Simple PoSW (i.e., $\Prover$ first sends $\phi$, receives random $\gamma$ from $\mathcal{V}$, and then sends $\tau$ to $\mathcal{V}$). Also, to start with, we assume that $\Prover$ does not make further oracle queries after sending $\phi$. % (this is also the setting analyzed in~\cite{cohen2018simple}). 
We review the argument of \cite{cohen2018simple} for bounding the probability that a $k$-parallel $q$-query classical oracle algorithm $\mathcal{A}$ with $q < N$ makes $\mathcal{V}$ accept, using the terminology we introduced in Section~\ref{sec:CTC}.

Let $D:\bool^{\leq B} \rightarrow \zo^w \cup \{\bot\}$ be the database at the point that $\mathcal{A}$ sends $\phi$ to $\mathcal{V}$ (after having made the $q$ $k$-parallel queries). Following the argument in Section~\ref{sec:CTC}, we can bound the success probability of $\mathcal{A}$ by bounding the probability that a random challenge $\gamma = \{v_i\}_{i\in[t]}$ can be opened based on the information in the database $D$. 
As argued in \ifnum\submission=0 Section~\ref{sec:CTC} \else {\color{red} one of the appendices} \fi, the probability that the database $D$ contains collisions, or a $(q+1)$-chain with respect to the relation defined in Remark~\ref{rmk:path-chain}, is small.
%; specifically, at most $O((k^2q^2+nkq^2)/2^w)$. 
Thus, by a union bound, we can assume that $D$ contains no collisions nor $(q+1)$-chains. 

Next, given the database $D$ and the ``commitment" $\phi$, claimed to be the root label $\ell_\rt$, we need to analyze the set of leaves $v$ that $\mathcal{A}$ can open. %, i.e., for which he can provide a consistently labeled authentication path $\ap(v)$. 
One of the key observations in~\cite{cohen2018simple} is that, for a database $D$ with no collisions, there exists a maximal subtree $T$ of $\gPoSW$ that contains $\rt$ and admits a consistent labeling $\ell$ with $\ell_\rt = \phi$. 
As observed in~\cite{cohen2018simple}, this subtree $T$ then contains all leaves that one can open given $D$. Thus, $\mathcal{A}$ can correctly answer a challenge $\gamma = \{v_1,\ldots,v_t\}$ if $\gamma \subseteq \leaves(T)$, while otherwise it is unlikely that he succeeds. %\yuhsuan{The definition of $T$ may be misuderstood.}

The subtree $T$, together with the labeling $\ell$ of $T$, can be extracted using $\ext{D}{n}{\phi}$, described in Algorithm~\ref{algo:extract} in the Appendix~\ref{app:Extract}. 
Roughly speaking, starting with $T := \{\rt\}$, consider $v := \rt$ and $\ell_\rt := \phi$, and add $\lch{v}$ and $\rch{v}$ to $T$ if (and only if) there exist \smash{$\ell_{\lch{v}}$} and $\ell_{\rch{v}}$ such that $\ell_v = D\bigl(v,\ell_{\lch{v}},\ell_{\rch{v}}\bigr)$, and repeat inductively with the newly added elements in $T$. In the end, for the leaves $v \in T$ check if $\ell_{v} = D(v,\ell_{\parent(v)})$ and remove $v$ from $T$ if this is not the case; we note here that $v \in \leaves(T) \Rightarrow \parent(v) \subseteq T$.  

In summary:

\begin{lemma}\label{lem:extract}
Let $D:\bool^{\leq B} \rightarrow \zo^w \cup \{\bot\}$ be a database with no collisions (beyond $\bot$). Then, for any $\phi \in \zo^w$, the subtree $T$ and the labeling $\ell$ produced by $\ext{D}{n}{\phi}$ are such that $\ell$ is a consistent labeling of $T$ with respect to $D$, having root label $\ell_\rt = \phi$. Furthermore, for any leave $v$ of $V_n$, if $v \in T$ then $\ell_{u} = D(u, \ell_{\parent(u)})$ for all $u \in \anc(v_i)$, and if $v \not\in T$ then there exists no labeling $\ell'$ with $\ell'_\rt = \phi$ and $\ell'_{u} = D(u, \ell'_{\parent(u)})$ for all $u \in \anc(v_i)$. 
%$v$ admits a consistently labeled authentication path $\ap(v)$ with respect to $D$ if and only if $v$ is in $T$. 
\end{lemma}

%Next, we need to analyze the set of leaf challenges $v$ that $\mathcal{A}$ can open consistent labels for its authentication path $\ap(v)$ with respect to $\phi = \lab{\rt}$ based on the database $D$. One of the key observations in~\cite{cohen2018simple} is that, for database $D$ with no collisions, there exists an extraction procedure $\ext{D}{n}{\ell_\rt}$ that extracts a ``consistent subtree'' $T$ containing $\rt$ such that it contains all leaves that one can open given $D$. (See Algorithm~\ref{algo:extract} in Supplementary Material for a formal description of $\ext{D}{n}{\ell_\rt}$.) 

%Roughly speaking, $\ext{D}{n}{\ell_\rt}$ first extracts all labels that can be naturally defined recursively from $\lab{\rt}$ and $D$ by $D(v,\lab{\parent(v)}) = \lab{v}$ for non-leaf vertex $v\in \VId{n}$, which is uniquely defined if $D$ contains no collisions. Let $T$ be the subtree with extracted labels. Then $ \ext{D}{n}{\ell_\rt}$ checks the consistency of the extracted labels for leaves $v\in \leaves(T)$ by $D(v,\lab{\parent(v)}) = \lab{v}$, and removes the inconsistent leaves. Then it outputs the remaining subtree $T$. As observed in~\cite{cohen2018simple}, $T$ contains all leaves that one can open given $D$. Thus, $\mathcal{A}$ can open labels for a challenge $\gamma = \{v_i\}_{i\in[t]}$ iff $\gamma \subseteq \leaves(T)$. %\yuhsuan{The definition of $T$ may be misuderstood.}

The last step is to bound the number of leaves in $T$. Another key argument in~\cite{cohen2018simple} uses a certain ``depth-robust'' property of $\gPoSW$ to show that for any subtree $T \subseteq \VId{n}$ with $\rt \in T$, there exists a path $\Path$ in $T$ with length $|\Path| \geq 2 \cdot |\leaves(T)| -2$. 
Recall we argued above that the graph labeling of a path $\Path \in \gPoSW$ induces a $|\Path|$-chain in $\Hash$. The same argument applies here to show that there exists a $|\Path|$-chain in $D$ since the extracted labels in $P\subseteq T$ are consistent (i.e., satisfying $D(v,\lab{\parent(v)}) = \lab{v}$). Combining these with the assumption that $D$ contains no $q+1$-chain, we have $|\leaves(T)| \leq (q+2)/2$. Therefore, the probability that $\mathcal{A}$ can open labels for a random challenge $\gamma = \{v_i\}_{i\in[t]}$ is at most 
$$\left(\frac{|\leaves(T)|}{2^n}\right)^t \leq \left(\frac{q+2}{2^{n+1}}\right)^t.$$

\begin{lemma}\label{lem:numberofleaves}
Let $D:\bool^{\leq B} \rightarrow \zo^w \cup \{\bot\}$ be a database with no $(q+1)$-chain. Let $T$ be a subtree of $\gPoSW$ admitting a consistent labeling with respect to $D$. Then, $|\leaves(T)| \leq (q+2)/2$. 
\end{lemma}

%Putting things together, the probability that a $k$-parallel $q$-query classical oracle algorithm $\mathcal{A}$  makes $\mathcal{V}$ accept can be upper bounded by
%$$O\left(\frac{k^2 q^2}{2^w}+\frac{nkq^2}{2^w}+\left(\frac{q+2}{2^{n+1}}\right)^t \right).$$

%The above assumes that $\mathlca{A}$ cannot make additional queries after sending $\phi$. 
Finally, we briefly discuss here how to handle the case that $\mathcal{A}$ can make additional queries after sending $\phi$, as a similar argument is required in the analysis of the non-interactive Simple PoSW in the next section. As before, let $D$ be the database right after $\mathcal{A}$ has sent $\phi = \lab{\rt}$, but now $\mathcal{A}$ can make additional queries after seeing $\gamma$, which adds new entries to $D$ and may help $\mathcal{A}$ to open labels for more challenges $\gamma$.

The main observation to analyze whether additional queries are helpful is as follows.
Recall that $T$ contains all leaves $v$ that admit consistently labeled ancestors. Thus for the additional queries to be helpful, they must enlarge the extracted subtree $T$. More precisely, let $D'$ be the database after the additional queries and let $T'$ and $\ell'$ be extracted by $\ext{D'}{n}{\phi}$. It must be that $T \subsetneq T'$ and $\ell'|_T = \ell$, and there must exist $x$ with $D(x) = \bot$ while $D'(x) = \ell_v$ for some $v \in T$. This happens with probability at most $O(qk/2^w)$ for each query since $\ell$ has support size at most $O(qk)$. We capture the above crucial observation by means of the following formal statement, which, in this form, will then be useful in the security proof against quantum attacks. 

%\begin{lemma}\label{lem:newpath}
%Let $D:\bool^{\leq B} \rightarrow \zo^w \cup \{\bot\}$ be a database with no collisions (beyond $\bot$). Let $\phi \in \zo^w$ and $(T,\ell) =\ext{D}{n}{\phi}$. Furthermore, let $D' = D[\bfx \!\mapsto \! \bfu]$ and $(T',\ell') =\ext{D'}{n}{\phi}$. Then the following holds for any leave $v$ of $V_n$. If $v \not\in T$ but $v \in T'$ then there exists $i \in \{1,\ldots,k\}$ and $v \in T$ with $D(x_i) \neq D'(x_i) = \ell_v$. 
%\end{lemma}

\begin{lemma}\label{lem:newpath}
Let $D:\bool^{\leq B} \rightarrow \zo^w \cup \{\bot\}$ be a database with no collisions (beyond $\bot$). Let $\phi \in \zo^w$ and $(T,\ell) =\ext{D}{n}{\phi}$. Furthermore, let $D' = D[\bfx \!\mapsto \! \bfu]$ and $(T',\ell') =\ext{D'}{n}{\phi}$, and let $v$ be a leave of $V_n$. If $v \in T' \setminus T$ then there exist $j \in \{1,\ldots,k\}$ and $z \in \anc(v)$ so that %\smash{$x_j=(z_i,\ell_{\parent(z_i)})$} and 
$D(x_j) \neq D'(x_j) = \ell'_z$. 
\end{lemma}

\begin{proof}
Given that $v \in T'$, the labeling $\ell'$ labels the ancestors of $v$ consistently with respect to $D'$, i.e., $\ell'_{z} = D'(z,\ell'_{\parent(z)})$ for all $z \in \anc(v)$. 
%Furthermore, $D_{\bfx,\bfu}^\mathsf{ChQ}(\lab\rt)\neq\bot$. Therefore, $\ell_{z_i} \in \lsupp(D_{\bfx,\bfu}) \subseteq \lsupp(D_{\bfx,\bf 0})$ for $0\leq i\leq n$. 
On another hand, since $v$ is not in $T$, %yet $D_{\bfx,\bfr}^\mathsf{ChQ}(\lab\rt) = D_{\bfx,\bfu}^\mathsf{ChQ}(\lab\rt)$, 
it must be that $\ell'$ does {\em not} label the ancestors of $v$ consistently with respect to $D$, i.e., there must exist $z \in \anc(v)$ such that $D(z,\ell'_{\parent(z)}) \neq \ell'_{z} = D'(z,\ell'_{\parent(z)})$. Since $D$ and $D'$ differ only within $\bfx$, there exists $j \in \{1,\ldots,k\}$ with $x_j=(z,\ell'_{\parent(z)})$. 
%Therefore, there must exist $j \in \{1,\ldots,k\}$ and $z \in \anc(v)$ such that 
%%at $z_i$ for some $0\leq i\leq n$. Since $D_{\bfx,\bfr}$ and $D_{\bfx,\bfu}$ are different at $(z_i,\ell_{\parent(z_i)})$, there is some $j$ such that 
%\smash{$x_j=(z,\ell'_{\parent(z)})$} and $D'(x_j) = \ell'_{z} \neq D(x_j)$. 
\end{proof}

\subsection{Post-Quantum Security of Simple PoSW}

In this section, we prove post-quantum security of the (non-interactive) Simple PoSW protocol. As we shall see, relying on the framework we developed in Section~\ref{sec:QTC_framework}, the proof uses {\em purely classical reasoning} only, and somewhat resembles the arguments in the classical analysis.  %Formally, we prove the following theorem.

\begin{theorem}[Post-Quantum Simple PoSW Security]
\label{thm:posw_security}
Consider the Simple PoSW protocol with parameters $w, t$ and $N = 2^{n+1}-1$ with $w\ge tn$. Let $\tilde{\Prover}$ be a $k$-parallel $q$-query quantum oracle algorithm acting as a prover. The probability $p$ that  $\tilde{\Prover}$ can make the verifier $\mathcal{V}$ accept is at most 
%If $\tilde{P}$ makes at most $q$ sequential quantum queries to $\Hash$ with $k$ parallel branches per query after receiving $\chi$, then $\V$ will output $\mathsf{accept}$ with probability at most
$$
p = O\left(k^2q^2\left(\frac{q+2}{2^{n+1}}\right)^t+\frac{k^3q^3n}{2^w}+\frac{tn}{2^w}\right). 
$$
\end{theorem}

The first step towards the proof is to invoke Corollary~\ref{cor:zha} (using the notation from Theorem~\ref{thm:QTCBound}), which, in the case here, bounds the success probability $p$ of a dishonest prover $\tilde{\Prover}$ by 
$$
\sqrt{p} \leq \qQTC{{\boldsymbol \bot}}{\P^R}{k}{q}+\sqrt{\frac{t\cdot (n+1)+1}{2^w}} \, ,
$$
where $R$ is the relation that checks correctness of $\tilde{\Prover}$'s output according to the scheme. In the following, we write $\suc:=\P^R$ and $\fail=\neg\suc$. 
Also, recall the database properties $\col$, $\SIZE[s]$ and $\chain^s$ defined previously, where the latter is with respect to the hash chain relation $\triangleleft$ considered in Remark~\ref{rmk:path-chain}. 
By the properties of (the subtree extracted with) $\ext{D}{n}{\cdot}$, we have  
%\serge{Good catch! I forgot about the collisions. }
\begin{equation}\label{eq:suc}
\suc \setminus \col = \bigl\{D \in\nocol \,\big|\, \exists \, \ell_\rt \in \bool^w \text{ s.t. } D^\mathsf{ChQ}(\ell_\rt) \subseteq \ext{D}{n}{\ell_\rt} \bigr\} \, .
\end{equation}

To bound $\qQTC{{\boldsymbol \bot}}{\P^R}{k}{q} = \qQTC{{\boldsymbol \bot}}{\suc}{k}{q}$, we consider database properties $\P_0,\dots,\P_q$ with $\P_0 = {\boldsymbol \bot}$ and $\P_s =\suc\cup\col\cup\chain^{s+1}$ for $1\leq s \leq q$. 
Using Lemma~\ref{lem:tech}, Remark~\ref{rem:size_condition} and Corollary~\ref{cor:subset}, 
$$
\qQTC{{\boldsymbol \bot}}{\suc}{k}{q} %\leq \qQTC{{\boldsymbol \bot}}{\P_q}{k}{q} 
\leq \sum_{1\leq s\leq q} \QTC{\SIZE[k(s-1)]\backslash \P_{s-1}}{\P_s}{k} \, .
$$
Thus, the proof of Theorem~\ref{thm:posw_security} follows immediately from the following bound on the considered transition capacity. 

\begin{proposition}
For integers $0 \leq s \leq q$, and for the database properties $\P_0,\dots,\P_q$ as defined above 
$$
\QTC{\SIZE[k(s-1)]\backslash \P_{s-1}}{\P_s}{k} \leq 4ek\sqrt{10\frac{q+1}{2^w}} + 3ek\sqrt{\frac{10kqn}{2^w}} + ek\sqrt{10\left(\frac{q+2}{2^{n+1}}\right)^t} \, .
$$
\end{proposition}
Intuitively, we consider the transition from a database that is bounded in size, has no collision, no $s$-chain and does not have a successful output for $\tilde{\Prover}$, into one that contains a collision or an $(s+1)$-chain or a successful output for $\tilde{\Prover}$. 

\begin{proof}
%our goal is to bound the transition capacities \smash{$\QTC{\SIZE[k(s-1)]\backslash\P_{s-1}}{\P_s}{k}$}. 
%By definition, $\neg\P_{s-1} = \fail \cap \neg \col \cap \neg \chain^s$.  
Applying Corollary~\ref{cor:parallel_condition} with $h:=2$, $X_1:=\mathsf{LbQ}$ and $X_2:=\mathsf{ChQ}$, and with $\P_0,\P_1,\P_2$ and $\Q$ in Corollary~\ref{cor:parallel_condition} set to% 
\footnote{Note that we have slight collision of notation here: $\P_0,\P_1,\P_2$ correspond to the choice of properties for applying Corollary~\ref{cor:parallel_condition}, and should not be confused with $\P_s$ with $s$ set to $0,1,2$, respectively.  }
$$
\neg \P_0 := \SIZE[k(s-1)]\setminus\P_{s-1} \: , \quad
\P_1:=\suc \: , \quad  
\P_2:=\suc\cup\col\cup\chain^{s+1} = \P_s \quad\text{and}\quad 
\Q:=\neg(\col\cup\chain^{s+1}) \, ,
$$
we can bound $\QTC{\SIZE[k(s-1)]\backslash\P_{s-1}}{\P_s}{k} = \QTC{\neg \P_0}{\P_2}{k}$ by
\begin{align*}
%\QTC{&\SIZE[k(s-1)]\backslash\P_{s-1}}{\P_s}{k} = \QTC{\neg \P_0}{\P_2}{k} \\
&\leq \condQTC{\neg\P_0}{\neg\Q}{k}{\mathsf{LbQ}} + \condQTC{\neg\P_0}{\neg\Q}{k}{\mathsf{ChQ} \cup\mathsf{LbQ} } 
+ \condQTC{\Q \backslash \P_0}{\Q \cap \P_1}{k}{\mathsf{LbQ}}  + \condQTC{\Q \backslash \P_1}{\Q \cap \P_2}{k}{\mathsf{ChQ}}  \\
&\leq 2 \QTC{\neg\P_0}{\neg\Q}{k} + \condQTC{\Q \backslash \P_0}{\Q \cap \P_1}{k}{\mathsf{LbQ}}  + \condQTC{\Q \backslash \P_1}{\Q \cap \P_2}{k}{\mathsf{ChQ}} \\
&= 2 \QTC{\SIZE[k(s-1)]\backslash\P_{s-1}}{\col\cup\chain^{s+1}}{k} + \condQTC{\SIZE[k(s-1)]\backslash\P_{s-1}\backslash\col\backslash\chain^{s+1}}{\suc\backslash\col\backslash\chain^{s+1}}{k}{\mathsf{LbQ}}  \\
&\qquad+ \condQTC{\neg(\suc\cup\col\cup\chain^{s+1})}{\suc\backslash\col\backslash\chain^{s+1}}{k}{\mathsf{ChQ}}   \\
&\leq 2 \QTC{\SIZE[k(s-1)]\backslash\P_{s-1}}{\col\cup\chain^{s+1}}{k}  + \condQTC{\SIZE[k(s-1)]\backslash\P_{s-1}}{\suc\backslash\col}{k}{\mathsf{LbQ}}  + \condQTC{\neg \P_s}{\suc\backslash\col}{k}{\mathsf{ChQ}} \, . %  \\
%&\leq 2 \QTC{\SIZE[k(s-1)]\backslash\P_{s-1}}{\col\cup\chain^{s+1}}{k}  \\
%&\qquad+ \condQTC{\SIZE[k(q-1)]\backslash\P_{q-1}}{\suc\backslash\col\backslash\chain^{q+1}}{k}{\mathsf{LbQ}} && (F) \\
%&\qquad+ \condQTC{\neg\P_{q}}{\suc\backslash\col\backslash\chain^{q+1}}{k}{\mathsf{ChQ}} && (E)
\end{align*}

By means of Lemma~\ref{lem:shrink} (and Corollary~\ref{cor:subset}), and recalling that $\P_{s-1} = \suc \cup \col \cup \chain^s$, the first capacity in the term can be controlled as
%\begin{align*}
%    \QTC{\SIZE[k(s-1)]&\backslash\P_{s-1}}{\neg\SIZE[ks]\cup\col\cup\chain^{s+1}}{k} \\ 
%    &\leq \QTC{\SIZE[k(s-1)]\backslash\P_{s-1}}{\neg\SIZE[ks]}{k} + \QTC{\SIZE[k(s-1)]\backslash\P_{s-1}}{\col}{k} + \QTC{\SIZE[k(s-1)]\backslash\P_{s-1}}{\chain^{s+1}}{k} \\ 
%    &\leq \QTC{\SIZE[k(s-1)]}{\neg\SIZE[ks]}{k} + \QTC{\SIZE[k(s-1)]\backslash\col}{\col}{k} + \QTC{\SIZE[k(s-1)]\backslash\chain^s}{\chain^{s+1}}{k} \\ 
%    &\leq 2ek\sqrt{10\frac{q+1}{2^w}}+ek\sqrt{\frac{10kqn}{2^w}}
%\end{align*}
\begin{align*}
    \QTC{\SIZE[k(s-1)]\backslash\P_{s-1}}{\col\cup\chain^{s+1}}{k} 
    &\leq \QTC{\SIZE[k(s-1)]\backslash\P_{s-1}}{\col}{k} + \QTC{\SIZE[k(s-1)]\backslash\P_{s-1}}{\chain^{s+1}}{k} \\ 
    &\leq \QTC{\SIZE[k(s-1)]\backslash\col}{\col}{k} + \QTC{\SIZE[k(s-1)]\backslash\chain^s}{\chain^{s+1}}{k} \\ 
    &\leq 2ek\sqrt{10\frac{q+1}{2^w}}+ek\sqrt{\frac{10kqn}{2^w}}
\end{align*}
using earlier derived bounds. It remains to bound the remaining two capacities appropriately, which we do below. 
\end{proof}

Intuitively, $\condQTC{\neg \P_s}{\suc\backslash\col}{k}{\mathsf{ChQ}}$ captures the likelihood that a database $D \not\in \suc$ (and with no collision and chain) is turned into one that does satisfy $\suc$ by (re)defining $D$ on $k$ values that correspond to challenge queries. For this to happen, one of the newly defined function values of $D$, corresponding to a challenge query and thus specifying a set of leaves, must ``hit" the set of leaves that can be answered, which is bounded in size. 

%Intuitively, $(E)$ corresponds the transition that a challenge query of the form $\Hash_{\chi}(\phi)$ for $\phi = \lab{\rt}\in\bool^w$ making a database $D \in \SIZE[k(s-1)]\backslash\P_{s-1}$ to ``become success.''  As we argued in the classical analysis, this happens when $\Hash_{\chi}(\lab{\rt}) \subseteq \leaves(\ext{D}{n}{\ell_\rt})$. This allows us to define a family of $1$-local properties $\L^{\bfx, D}_j$ that $k$-non-uniformly weakly recognize the transition in $(E)$, and use Theorem~\ref{thm:tricky_restricted} to bound the capacity $(E)$. 
%We state the bound as the following lemma, and present its proof in Appendix~\ref{app:proof_ChQ}. 

\begin{lemma}\label{lem:QTC_ChQ}
For any positive integer $q$, it holds that  
$\condQTC{\neg\P_{q}}{\suc\backslash\col}{k}{\mathsf{ChQ}} \leq ek \cdot \sqrt{10\left(\frac{q+2}{2^{n+1}}\right)^t}$. 
\end{lemma}

\begin{proof}
For convenience, we will denote $D[\bfx\!\mapsto\!\bfy]$ by $D_{\bfx,\bfy}$. 
In order to bound the above capacity, we define $1$-local properties $\L_j^{\bfx, D}$ and show that $\L_j^{\bfx,D}$ (weakly) recognize the considered transition (with input restricted to $\mathsf{ChQ}$). 

For any $D$ and $\bfx = (\ell_\rt^1,\ldots,\ell_\rt^k) \in \mathsf{ChQ}^k$, we set
$$
\L^{\bfx, D}_j:=\left\{D_\circ \in D|^{\bfx}\,\middle|\, D_\circ^\mathsf{ChQ}(x_j)\subseteq\leaves\Bigl(\ext{D_{\bfx,\perp}}{n}{\ell_\rt^j}\Bigr)\right\}
$$ 
Suppose $D_{\bfx,\bfr}\in\neg\P_q=\fail\setminus\col\setminus\chain^{q+1}$ but $D_{\bfx,\bfu}\in\suc\setminus\col$. 
Thus, by (\ref{eq:suc}), 
%As $D_{\bfx,\bfu}\in\suc$, 
there exists $\lab\rt\in\bool^w$ with 
\begin{equation}{\label{eq:Cl+}}
    D_{\bfx,\bfu}^\mathsf{ChQ}(\lab\rt)\subseteq \mathsf{leaves}\Bigl(\ext{D_{\bfx,\bfu}}{n}{\lab\rt}\Bigr) \, ,
\end{equation}
%On another hand, since $D_{\bfx,\bfr} \in\fail$, we must have
while
\begin{equation}{\label{eq:Cl-}}
    D_{\bfx,\bfr}^\mathsf{ChQ}(\lab\rt)\not\subseteq \leaves\Bigl(\ext{D_{\bfx, \bfr}}{n}{\lab\rt}\Bigr).
\end{equation}
Since the output of the extraction procedure $\ext{D}{n}{\cdot}$ only depends on those function values of $D$ that correspond to {\em label} queries ($\bfx$ here consists of challenge queries), 
%Because the output of procedure $\ext{D}{n}{\phi}$ does not depend on the the database entries \serge{It should be ``database entr{\em y}" here, right?} $D(\phi)$, 
we have
$$
\ext{D_{\bfx, \bfr}}{n}{\lab\rt} = \ext{D_{\bfx, \perp}}{n}{\lab\rt} = \ext{D_{\bfx, \bfu}}{n}{\lab\rt}.
$$
If $\lab\rt$ is different from all $\lab\rt^j$, then equations $(\ref{eq:Cl+})$ and $(\ref{eq:Cl-})$ contradict. So there is some $j$ such that $\lab\rt^j = \lab\rt$. 
Equations $(\ref{eq:Cl+})$ and $(\ref{eq:Cl-})$ thus become
\begin{align*}
    u_j\subseteq \mathsf{leaves}\Bigl(\ext{D_{\bfx, \perp}}{n}{\lab\rt}\Bigr) \qquad\text{and}\qquad
    r_j\not\subseteq \mathsf{leaves}\Bigl(\ext{D_{\bfx, \perp}}{n}{\lab\rt}\Bigr),
\end{align*}
understanding that $u_j$ and $r_j$ represent lists/sets of $t$ (challenge) leaves.
Hence $r_j\neq u_j$. This concludes that $\L_j^{\bfx,D}$ indeed weakly recognizes the considered database transition.

We note that, for each $\bfx\in\mathsf{ChQ}^k$ and $D\in\fail\setminus\col\setminus\chain^{q+1}$, since the longest hash chain in $D$ is of length no more than $q$ and $T:= \ext{D_{\bfx,\perp}}{n}{\ell_\rt^j}$ admits a consistent labeling (Lemma~\ref{lem:extract}), it follows from  Lemma~\ref{lem:numberofleaves} that
$$
\Big|\mathsf{leaves}\Bigl(\ext{D_{\bfx,\perp}}{n}{\ell_\rt^j}\Bigr)\Big|\leq\frac{q+2}{2} \, .
$$
Therefore,
$$
P\bigl[U\in\L_j^{\bfx,D}\bigr] \leq \left(\frac{\leaves\bigl(\ext{D_{\bfx,\perp}}{n}{\lab\rt^j}\bigr)}{2^n}\right)^t\leq\left(\frac{q+2}{2^{n+1}}\right)^t \, ,
$$
and so the claimed bound follows by applying Theorem~\ref{thm:tricky_restricted}. 
%\begin{align*}
%\condQTC{\neg\P_q}{\suc\backslash\col}{k}{\mathsf{ChQ}}
%\leq ek\sqrt{10\left(\frac{q+2}{2^{n+1}}\right)^t} \, .    
%\end{align*}
\end{proof}

Similarly here, the intuition is that $\condQTC{\neg \P_s}{\suc\backslash\col}{k}{\mathsf{LbQ}}$ captures the likelihood that a database $D \not\in \suc$ (and with no collision and chain) is tuned into one that does satisfy $\suc$ by (re)defining $D$ on $k$ values that correspond to label queries. For this to happen, one of the newly defined function values of $D$, corresponding to a label, must ``match up" with the other labels. 

%Similarly, $(F)$ corresponds to the transition that a label query of the form $\Hash_\chi(v,\ell_{\parent(v)})$ making a database $D \in \SIZE[k(s-1)]\backslash\P_{s-1}$ to ``become success.''   This corresponds to the case that $\tilde{\Prover}$ making additional queries after sending $\phi = \lab{\rt}$ we analyzed for the interactive version in the end of Section~\ref{sec:classical_PoSW} (but here $\tilde{\Prover}$ is allowed to try multiple challenge queries $\phi$), and happens only if the additional queries enlarge some extracted subtree in the database. As above, we use the observation to define a family of $1$-local properties $\L^{\bfx, D}_j$ that $k$-non-uniformly weakly recognize the transition in $(F)$, and use Theorem~\ref{thm:tricky_restricted} to bound the capacity $(F)$. We state the bound as the following lemma, and present its proof in Appendix~\ref{app:proof_LbQ}.

\begin{lemma}\label{lem:QTC_LbQ}
For any positive integer $q$, it holds that $\condQTC{\SIZE[k(q-1)]\backslash\P_{q-1}}{\suc\backslash\col}{k}{\mathsf{LbQ}} \leq  ek \sqrt{\frac{10nkq}{2^w}}.$
\end{lemma}

\begin{proof}
Define the notion of labeling support $\lsupp(D)$ of a database $D\in\DB$ as follows. 
%$$
%\lsupp(D):=\left\{\ell_1,\dots,\ell_d\in\bool^w: \substack{
%    v\in V_n, 0\leq d\leq n\\
%    D(v,\ell_1,\dots,\ell_d)\neq\bot
%}\right\}\cup\left\{
%    \lab\rt\in\bool^w:D(\lab\rt)\neq\bot
%\right\}.
%$$
$$
\lsupp(D) := \left\{ \lambda \in \bool^w \,\middle|\, \!\begin{array}{l} \exists \, 0 \!\leq\! i \!\leq\! d \!\leq\! n, v \!\in\! V_n, \ell_1,\ldots,\ell_d \!\in\! \bool^w \\ \text{s.t. } D(v,\ell_1,\dots,\ell_{i-1},\lambda,\ell_{i+1},\ldots\ell_d) \neq \bot \end{array} \!\! \right\} \cup \left\{
    \lab\rt\in\bool^w \,\middle|\, D^\textsf{ChQ}(\lab\rt)\neq\bot \right\}  .
$$
We note that since $\lsupp$ is defined only in terms of where $D$ is defined, but does not depend on the actual function values (beyond being non-$\bot$), $\lsupp(D) \subseteq \lsupp(D_{\bfx,\bf 0})$ for any $\bfx \in \Xcal^k$, where ${\bf 0} \in \bool^k$ is the all-$0$ string. 

In order to bound above capacity, we define $1$-local properties and show that they (weakly) recognize the considered transition (with input restricted to $\mathsf{LbQ}$). For any $D$ and $\bfx \in \mathsf{LbQ}^k$, consider the local properties 
$$
\L_j^{\bfx,D}:=\left\{D_\circ \in D|^{\bfx} \,\big|\, D_\circ(x_j)\in\lsupp(D_{\bfx,\bf 0})\right\} \, .
$$
Let $D_{\bfx,\bfr}\in\neg\P_{q-1}=\fail\setminus\col\setminus\chain^{q}$ yet $D_{\bfx,\bfu}\in\suc\setminus\col$.
By (\ref{eq:suc}), there exists $\lab\rt$ so that 
%and $\lab\rt$ be some successful label where 
$D_{\bfx,\bfu}^\textsf{ChQ}(\lab\rt)\subseteq\ext{D_{\bfx,\bfu}}{n}{\lab\rt}$, while, on the other hand, there exists some $v\in D_{\bfx,\bfr}^\textsf{ChQ}(\lab\rt)\setminus\leaves\bigl(\ext{D_{\bfx, \bfr}}{n}{\lab\rt}\bigr)$. 
Given that here $\bfx\in\mathsf{LbQ}^k$, we have $D_{\bfx,\bfr}(\lab\rt) =  D_{\bfx,\bfu}(\lab\rt)$, and thus, by $(\ref{eq:Cl+})$, we have
$$
v\in \mathsf{leaves}\Bigl(\ext{D_{\bfx, \bfu}}{n}{\lab\rt}\Bigr)\setminus\mathsf{leaves}\Bigl(\ext{D_{\bfx, \bfr}}{n}{\lab\rt}\Bigr).
$$

Writing $\ell'$ for the labeling extracted by $\ext{D_{\bfx, \bfu}}{n}{\lab\rt}$, it then follows from Lemma~\ref{lem:newpath} that there exist $j \in \{1,\ldots,k\}$ and $z \in \anc(v)$ such that 
%\smash{$x_j=(z,\ell_{\parent(z)})$} and 
$u_j = D_{\bfx,\bfu}(x_j) = \ell'_{z} \neq D_{\bfx,\bfr}(x_j) = r_j$. Furthermore, since $D_{\bfx,\bfu}^\mathsf{ChQ}(\ell'_{z}) = D_{\bfx,\bfu}^\mathsf{ChQ}(\ell_\rt) \neq\bot$ in case $z = \rt$, and $\ell'_z$ is part of the input that is mapped to \smash{$\ell'_{\mpar(z)}$} under $D_{\bfx,\bfu}$ in all other cases, we also have $u_j = \ell'_{z} \in \lsupp(D_{\bfx,\bfu}) \subseteq \lsupp(D_{\bfx,\bf 0})$. 
Therefore, the local properties \smash{$\L_j^{\bfx,D}$} do indeed weakly recognize the considered transition for input restricted to $\mathsf{LbQ}$.

For $D\in\SIZE[k(q-1)]\setminus\P_{q-1}$, since there are only $k(q-1)$ entries in $D$, we have
$$
P[U\in\L_j^{\bfx,D}]\leq \frac{\left|\lsupp(D_{\bfx,\bf 0})\right|}{2^w}\leq
\frac{nkq}{2^w}. \, ,
$$
and thus the claimed bound follows from applying Theorem~\ref{thm:tricky_restricted}. 
%$$
%(F) = \condQTC{\SIZE[k(q-1)]\backslash\P_{q-1}}{\suc\backslash\col}{k}{\mathsf{LbQ}}\leq ek \sqrt{\frac{10nkq}{2^w}}.
%$$
\end{proof}

\section*{Acknowledgements}

We thank Jeremiah Blocki, Seunghoon Lee, and Samson Zhou for the open discussion regarding their work~\cite{BlockiLZ20}, which achieves comparable results for the hash-chain problem and the Simple PoSW scheme.

{\small

\bibliography{citation}
\bibliographystyle{plain}

}

\appendix

\section{Efficient Simulation of the Compressed Oracle}\label{app:Efficiency}

\def\Upd{U\!pd}
\def\upd{u\hspace{-0.1ex}pd}

%\serge{I've added this for completeness.}

In order to complete our exposition of the compressed oracle, we show here another aspect of the technique, which is not relevant in our context but an important feature in other applications: similarly to the classical lazy-sampling technique, the evolution of the compressed oracle can be {\em efficiently} computed, {\em and} useful information can be {\em efficiently} extracted from the compressed oracle. 

For concreteness, we assume here that $\Ycal = \{0,1\}^m$. This in particular means that $\hat\Ycal = \Ycal$, and that there is a designated and efficiently computable {\em quantum Fourier transform} $\qft: \ket{y} \mapsto \ket{\hat y} = H^{\otimes m} \ket{y}$. 
This then also means that $\DB = \hat\DB$, but we still distinguish between $\ket{D} = \bigotimes_x \ket{D(x)}$ and $\ket{\hat D} = \bigotimes_x \qft\ket{D(x)}$ for any $D \in \DB$. Additionally, we assume that $\Xcal$ comes with an efficiently computable total order, say $\Xcal = \{0,1\}^n$. 

Consider the classical encoding function $Enc: \DB \to \frak{L} := (({\cal X} \times {\cal Y}) \cup \{\bot\})^{|\Xcal|}$ that maps $D \in \DB$ to the list $L = \big[(x_1,y_1),\ldots,(x_s,y_s),\bot,\ldots,\bot\big]$ of pairs $(x_i,y_i)$ for which $y_i = D(x_i) \neq \bot$, sorted as $x_1 < \cdots < x_s$ and padded with $\bot$'s. Recall the unitary $\CO$, defined in Section~\ref{sec:TransitionMatrix} and which describes the evolution of the compressed oracle, and consider the corresponding  ``update function'' $\Upd : {\cal X} \times {\cal Y} \times {\frak{L}} \to {\cal X} \times {\cal Y} \times {\frak{L}}$, defined to satisfy 
\begin{align*}
\Upd\bigl(x,y,Enc(D)&\bigr) = \bigl(x,y,Enc(D')\bigr) \\
&\:\Longleftrightarrow\:
\ket{x}\ket{\hat y}\ket{\hat{D'}} = \CO\ket{x} \ket{\hat y} \ket{\hat D} = \ket{x} \ket{\hat y} \otimes \CO_{x \hat y}\ket{\hat D} 
\end{align*} 
for any $x \in \Xcal$, $y \in \Ycal$ and $D \in \DB$. 
By construction, and exploiting (\ref{eq:F-Step}), 
\ifnum\submission=0 it turns out that \fi
$\Upd$ is a rather simple function. Applied to $x \in \Xcal$, $y \in \Ycal$ and $L = [(x_1,y_1),\ldots,(x_s,y_s),\bot,\ldots,\bot] \in \frak{L}$, it acts as follows. If $y_i = 0$ for some $i$ then it acts as identity,%
\footnote{This is the ``artificial" case, which was introduced to have $\CO$ defined on the entire space $\CC[\Xcal] \otimes \CC[\Ycal] \otimes \CC[\DB]$} 
otherwise, the following two cases are distinguished: if $x \not\in \{x_1,\ldots,x_s\}$  and $y \neq 0$ then $\Upd$ inserts the pair $(x, y)$ to the list, while if $x = x_i$ and $y \neq y_i$ for some $i$ then $\Upd$ replaces $(x_i,y_i)$ by $(x_i,y_i \oplus y)$. In particular, for lists $L$ of bounded size $s \leq Q$, the classical function $\Upd$ can be efficiently computed, i.e., in time polynomial in $Q$ and in the size of the bit representations of the elements of $\Xcal$ and $\Ycal$. 

Formally, for a fixed $Q$, let $\DB_{\leq Q} := \{D \in \DB \,:\,|\{x \in \Xcal \,:\, D(x) \!=\! \bot\}| \leq Q\}$, and let $enc: \DB_{\leq Q} \to \frak{L}_{\leq Q} := (({\cal X} \times {\cal Y}) \cup \{\bot\})^Q$ be defined in the obvious way, i.e., so that $enc(D)$ is obtained from $Enc(D)$ by removing the rightmost $\bot$-paddings. Similarly, $\upd: {\cal X} \times {\cal Y} \times \frak{L}_{\leq Q}\to {\cal X} \times {\cal Y} \times \frak{L}_{\leq Q}$ is defined in the obvious way to coincide with $\Upd$ except for the shorter $\bot$-padding, and {\em except for the following additional modification}: $\upd$ is declared to act as identity on $(x, y, L)$ whenever $s = Q$ and $x \neq \{x_1,\ldots,x_s\}$, i.e., when there would be an ``overflow''. It then follows that $\upd$ is an {\em efficiently computable permutation}. Thus, by basic theory of quantum computation, the corresponding unitary $\ket{x,y,L} \mapsto \ket{\upd(x,y,L)}$ can be efficiently computed by means of a polynomial sized quantum circuit. Hence, by means of the encoding 
$$
{\sf \hat enc}: \ket{\hat D} \mapsto \ket{enc(D)} = \ket{x_1}\ket{y_1}\cdots\ket{x_s}\ket{y_s}\ket{\bot}\cdots\ket{\bot} \, ,
$$ 
the unitary $\CO$ can be efficiently computed, as long as it acts on $\CC[\Xcal] \otimes \CC[\Ycal] \otimes \CC[\DB_{< Q}]$, i.e., as long as fewer than $Q$ queries are being made. 
% Thus, like the classical lazy sampling technique, the compressed oracle can be appreciated as a way to simulate the QROM efficiently. 

Alternatively, we can also consider the following variant, where $\ket{D}$, rather than $\ket{\hat D}$, is encoded as $\ket{enc(D)}$: 
$$
{\sf enc}: \ket{D} \mapsto \ket{enc(D)} = \ket{x_1}\ket{y_1}\cdots\ket{x_s}\ket{y_s}\ket{\bot}\cdots\ket{\bot} 
$$ 
This encoding offer the following useful property. Consider a unitary $U_f$ on $\CC[\DB]$, plus an ancilla, that computes a classical function $f$, meaning that $U_f:\ket{D}\ket{w} \mapsto \ket{D}\ket{w \oplus f(D)}$, and for which the classical function $f$ is efficiently computable for $D \in \DB_{\leq Q}$ and given that $D$ is represented by $enc(D)$. Then, the unitary $U_f$ is efficiently computable with the considered encoding $\sf enc$. This allows for efficient extraction of useful information from the compressed oracle. 
Typical examples would be to check whether a certain preimage $x_\circ \in \Xcal$ is in the database, i.e., whether $D(x_\circ) \neq \bot$, or to check whether there is a $0$-preimage in the database, i.e. whether $\exists \, x: D(x) = 0$, etc. 

The final, simple yet crucial, observation is that one can efficiently switch between these two encodings. Indeed, it is easy to see that, say, ${\sf enc}$ commutes with applying the quantum Fourier transform $\qft$ in the obvious way, i.e., 
$$
{\sf enc} \ket{\hat D} = \ket{x_1}\ket{\hat y_1} \cdots \ket{x_s}\ket{\hat y_s}\ket{\bot}\cdots\ket{\bot} \, .
$$
Thus, ${\sf enc}$ equals ${\sf \hat enc}$ up to some $\qft$'s to be applied (controlled by the corresponding register not being $\bot$), which can be efficiently done. Hence, by a suitable encoding, both the evolution of the compressed oracle as well as efficiently computable classical functions on the (suitably encoded) database $D$, can be efficiently computed by a quantum circuit.

\section{PoSW Definition} \label{sec:posw-def}

%\paragraph{PoSW Definition.}  
%\serge{This eats up quite some space, and I expect the reader who is interested in this section to know what a PoSW is. Thus, this could go to the appendix to save space. }
Informally, a (non-interactive) PoSW allows a prover $\Prover$ to generate an efficiently verifiable proof showing that some computation was going on for $N$ sequential steps since some ``statement'' $\chi$ was received, in the sense that even a powerful adversary with parallel computation power cannot compute a valid proof with much less than $N$ steps. PoSW is typically constructed in the random oracle model. We recall its formal definition from~\cite{cohen2018simple} (after applying the Fiat-Shamir transformation) as follows (see Figure~\ref{fig:noninteractive_PoSW} for an illustration).

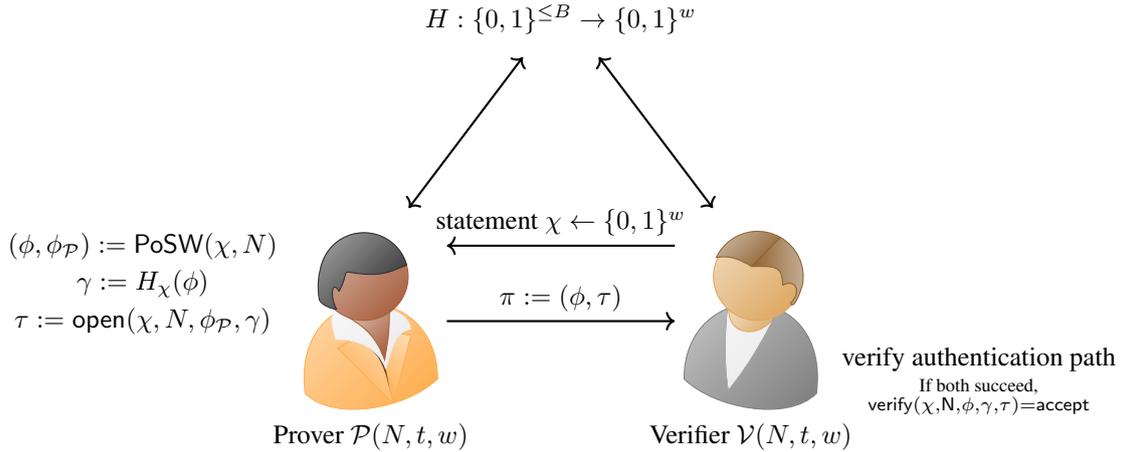
\begin{figure}[H]
    \centering
    \begin{tikzpicture}
    \node at (2.5,4) {$\Hash:\bool^{\leq B}\to\bool^w$};
    \node[alice,minimum size=5em] (P) at (0,0) {Prover $\mathcal{P}(N,t,w)$};
    \node[bob,mirrored,minimum size=5em] (V) at (5,0) {Verifier $\mathcal{V}(N,t,w)$};
    
    \draw[<->,thick] (0.5,1.5)--(2,3.5);
    \draw[<->,thick] (4.5,1.5)--(3,3.5);
    %\node at (2.5,2.5){$\Hash_\chi:=\Hash(\chi,\cdot)$};
    
    \draw[<-,thick] (1,1)-- node[above]{statement $\chi\leftarrow\bool^w$} (4,1) ;
    
    \node at (-3,1) {$(\phi,\phi_\mathcal{P}):=\mathsf{PoSW}(\chi,N)$};
    \node at (-3,0.5) {$\gamma := \Hash_\chi(\phi)$};
    \node at (-3,0) {$\tau:=\mathsf{open}(\chi,N,\phi_\mathcal{P},\gamma)$};
    \draw[->,thick] (1,0) -- node[above]{$\pi:=(\phi,\tau)$} (4,00);
    %\node at (8,0) {verify if $\gamma$ is from $\Hash_\chi(\phi)$};
    \node at (8,-0.5) {verify authentication path};
    \node at (8,-1){$\substack{\text{If both succeed,}\\ \mathsf{verify(\chi,N,\phi,\gamma,\tau)}=\mathsf{accept}}$};
    \end{tikzpicture}
    \caption{Non-interactive PoSW.}
    \label{fig:noninteractive_PoSW}
\end{figure}

\begin{itemize}
    \item \emph{Common Inputs:} The prover $\Prover$ and the verifier $\mathcal{V}$ get as common input two statistical security parameters $w,t\in \mathbb{N}$ and a time parameter $N \in \mathbb{N}$. They have access to a random oracle $\Hash:\bool^{\leq B}\to\bool^w$, where  $B$ is sufficiently large but otherwise arbitrary.%
\footnote{The original paper~\cite{cohen2018simple} considers $\Xcal=\bool^*$; however, we want $\cal X$ to be finite so that our results from the previous sections apply. Thus, we simply choose $B$ large enough, so that the scheme is well defined, but also larger than any query that an arbitrary but fixed attacker will make. }
    \item \emph{Statement:} $\mathcal{V}$ samples a random $\chi \leftarrow \bool^w$ and sends it to $\Prover$.
    \item \emph{Compute PoSW:} $\Prover$ computes  $(\phi,\phi_{\Prover}) := \mathsf{PoSW}^H(\chi,N)$, where $\phi$ is a proof and $\phi_{\Prover}$ is a state $\Prover$ uses to compute the opening. 
    \item \emph{Opening Challenge:} The opening challenge $\gamma$ is determined by $\gamma := \Hash(\chi,\phi) \in \bool^w$.
    \item \emph{Open:} $\Prover$ computes  $\tau := \mathsf{open}^H(\chi,N,\phi_{\Prover}, \gamma)$. $\Prover$ sends $\pi := (\phi, \tau)$ to $\mathcal{V}$.
    \item \emph{Verify:} $\mathcal{V}$ computes and outputs $\mathsf{verify(\chi,N,\phi,\gamma,\tau)} \in \{ \mathsf{accept}, \mathsf{reject}\}$.
\end{itemize}

Since our goal is to analyze post-quantum security of Simple PoSW~\cite{cohen2018simple}, we will not present the formal security properties for PoSW here. Instead, we will prove concrete upper bounds on the probability that a $k$-parallel $q$-query quantum oracle algorithm $\mathcal{A}$ with $q < N$ can generate a valid proof.

\section{The Extraction Algorithm}\label{app:Extract}

\begin{algorithm}[H]
\SetAlgoLined
\textbf{Input: }$\lab\rt\in\bool^w$\\
\textbf{Output: } a subtree $T\subseteq V_n$\\
\textbf{Initialize: }\\
Set $\lext: \VId{n} \to \bool^w\cup\{\perp\}$ with $\lext_\rt\leftarrow\lab\rt$ and $\lext_{v} \leftarrow \perp$
for all $v\in V_n\setminus\{\rt\}$;\\
Set all vertex $v\in V_n$ as unmarked;\\
\textbf{Notation: } Define the support of a labeling as $\supp(\lext):=\{v\in V_n:\lext_v\neq\perp\}$\\
\textbf{Labeling extraction: }\\
\While{there is an unmarked $v\in\supp(\lext)\setminus\leaves(V_n)$}{
    mark the vertex $v$;\\
    \If{there exists some $x, y\in \bool^w$ such that $\lext_{v} = D(v, x, y)$}{
        $\lext_{\lch{v}} \leftarrow x$;\\
        $\lext_{\rch{v}} \leftarrow y$;
    }
}

\textbf{Consistency check:}\\
$T\leftarrow \supp(\lext)$;\\
\For{$v\in\leaves(T)$}{
    \If{$\lext_v\neq D(v,\lext_{\parent(v)})$}{
        $T\leftarrow T\setminus\{v\}$;
    }
}
output $T$;
\caption{$\ext{D}{n}{\lab\rt}$}\label{algo:extract}
\end{algorithm}

\end{document}